\documentclass[journal,onecolumn]{IEEEtran}

\ifCLASSINFOpdf
\else
\fi

\usepackage{caption}
\usepackage{cite}
\usepackage{multirow}
\usepackage{setspace}
\usepackage{graphicx}
\usepackage{graphics}
\usepackage{subfigure}
\usepackage{url}
\usepackage{verbatim} 
\usepackage{multicol}
\usepackage{array}
\usepackage{colortbl}
\usepackage{psfrag}
\usepackage{xcolor}
\usepackage{enumerate}
\usepackage{amsthm}
\usepackage{amsmath, amssymb}
\usepackage{amsfonts}

\newcommand{\Te}{T_{\mrm{eff}}}

\newcommand{\Fq}{\mathbb{F}_q}

\newcounter{actr}
{\begin{list}{(\alph{actr})}{\usecounter{actr}}}{\end{list}}

\newcounter{ictr}
{\begin{list}{(\roman{ictr})}{\usecounter{ictr}}}{\end{list}}

 {\everymath{\scriptscriptstyle\everymath{}}\array}%
 {\endarray}

\newcolumntype{C}[1]{>{\centering\let\newline\\\arraybackslash\hspace{0pt}}m{#1}}
\newcolumntype{L}[1]{>{\raggedright\let\newline\\\arraybackslash\hspace{0pt}}m{#1}}
\newcolumntype{R}[1]{>{\raggedleft\let\newline\\\arraybackslash\hspace{0pt}}m{#1}}

\newtheorem{remark}{Remark}
\newtheorem{thm}{Theorem}
\newtheorem{lemma}{Lemma}

\newtheorem{corol}{Corollary}
\newtheorem{prop}{Proposition}
\newtheorem{defn}{Definition}
\newtheorem{fact}{Fact}

\newenvironment{new-proof}[1]
{{\em Proof }:\\}%
{ \noindent\qed }
\newcommand{\defeq}{\stackrel{\Delta}{=}}

\newcommand{\mrm}{\mathrm}

\newcommand{\bc}{{\mathbf{c}}}

\newcommand{\cC}{{\mathcal{C}}}

\newcommand{\cE}{{\mathcal{E}}}

\newcommand{\bG}{{\mathbf{G}}}

\newcommand{\bH}{{\mathbf{H}}}

\newcommand{\bI}{{\mathbf{I}}}

\newcommand{\cO}{{\mathcal{O}}}

\newcommand{\bp}{{\mathbf{p}}}

\newcommand{\bq}{{\mathbf{q}}}

\newcommand{\bQ}{{\mathbf{Q}}}

\newcommand{\bs}{{\mathbf{s}}}

\newcommand{\cS}{{\mathcal{S}}}

\newcommand{\bu}{{\mathbf{u}}}
\newcommand{\bU}{{\mathbf{U}}}

\newcommand{\bv}{{\mathbf{v}}}
\newcommand{\bV}{{\mathbf{V}}}

\newcommand{\bw}{{\mathbf{w}}}

\newcommand{\cW}{{\mathcal{W}}}

\newcommand{\bx}{{\mathbf{x}}}

\newcommand{\bX}{{\mathbf{X}}}
\newcommand{\cX}{{\mathcal{X}}}

\newcommand{\by}{{\mathbf{y}}}

\newcommand{\bY}{{\mathbf{Y}}}

\newcommand{\al}{\alpha}

\newcommand{\eps}{\varepsilon}

\DeclareMathAlphabet{\mathbsf}{OT1}{cmss}{bx}{n}
\DeclareMathAlphabet{\mathssf}{OT1}{cmss}{m}{sl}

\DeclareSymbolFont{bsfletters}{OT1}{cmss}{bx}{n}
\DeclareSymbolFont{ssfletters}{OT1}{cmss}{m}{n}
\DeclareMathSymbol{\bsfGamma}{0}{bsfletters}{'000}
\DeclareMathSymbol{\ssfGamma}{0}{ssfletters}{'000}
\DeclareMathSymbol{\bsfDelta}{0}{bsfletters}{'001}
\DeclareMathSymbol{\ssfDelta}{0}{ssfletters}{'001}
\DeclareMathSymbol{\bsfTheta}{0}{bsfletters}{'002}
\DeclareMathSymbol{\ssfTheta}{0}{ssfletters}{'002}
\DeclareMathSymbol{\bsfLambda}{0}{bsfletters}{'003}
\DeclareMathSymbol{\ssfLambda}{0}{ssfletters}{'003}
\DeclareMathSymbol{\bsfXi}{0}{bsfletters}{'004}
\DeclareMathSymbol{\ssfXi}{0}{ssfletters}{'004}
\DeclareMathSymbol{\bsfPi}{0}{bsfletters}{'005}
\DeclareMathSymbol{\ssfPi}{0}{ssfletters}{'005}
\DeclareMathSymbol{\bsfSigma}{0}{bsfletters}{'006}
\DeclareMathSymbol{\ssfSigma}{0}{ssfletters}{'006}
\DeclareMathSymbol{\bsfUpsilon}{0}{bsfletters}{'007}
\DeclareMathSymbol{\ssfUpsilon}{0}{ssfletters}{'007}
\DeclareMathSymbol{\bsfPhi}{0}{bsfletters}{'010}
\DeclareMathSymbol{\ssfPhi}{0}{ssfletters}{'010}
\DeclareMathSymbol{\bsfPsi}{0}{bsfletters}{'011}
\DeclareMathSymbol{\ssfPsi}{0}{ssfletters}{'011}
\DeclareMathSymbol{\bsfOmega}{0}{bsfletters}{'012}
\DeclareMathSymbol{\ssfOmega}{0}{ssfletters}{'012}

\renewcommand{\defeq}{\triangleq}

\newcommand{\rvbs}{{\mathbsf{s}}}

\newcommand{\rvbx}{{\mathbsf{x}}}

\newcommand{\braceup}[4]{\draw[decorate, decoration={brace, amplitude=5pt},thick] ([xshift=0.5mm,yshift=#3]#1.north west)--([xshift=-0.5mm,yshift=#3]#2.north east) node[midway,anchor=south,outer sep=2mm] {#4}}
\newcommand{\bracedn}[4]{\draw[decorate, decoration={brace, amplitude=5pt},thick] ([xshift=-0.5mm,yshift=#3]#2.south east)--([xshift=0.5mm,yshift=#3]#1.south west) node[midway,anchor=north,outer sep=2mm] {#4}}
\newcommand{\dimup}[4]{\draw[<->,thick] ([xshift=0.5mm,yshift=#3]#1.north west)--([xshift=-0.5mm,yshift=#3]#2.north east) node[midway,anchor=south] {#4}}
\newcommand{\dimdn}[4]{\draw[<->,thick] ([xshift=-0.5mm,yshift=#3]#2.south east)--([xshift=0.5mm,yshift=#3]#1.south west) node[midway,anchor=north] {#4}}
\definecolor{light-gray}{gray}{0.75}
\newcolumntype{g}{>{\columncolor{light-gray}}c}

\usepackage{tikz}
\usetikzlibrary{shapes,arrows,shadows,positioning,decorations.pathreplacing,trees,calc,patterns}
\tikzstyle{nosym} = [rectangle, font=\large, minimum width=4mm, minimum height=4mm, text centered]
\tikzstyle{sym} = [draw, thick, rectangle, font=\large, minimum width=4mm, minimum height=4mm, text centered]
\tikzstyle{esym} = [sym, fill=light-gray]
\tikzstyle{usym} = [sym, fill=white]
\tikzstyle{diagbox} = [draw, rectangle, font=\footnotesize, fill=white, text centered, rounded corners]
\tikzstyle{codebox} = [draw, rectangle, font=\footnotesize, minimum height=7mm, fill=white, text centered]

\tikzstyle{triple2} = [rectangle split, anchor=text,rectangle split parts=4]
\tikzstyle{double2} = [rectangle split, anchor=text,rectangle split parts=2]
\tikzstyle{single2} = [rectangle split, anchor=text,rectangle split parts=1]
\tikzstyle{triple} = [draw, rectangle split,rectangle split parts=4]
\tikzstyle{double} = [draw, rectangle split,rectangle split parts=2]
\tikzstyle{single} = [draw, rectangle split,rectangle split parts=1]
\tikzset{block/.style={rectangle,draw}}
\tikzset{block2/.style={rectangle}}

\usetikzlibrary{calc}
\usetikzlibrary{decorations.pathreplacing}

\pgfdeclarepatternformonly{north east lines wide}{\pgfqpoint{-1pt}{-1pt}}{\pgfqpoint{4pt}{4pt}}{\pgfqpoint{3pt}{3pt}}%
{
  \pgfsetlinewidth{0.7pt}
  \pgfpathmoveto{\pgfqpoint{0pt}{0pt}}
  \pgfpathlineto{\pgfqpoint{9.1pt}{9.1pt}}
  \pgfusepath{stroke}
}

\tikzset{%
    body/.style={inner sep=0pt,outer sep=0pt,shape=rectangle,draw,thick,pattern=north east lines wide},
    dimen/.style={<->,>=latex,thin,every rectangle node/.style={fill=white,midway,font=\sffamily}},
    symmetry/.style={dashed,thin},
}

\pgfdeclarepatternformonly{soft horizontal lines}{\pgfpointorigin}{\pgfqpoint{100pt}{1pt}}{\pgfqpoint{100pt}{2pt}}%
{
  \pgfsetstrokeopacity{0.3}
  \pgfsetlinewidth{0.5pt}
  \pgfpathmoveto{\pgfqpoint{0pt}{0.5pt}}
  \pgfpathlineto{\pgfqpoint{100pt}{0.5pt}}
  \pgfusepath{stroke}
}

\begin{document}
\title{Layered Constructions for Low-Delay \\ Streaming Codes}


\author{Ahmed~Badr,~\IEEEmembership{Student Member,~IEEE},
 Pratik Patil, 
 Ashish~Khisti,~\IEEEmembership{Member,~IEEE},
 Wai-Tian~Tan,~\IEEEmembership{Member,~IEEE} and John Apostolopoulos,~\IEEEmembership{Fellow,~IEEE}
\thanks{A.~Badr, P.~Patil and A.~Khisti are with University of Toronto, Toronto, ON, Canada. W.~Tan and J.~Apostolopoulos were with Hewlett Packard Laboratories, USA when this work was done. They are now with Cisco Systems, USA. The corresponding author is Ashish Khisti (ashish.khisti@gmail.com).}
\thanks{This work was supported by an Ontario Early Researcher Award, the Canada Research Chair program and by Hewlett Packard through a HP-IRP Award. }
\thanks{Part of this work was presented at the INFOCOM, Turin, Italy, 2013, CWIT, Toronto, ON, Canada, 2013, ISIT, Istanbul, Turkey, 2013 and the Asilomar Conference on Signals, Systems, and Computers, Pacific Grove, CA, 2013.}}

\maketitle

\begin{abstract}
We propose a new class of error correction codes for low-delay streaming communication. We consider an online setup where a source packet arrives at the encoder every $M$ channel uses, and needs to be decoded with a maximum delay of $T$ packets. We consider a sliding-window erasure channel --- $\cC(N,B,W)$ --- which introduces either up to $N$ erasures in arbitrary positions, or $B$ erasures in a single burst, in any window of length $W$. When $M=1$, the case where source-arrival and channel-transmission rates are equal, we propose a class of codes --- MiDAS codes --- that achieve a near optimal rate. Our construction is based on a {\em layered} approach. We first construct an optimal code for the $\cC(N=1,B,W)$ channel, and then concatenate an additional layer of parity-check symbols to deal with $N>1$. When $M > 1$, the case where source-arrival and channel-transmission rates are unequal, we characterize the capacity when $N=1$ and $W \ge M(T+1),$ and for $N>1$, we propose a construction based on a layered approach. Numerical simulations over Gilbert-Elliott and Fritchman channel models indicate significant gains in the residual loss probability over baseline schemes. We also discuss the connection between the error correction properties of the MiDAS codes and their underlying column distance and column span.
\end{abstract}


\IEEEpeerreviewmaketitle

\section{Introduction}
\label{sec:intro}

\IEEEPARstart{M}{any} multimedia applications such as interactive audio/video conferencing, mobile gaming and cloud-computing require error correction of streaming sources under strict latency constraints. The transmitter must encode a source stream sequentially, and the receiver must decode each source packet within a fixed playback deadline. 
Classical error correction codes such as Maximum Distance Separable (MDS) and rateless codes are far from ideal in such situations. Their encoders operate on a source stream in blocks and introduce buffering delays, whereas the decoders can only recover missing source packets simultaneously without considering the deadline of each packet.
Naturally both the structure of optimal codes and the associated fundamental limits are expected to be different in the streaming scenario. For example, it is well known that the Shannon capacity of an erasure channel only depends on the fraction of erasures. However when delay constraints are imposed, the actual pattern of packet losses also becomes relevant. The decoding delay over channels with burst losses can be very different than isolated or random losses. In practice, channels introduce both burst and isolated losses often captured by statistical models such as the Gilbert-Elliott (GE) channel. Therefore the underlying codes must simultaneously correct both types of patterns. The central question we address in this paper is how to construct near optimal \emph{streaming codes} for such channels.

Problems involving real-time coding and compression have been studied from many different perspectives in related literature. Some structural properties of optimal codes have been studied in e.g.,~\cite{Witsenhausen:79,Teneketzis,jscc-rt}, and a dynamic programming based formulation is proposed. Schulman~\cite{Schulman96} and Sahai~\cite{sahaiThesis} study coding techniques based on {\em tree codes} in a streaming setup with discrete memoryless channels. Sukhavasi and Hassibi~\cite{sukhavai-11} have proposed linear time-invariant tree codes for the class of i.i.d.\ erasure channels, which are attractive due to low decoding complexity. For the class of burst erasure channels, Martinian et.\ al.\ \cite{MartinianS04,MartinianT07} study optimal streaming codes assuming an upper-bound on the erasure burst length. Unfortunately these constructions are not robust in the presence of random or isolated erasures. Some specific examples of robust codes are presented in~\cite{MartinianS04} using a computer search, but these examples provide limited insights towards a general construction.

In the present paper we initiate a systematic investigation of streaming codes over channels which introduce both burst and isolated erasures. As discussed before, in many practical channel models, both these erasure patterns are relevant. We introduce a class of deterministic channels that could be interpreted as an approximation of the GE channel. We show that in such models there exists an inherent tradeoff between the burst error correction and isolated error correction capabilities of streaming codes i.e., given a fixed rate and delay, a streaming code that can correct a long burst cannot recover from many isolated erasures and vice versa. We also present a new class of streaming codes that are near optimal with respect to this tradeoff. Our constructions are based on a layered approach. We first construct an optimal burst erasure code and then concatenate additional parity-check symbols to enable recovery from isolated erasures. In practice, our proposed constructions can be easily adapted if the number of isolated erasures to be corrected varies based on channel conditions. 

In the first part of the paper, we consider the case when the source-arrival and channel-transmission rates are equal. We propose a class of codes --- MiDAS codes --- that attain a near optimal tradeoff between the burst error correction and isolated error correction properties as discussed above. In the second part of  the paper, we consider the case when the source-arrival and channel-transmission rates are unequal. Each source symbol may arrive once every $M$ channel uses, where $M$ is a parameter determined by the application. For example in high definition video streaming, each source frame may arrive every 40 ms, whereas each channel packet may be transmitted every millisecond, resulting in $M=40$. We note that the first part of the paper corresponds to $M=1$. We show that a straightforward extension of these codes to $M>1$ is suboptimal, and propose an optimal construction for the burst erasure channel, as well as its robust extension. Finally we present extensive simulation results over the GE and Fritchman channels that indicate substantial performance gains over baseline codes for a wide range of channel parameters.

In other related works, references~\cite{khistiSingh:09,de-sco,mu-sco} study an extension of streaming codes over a burst erasure channel that adapt the delay based on the burst length. When the burst length is small, the decoding delay is smaller, when it is long, the corresponding delay is also longer. References~\cite{LuiCWIT,LuiMASc} study an extension of streaming codes to parallel channels with burst erasures whereas~\cite{liKG:11} studies streaming codes that can correct multiple bursts. In a parallel work, references~\cite{tekin,leong,leong2} also study streaming codes motivated by connections between streaming and unicast network coding. However to the best of our knowledge, these papers do not consider channels with {\em both} burst and isolated erasures, nor consider the layered approach for coding, which is the focus of this work. There is also a significant body of literature on adapting various coding techniques for streaming systems, see e.g.,~\cite{streaming-1, streaming-2,
 streaming-4, streaming-9} and references therein.

\section{System Model and Main Results}
\label{sec:system-model}

\begin{center}
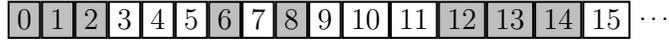
\begin{figure}[t]
\centering
	\centering
	\resizebox{0.5\columnwidth}{!}{
	\begin{tikzpicture}[node distance=0mm]
		\node[esym]  (x100) {$0$};
		\node[esym, right = of x100]     (x101) {$1$};
		\node[esym, right = of x101]     (x102) {$2$};
		\node[usym, right = of x102]     (x103) {$3$};
		\node[usym, right = of x103]     (x104) {$4$};
		\node[usym, right = of x104]     (x105) {$5$};
		\node[esym, right = of x105]     (x106) {$6$};
		\node[usym, right = of x106]     (x107) {$7$};
		\node[esym, right = of x107]     (x108) {$8$};
		\node[usym, right = of x108]     (x109) {$9$};
		\node[usym, right = of x109]     (x110) {$10$};
		\node[usym, right = of x110]     (x111) {$11$};
		\node[esym, right = of x111]     (x112) {$12$};
		\node[esym, right = of x112]     (x113) {$13$};
		\node[esym, right = of x113]     (x114) {$14$};
		\node[usym, right = of x114]     (x115) {$15$};
		\node      [right = of x115]     (x1end) {$\cdots$};		
	\end{tikzpicture}}
	\caption{An example of the  deterministic channel $\cC(N=2,B=3,W=5)$.  In any sliding window of length $W=5,$ there is either a single erasure burst of length no greater than $B=3,$ or no more than $N=2$ isolated erasures.}
	\label{fig:chan-mix}
\end{figure}

\vspace{-1em}
\end{center}

We consider a class of packet erasure channels where the erasure patterns are locally constrained. In any sliding window of length $W$, the channel can introduce one of the following patterns,
(i) a single erasure burst of maximum length $B$, 
or (ii) a maximum of $N$ erasures in arbitrary locations.
We use the notation $\cC(N,B,W)$ for such a channel. Fig.~\ref{fig:chan-mix} provides an example of such a channel when $N=2$, $B=3$ and $W=5$. Note that the condition ${N \le B}$ follows since a burst erasure is a special type of erasure pattern. We will assume throughout the paper that $B+1 \le W$, so that in any window of length $W$ there is at-least one non-erased symbol\footnote{If this condition is violated, it can be easily shown that the capacity is trivially zero.}. Note that the special case when $N=1$ reduces the above model to a burst-only channel model. In this case the guard separation between successive bursts is at-least ${W-1}$. 

In practice we can view $\cC(N,B,W)$ as an approximation of statistical models such as the Gilbert-Elliott (GE) channel model. A GE channel is in one of two states. In the good state, it behaves an an i.i.d.\ erasure channel, while in the bad state, it behaves as a burst-erasure channel. Thus the interval consisting of a burst loss corresponds to the bad state, whereas a window comprising of isolated erasures corresponds to the good state. The values of $N$, $B$ as well as $W$ will depend on the underlying channel parameters. Some insights into the connections between the deterministic and statistical models will be presented in our simulation results.


We further separately treat the cases when the source-arrival and channel-transmission rates are equal and when they are not.

\subsection{Equal Source-Channel Rates}
\label{subsec:eq-rates}
At each time-slot ${i \ge 0}$, the encoder observes a source symbol $\bs[i]$, and transmits a channel symbol $\bx[i]$.
We will assume throughout that the source and channel alphabets are $\cS = \mathbb{F}_q^k$ and $\cX = \mathbb{F}_q^n$ respectively, where $q$ denotes the size of the base field. The rate of the code equals $R = \frac{k}{n}$. The channel input at time $i$ can depend causally on all the source symbols observed up to and including time $i$, but not on any future symbols i.e., $\bx[i] = f_i(\bs[0],\ldots, \bs[i])$.
The channel output at time $i$ is denoted by the symbol $\by[i]$. Note that for the erasure channel we have either $\by[i] =\bx[i]$, or $\by[i] = \star$, when the channel introduces an erasure. The decoder is required to reconstruct each source symbol with a delay of $T$, i.e., for each $i\ge 0$ we must have a decoding function,
$\bs[i] = g_i(\by[0], \ldots, \by[i+T]).$ Such a collection of encoding and decoding functions constitute a streaming code.

\begin{defn}[Streaming Capacity - Equal Source-Channel Rates]
A rate $R$ is achievable with a delay of $T$ over the $\cC(N,B,W)$ channel, if there exists a streaming code of this rate over some field-size $q$ such that every source symbol $\bs[i]$ can be decoded at the destination with a delay of $T$ symbols. The largest achievable rate is the capacity.
\end{defn}

We establish the following upper and lower bounds on the capacity.

\begin{thm}
\label{thm:Chan-1-UB}
Any achievable rate for $\cC(N,B,W)$ and delay $T \ge B$ must satisfy
\begin{align}
\left(\frac{R}{1-R}\right)B + N \le \Te+1,
\label{eq:r-ub}
\end{align}
where $\Te \defeq \min(T,W-1)$. 
\hfill$\Box$

\end{thm}
To interpret~\eqref{eq:r-ub}, note that $\Te+1 = \min(T+1,W)$, and ${T+1}$ denotes the {\em active} duration of each source packet i.e., each source packet $\bs[i]$ arrives at time $t=i$ and must be decoded by time ${t=i+T}$. When $W > T+1$ the upper bound in~\eqref{eq:r-ub} is governed by the decoding delay, otherwise it depends on $W$.
One can also interpret~\eqref{eq:r-ub} as follows. When the rate $R$ and delay $T$ are fixed, there exists a tradeoff between the achievable values of $B$ and $N$. We cannot have a streaming code that can simultaneously correct long erasure bursts and many isolated erasures. 

\begin{remark}
By definition, the streaming code is a convolutional code where the source stream $\bs[i]$ is the input and $\bx[i]$ is the output. In Appendix~\ref{app:distance-span}, we establish that a key property for any feasible code for the $\cC(N,B,W)$ channel is that it should simultaneously have a certain {\em column distance} and {\em column span}. Thus Theorem~\ref{thm:Chan-1-UB} also leads to a fundamental tradeoff between the achievable column distance and column span of any convolutional code (cf.~Prop.~\ref{prop:cTdT-tradeoff} in Appendix~\ref{app:distance-span}).
\end{remark}

The proof of Theorem~\ref{thm:Chan-1-UB} is provided in Section~\ref{subsec:match-ch1-converse}.
We further propose a class of streaming codes, {\em Maximum Distance And Span tradeoff (MiDAS) codes}, that achieve a near-optimal tradeoff.

\begin{thm}
\label{thm:midas}
For any channel $\cC(N,B,W)$ and delay $T \ge B$, there exists a code of rate $R$ that satisfies
\begin{align}
\left(\frac{R}{1-R}\right)B + N > \Te,
\label{b-achiev}
\end{align}
where $\Te \defeq \min(T,W-1)$. 
$\hfill\Box$
\end{thm}
The proof of Theorem~\ref{thm:midas} is presented in Section~\ref{subsec:midas}. 
As will be apparent, our construction is based on a layered approach. We first construct an optimal streaming code for $\cC(N'=1, B, W)$ channel. Then we append an additional layer of parity-check sub-symbols that enables us to correct $N$ erasures in any sliding window of length $W$. By directly comparing~\eqref{b-achiev} and~\eqref{eq:r-ub}, we see that the proposed codes are near-optimal. We also remark that the upper bound~\eqref{eq:r-ub} establishes that some of the $R=1/2$ codes found via a computer search in~\cite[Section V-B]{MartinianS04} are indeed optimal.

Note that Theorem~\ref{thm:midas} does not explicitly state the field-size $q$. 
The underlying constructions are based on Strongly-MDS codes~\cite{strongly-mds-2, strongly-mds} which are known to exist for field-sizes
that increase exponentially in $\Te$. However we also provide an alternate construction in Section~\ref{subsec:midas-field}, that attains~\eqref{b-achiev}, and whose field-size increases as $\cO(\Te^3)$.

\begin{center}
\begin{figure}[t]
	\centering
	\resizebox{0.7\columnwidth}{!}{
	\begin{tikzpicture}[node distance=0mm]
		\node[usym]  (x100) {$$};
		\node[nosym, above = 1em of x100]     (x200) {$\bs[0]$};
		\node[usym, right = of x100]     (x101) {$$};
		\node[usym, right = of x101]     (x102) {$$};
		\node[usym, right = of x102]     (x103) {$$};
		\node[usym, right = of x103]     (x104) {$$};
		\node[usym, right = of x104]     (x105) {$$};
		\node[nosym, above = 1em of x105]     (x205) {$\bs[1]$};
		\node[usym, right = of x105]     (x106) {$$};
		\node[usym, right = of x106]     (x107) {$$};
		\node[usym, right = of x107]     (x108) {$$};
		\node[usym, right = of x108]     (x109) {$$};
		\node[usym, right = of x109]     (x110) {$$};
		\node[nosym, above = 1em of x110]     (x210) {$\bs[2]$};
		\node[usym, right = of x110]     (x111) {$$};
		\node[usym, right = of x111]     (x112) {$$};
		\node[usym, right = of x112]     (x113) {$$};
		\node[usym, right = of x113]     (x114) {$$};
		\node[usym, right = of x114]     (x115) {$$};
		\node[usym, right = of x115]     (x116) {$$};
		\node[usym, right = of x116]     (x117) {$$};
		\node[usym, right = of x117]     (x118) {$$};
		\node[usym, right = of x118]     (x119) {$$};
		\node[usym, right = of x119]     (x120) {$$};
		\node[usym, right = of x120]     (x121) {$$};
		\node[usym, right = of x121]     (x122) {$$};
		\node[usym, right = of x122]     (x123) {$$};
		\node[usym, right = of x123]     (x124) {$$};
		\node[usym, right = of x124]     (x125) {$$};
		\node[nosym, above = 1em of x125]     (x225) {$\bs[T]$};
		\node[usym, right = of x125]     (x126) {$$};
		\node[usym, right = of x126]     (x127) {$$};
		\node[usym, right = of x127]     (x128) {$$};
		\node[usym, right = of x128]     (x129) {$$};
		\node[nosym, below = 2em of x129]     (x229) {\begin{tabular}{c} Recover\\ $\bs[0]$ \end{tabular}};
		\node[usym, right = of x129]     (x130) {$$};
		\node[nosym, above = 1em of x130]     (x230) {$\bs[T+1]$};
		\node[usym, right = of x130]     (x131) {$$};
		\node[usym, right = of x131]     (x132) {$$};
		\node[usym, right = of x132]     (x133) {$$};
		\node[usym, right = of x133]     (x134) {$$};
		\node[nosym, below = 2em of x134]     (x234) {\begin{tabular}{c} Recover\\ $\bs[1]$ \end{tabular}};
		\node      [right = of x134]     (x1end) {$\cdots$};
		\draw[-latex] (x200.south) -| (x100.north);
		\bracedn{x100}{x104}{-1mm}{\footnotesize{$\bX[0,:]$}};
		\draw[-latex] (x205.south) -| (x105.north);
		\bracedn{x105}{x109}{-1mm}{\footnotesize{$\bX[1,:]$}};
		\draw[-latex] (x210.south) -| (x110.north);
		\bracedn{x110}{x114}{-1mm}{\footnotesize{$\bX[2,:]$}};
		\draw[-latex] (x225.south) -| (x125.north);
		\bracedn{x125}{x129}{-1mm}{\footnotesize{$\bX[T,:]$}};
		\draw[-latex] (x230.south) -| (x130.north);
		\draw[-latex] (x129.south) -| (x229.north);
		\bracedn{x130}{x134}{-1mm}{\footnotesize{$\bX[T+1,:]$}};
		\draw[-latex] (x134.south) -| (x234.north);
	\end{tikzpicture}}
	\caption{Each source symbol $\bs[i]$ arrives just before the transmission of $\bX[i,:]$ and needs to be reconstructed at the destination after a delay of $T$ macro-packets.}
	\label{fig:mismatch}
\end{figure}
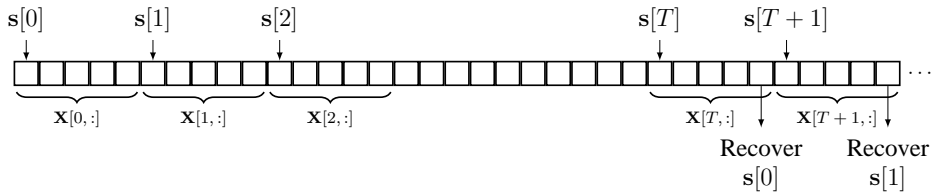
\end{center}

\subsection{Unequal Source-Channel Rates}
\label{subsec:uneq-rates}
We discuss a generalization of the setup in Section~\ref{subsec:eq-rates} where one source symbol arrives every $M$ channel uses. The alphabets of source and channel symbols are as before. For convenience, the collection of $M$ channel symbols is termed as a macro-packet. The index of each macro-packet is denoted using the letter $i$ i.e.,
\begin{align}
\bX[i,:] = \left[\bx[i,1] ~|~ \dots~|~ \bx[i,M]\right] \in \mathbb{F}_q^{n \times M}
\label{eq:macro-packet}
\end{align}
denotes the macro-packet $i$ consisting of $M$ channel symbols. At the start of macro-packet $i$, the encoder observes the source symbol $\bs[i] \in {\mathbb F}_q^k$ and generates $M$ symbols $\bx[i,j] \in {\mathbb F}_q^n$, for $j \in \{1,\dots,M\}$ which can depend on all the observed source packets up to that time i.e.,
\begin{equation}
\bx[i,j] = f_{i,j}(\bs[0], \bs[1], \cdots, \bs[i]).
\end{equation} 
These symbols are transmitted in the $M$ time-slots corresponding to the macro-packet $i$. Fig.~\ref{fig:mismatch} shows the system model. Note that for the case when $M=1$ the setup reduces to that in Section~\ref{subsec:eq-rates}.

The $j$th channel output symbol in the macro-packet $i$ is denoted by $\by[i,j]$. When the channel input is not erased, we have, $\by[i,j]=\bx[i,j]$, whereas when the channel input is erased, $\by[i,j]=\star$. The channel output macro-packets are expressed as $\bY[i,:] = \left[\by[i,1] ~|~ \dots ~|~ \by[i,M] \right]$. The decoder is required to decode each source symbol with a maximum delay of $T$ macro-packets i.e.,
\begin{equation}
 \bs[i] = g_i(\bY[0,:], \bY[1,:], \cdots, \bY[i+T,:]).
\end{equation}


We note that the rate of this code is given by $R = \frac{k}{Mn}$. In this definition, we are normalizing the rate with the size of each macro-packet. The rate expression is motivated by the fact that $nM$ channel symbols are transmitted over the channel for each $k$ source symbols.

\begin{defn}[Streaming Capacity - Unequal Source-Channel Rates]
A rate $R$ is achievable with a delay of $T$ macro-packets over $\cC(N,B,W)$ if there exists a streaming code of this rate over some field-size $q$ such that every source symbol $\bs[i]$ can be decoded with a delay of $T$ macro-packets. The largest achievable rate is the capacity.
\end{defn}
For the above setup, the capacity has been obtained when $N=1$ and $W \ge M(T+1)$.
\begin{thm}
\label{thm:capacity}
For the channel $\cC(N=1, B, W)$, and any $M$ and delay $T$, such that $W \ge M(T+1)$, the capacity $C$ is expressed as follows,
\begin{align}
\label{eq:capacity}
C = \left\{ \begin{array}{ll} 
\frac{T}{T+b}, & B' \leq \frac{b}{T+b}M,~T \geq b, \\ 
\frac{M(T+b+1)-B}{M(T+b+1)}, & B' > \frac{b}{T+b}M,~T > b,\\ 
\frac{M-B'}{M}, & B' > \frac{M}{2},~T = b,\\
0, & T < b,
\end{array}\right.
\end{align}
where the constants $b$ and $B'$ be defined via
\begin{align}B =bM+B', \quad B' \in \{0,\dots,M-1\}, b \in \mathbb{N}^0,\label{eq:B-split} \end{align}
$\hfill\Box$
\end{thm}

The proof of Theorem~\ref{thm:capacity} is divided into two main parts. The code construction is illustrated in section~\ref{subsec:mismatch-construction} while the converse appears in section~\ref{subsec:mismatch-converse}.

The constructions in Theorem~\ref{thm:capacity} only apply to the burst-erasure channel. Based on the layered approach in Theorem~\ref{thm:midas} we also propose a robust construction for the case when $N>1$ in Section~\ref{subsec:robust-ext}. However the optimal construction is left for future work.

\section{Performance Analysis of Baseline Schemes}
\label{sec:background}

We review two constructions --- Strongly-MDS codes and Maximally Short codes --- that have been proposed in earlier works. While the achievable rate of these codes in the present setup can be far from optimal, they constitute important building blocks in our proposed constructions.


\subsection{Strongly-MDS Codes}
\label{subsec:strongly-mds}

Classical erasure codes are designed for maximizing the underlying distance properties. Roughly speaking, such codes will recover all the missing source symbols simultaneously once sufficiently many parity-check symbols have been received at the decoder. Indeed a commonly used family of such codes, {\em random-linear codes}, see e.g.,~\cite{ho-med, roundrobin}, are designed to guarantee that the underlying system of equations is full rank with high probability. We discuss one particular class of {\em deterministic code} constructions with optimal distance properties~\cite{strongly-mds-2, strongly-mds} in this section.

Consider a $(\bar{n},\bar{k},\bar{m})$ convolutional code that maps an input source stream $\bs[i] = (s_0[i],\dots,s_{\bar{k}-1}[i])^\dagger \in {\mathbb F}_q^{\bar{k}}$ to an output $\bx[i] = (x_0[i],\dots,x_{\bar{n}-1}[i])^\dagger \in {\mathbb F}_q^{\bar{n}}$ using a memory $\bar{m}$ encoder\footnote{We use $^\dagger$ to denote the vector/matrix transpose operation. Throughout this paper, we will treat $\bs[i]$ and $\bx[j]$ as column vectors and therefore $\bs^\dagger[i]$ and $\bx^\dagger[j]$ denote the associated row vectors. For convenience, we will not use the $^\dagger$ notation when the dimensions are clear.}. In particular let
\begin{align}
\bx[i] = \left(\sum_{t=0}^{\bar{m}} \bs^\dagger[{i-t}] \cdot \bG_{t}\right)^\dagger, \label{eq:conv-code}
\end{align}
where $\bG_0,\ldots, \bG_{\bar{m}}$ are ${\bar{k}\times \bar{n}}$ matrices with elements in ${\mathbb F}_q$. 
Furthermore the convolutional code is systematic\footnote{Throughout the paper, we only consider systematic Strongly-MDS codes and thus the word systematic is dropped for convenience}
 if we can express each sub-generator matrix in the following form,
\begin{align}
\label{eq:systematic}
\bG_0=[\bI_{\bar{k}\times \bar{k}}~ \bH_0], \qquad \bG_t = [{\bf 0}_{\bar{k} \times \bar{k}}~\bH_t],~t =1,\ldots,\bar{m}
\end{align}
where $\bI_{\bar{k}\times \bar{k}}$ denotes the ${\bar{k}\times \bar{k}}$ identity matrix, ${\bf 0}_{\bar{k} \times \bar{k}}$ denotes the ${\bar{k}\times \bar{k}}$ zero matrix, and ${\bH_t \in {\mathbb F}_q^{\bar{k} \times (\bar{n}-\bar{k})}}$ for $t=0,1,\dots,\bar{m}$. For a systematic convolutional code,~\eqref{eq:conv-code} reduces to 
\begin{align}
\bx[i] = \left[\begin{array}{c}\bs[i] \\ \bp[i]\end{array}\right], \qquad \bp[i] = \left(\sum_{t=0}^{\bar{m}} \bs^\dagger[i-t] \cdot\bH_t\right)^\dagger. \label{eq:conv-code-sys}
\end{align}

The Strongly-MDS codes (see e.g.~\cite[Corollary~2.5]{strongly-mds}), correspond to a certain choice of $\bH_t$ that result in the following error correction properties in the streaming setup.

\begin{lemma}
\label{lem:mds-sub}
Consider a systematic $(\bar{n},\bar{k},\bar{m})$ Strongly-MDS code and suppose that the sub-symbols in $\bx[i]$ i.e.,
\begin{align}
\bx[i] = \left(s_0[i], \ldots, s_{\bar{k}-1}[i], p_0[i],\ldots, p_{\bar{n}-\bar{k}-1}[i]\right) \label{eq:strong-mds-systematic}
\end{align}
are transmitted sequentially in the time interval $[i\cdot \bar{n}, (i+1)\cdot \bar{n} -1]$ over the channel\footnote{Note that in this statement we are only transmitting sub-symbols over ${\mathbb F}_q$ over the channel. Subsequently, we will adapt these properties for transmitting symbols over $({\mathbb F}_q)^{\bar{n}}$}. 
The following properties hold for each $j =0,1,\ldots,\bar{m}$.
\begin{enumerate}
\item[L1.] If $\hat{N} \le (\bar{n}-\bar{k})(j+1)$ transmitted sub-symbols are erased in the interval $[0,(j+1)\bar{n}-1]$, then  $\bs[0] = (s_0[0],\dots,s_{\bar{k}-1}[0])$ can be recovered by time $(j+1)\bar{n}-1$.
\item[L2.] If the channel introduces an erasure-burst of length $\hat{B}$ sub-symbols in the interval $[c, c+\hat{B}-1],$ where $\hat{B} \leq (\bar{n}-\bar{k})(j+1)$ and $0 \le c \le \bar{k}-1,$ then all erased source symbols are recovered by time $(j+1)\bar{n}-1$.
\item[L3.] If the channel introduces an erasure burst of length $\hat{B}$ sub-symbols in the interval $[c, c+\hat{B}-1],$ where $0 \le c \le \bar{k}-1,$ followed by a total of no more than $\hat{I}$ isolated erasures such that $\hat{B}+\hat{I} \le (\bar{n}-\bar{k})(j+1)$, then all the erased symbols in the burst are recovered by time $(j+1)\bar{n}-1$.
\end{enumerate}
\end{lemma}

\begin{proof}
See Appendix~\ref{app:strongly-mds}.
\end{proof}

We now discuss how the properties in Lemma~\ref{lem:mds-sub} can be applied to our system model. In the case in Section~\ref{subsec:eq-rates} when the source and channel-transmission rates are equal, Lemma~\ref{lem:mds-sub} immediately yields the following.


\begin{corol}
\label{cor:mds}
Consider a systematic $(\bar{n},\bar{k},\bar{m})$ Strongly-MDS code of rate $R = \frac{\bar{k}}{\bar{n}}$ which transmits the entire channel symbol $\bx[i] = (x_0[i],\dots,x_{\bar{n}-1}[i]) \in \mathbb F_q^{\bar{n}}$ in time-slot $i$. For each $j =0,1,\ldots,\bar{m}$, we have the following,
\begin{enumerate}
\item[P1.] Suppose that in the window $[0,j],$ the channel introduces $N \le (1-R)(j+1)$ erasures in arbitrary locations, then $\bs[0]$ is recovered by time $t=j$.
\item[P2.] Suppose an erasure burst happens in the interval $[0,B-1],$ where $B \le (1-R)(j+1)$, then all the symbols $\bs[0],\ldots, \bs[B-1]$ are simultaneously recovered by time $t=j$.
\end{enumerate}
$\hfill\Box$
\end{corol}
\begin{proof}
Since the channel erases entire symbols, an erasure of $N$ symbols is equivalent to the erasure of $\hat{N} = \bar{n}N \le (\bar{n}-\bar{k})(j+1)$ sub-symbols in Lemma~\ref{lem:mds-sub}. Property L1 in Lemma~\ref{lem:mds-sub} guarantees that $\bs[0]$ is recovered by time $t = (j+1)\bar{n}-1$, when any $(\bar{n}-\bar{k})(j+1)$ {\em sub-symbols} are erased in the interval $[0, (j+1)\bar{n}-1]$. It immediately follows that if no more than $\frac{\bar{n}-\bar{k}}{\bar{n}}(j+1)$ {\em symbols} in the interval $[0,j]$ are erased, the symbol $\bs[0]$ can be recovered. Furthermore, the interval $[0, (j+1)\bar{n}-1]$ consisting of $(j+1)\bar{n}$ sub-symbols in Lemma~\ref{lem:mds-sub} corresponds to the interval $[0,j]$ consisting of $j+1$ symbols in Corollary~\ref{cor:mds} and thus $\bs[0]$ is recovered by time $j$. Thus Property P1 follows upon substituting $R = \frac{\bar{k}}{\bar{n}}$. 

Property P2 follows in an analogous fashion upon using property L2 in Lemma~\ref{lem:mds-sub} with $c=0$.

\end{proof}

From Corollary~\ref{cor:mds}, it follows that any $(N,B)$ pair that satisfies
\begin{align}
N \le (1-R)(T+1), \quad B \le (1-R)(T+1) \label{eq:NB-MDS}
\end{align}
is achieved using a $(n,k,T)$ Strongly-MDS code of rate $R = \frac{k}{n}$ with delay $T$ and $W \ge T+1$. In particular, if the channel introduces up to $(1-R)(T+1)$ erasures in the window $[0,T]$, it follows from Property P1 in Corollary~\ref{cor:mds} that $\bs[0]$ is recovered at $t=T$. Once $\bs[0]$ has been recovered, its effect can be subtracted out from all parity-checks involving $\bs[0]$. By the same property, $\bs[1]$ is guaranteed to be recovered at time ${t=T+1}$. This argument can be successively repeated until all the erased symols are recovered. Furthermore upon substituting $B=N$ in~\eqref{eq:r-ub}, we note that the Strongly-MDS attain one extreme point on the tradeoff, namely when $N=B$. This is clearly the largest feasible value of $N$ in~\eqref{eq:r-ub}.

In a similar fashion, it can be shown that for the case of unequal source-channel rates in Section~\ref{subsec:uneq-rates} , when $W \ge M(T+1)$, any $(N, B)$ is achievable that satisfies
\begin{align}
N \le M(1-R)(T+1), \quad B \le M(1-R)(T+1). \label{eq:r-MDS-mismatch}
\end{align}
We will omit the details as they are very similar to the justification of~\eqref{eq:NB-MDS}. 

\subsection{Maximally Short (MS) Codes}
\label{subsec:maximally-short}
While the Strongly-MDS codes achieve the extreme point of the upper bound~\eqref{eq:r-ub} corresponding to $N=B$, the Maximally Short (MS) codes
achieve the other extreme point, corresponding to $N=1$. In particular the maximum value of $B$ with $N=1$ is given in the following result.\footnote{The construction in~\cite{MartinianS04,MartinianT07} only consider a single erasure burst during the entire duration. However it easily follows that the resulting codes can correct multiple erasure bursts provided that the separation between them is at-least $T$ symbols or equivalently $W \ge T+1$.} 
\begin{lemma}[Martinian and Sundberg~\cite{MartinianS04}]
\label{lem:ms}
Consider the channel $\cC(N=1,B, W)$ with $W \ge {T+1}$ and $M=1$. There exists an MS code of rate $R$ satisfying
\begin{align}
\label{eq:B-UB}
R = \begin{cases}
\frac{T}{T+B}, & B \le T, \\
0, &\text{else}.
\end{cases}
\end{align}
$\hfill\Box$
\end{lemma}

Furthermore, $R$ in~\eqref{eq:B-UB} is the maximum achievable rate for $\cC(N=1,B,W \ge T+1)$ channel.

The construction of MS codes presented in~\cite{MartinianS04,MartinianT07} involves first constructing a specific low-delay block code and then converting it into a streaming code using a diagonal interleaving technique. Thus the problem of constructing a streaming code is reduced to the problem of constructing a block code with certain properties. While such a simplification is appealing, unfortunately it does not appear to easily generalize when seeking extensions of MS codes.
Note that the above MS codes can only achieve $N=1$ and are highly sensitive to isolated losses over the channel. In~\cite{MartinianS04} some examples of codes with higher $N$ were reported using a numerical search but a general approach for constructing robust streaming codes remained elusive.
In Section~\ref{subsec:gms}, we present an alternative perspective that easily extends to achieve a near optimal rate for any $(N,B)$.

\begin{center}
\begin{figure}[t]
	\centering
	\resizebox{\columnwidth}{!}{
	\begin{tikzpicture}[node distance=1mm]
		\node[usym,minimum height = 3em,minimum width = 4em]  										(x100) {$\bx[i,1]$};
		\node[usym, right = of x100,minimum height = 3em,minimum width = 4em]     (x101) {$\bx[i,2]$};
		\node[usym, right = of x101,minimum height = 3em,minimum width = 4em]     (x102) {$\dots$};
		\node[usym, right = of x102,minimum height = 3em,minimum width = 4em]     (x103) {$\dots$};
		\node[usym, right = of x103,minimum height = 3em,minimum width = 4em]     (x104) {$\bx[i,M]$};
		\node[usym, above = 3em of x100,minimum height = 3em,minimum width = 4em]     (w100) {$\bw[i,1]$};
		\node[usym, above = 3em of x101,minimum height = 3em,minimum width = 4em]     (w101) {$\bw[i,2]$};
		\node[usym, above = 3em of x102,minimum height = 3em,minimum width = 4em]     (w102) {$\dots$};
		\node[usym, above = 3em of x103,minimum height = 3em,minimum width = 4em]     (w103) {$\dots$};
		\node[usym, above = 3em of x104,minimum height = 3em,minimum width = 4em]     (w104) {$\bw[i,M]$};
		\node[nosym, font=\LARGE, above = 3em of w100]    												(s100) {$\bs[i]$};
		\node[nosym, below = 2em of x100,minimum width = 4em]    (rw100) {$$};
		\node[nosym, below = 2em of rw100,minimum width = 4em]   (rs100) {$$};
		
		\node[usym, right = of x104,minimum height = 3em,minimum width = 4em]  		(x200) {$\bx[i+1,1]$};
		\node[usym, right = of x200,minimum height = 3em,minimum width = 4em]     (x201) {$\bx[i+1,2]$};
		\node[usym, right = of x201,minimum height = 3em,minimum width = 4em]     (x202) {$\dots$};
		\node[usym, right = of x202,minimum height = 3em,minimum width = 4em]     (x203) {$\dots$};
		\node[usym, right = of x203,minimum height = 3em,minimum width = 4em]     (x204) {$\bx[i+1,M]$};
		\node[usym, above = 3em of x200,minimum height = 3em,minimum width = 4em]     (w200) {$\bw[i+1,1]$};
		\node[usym, above = 3em of x201,minimum height = 3em,minimum width = 4em]     (w201) {$\bw[i+1,2]$};
		\node[usym, above = 3em of x202,minimum height = 3em,minimum width = 4em]     (w202) {$\dots$};
		\node[usym, above = 3em of x203,minimum height = 3em,minimum width = 4em]     (w203) {$\dots$};
		\node[usym, above = 3em of x204,minimum height = 3em,minimum width = 4em]     (w204) {$\bw[i+1,M]$};
		\node[nosym, font=\LARGE, above = 3em of w200]     												(s200) {$\bs[i+1]$};
		
		\node[usym, right = of x204,minimum height = 3em,minimum width = 4em]  		(x300) {$$};
		\node[usym, right = of x300,minimum height = 3em,minimum width = 4em]     (x301) {$$};
		\node[usym, right = of x301,minimum height = 3em,minimum width = 4em]     (x302) {$$};
		\node[usym, right = of x302,minimum height = 3em,minimum width = 4em]     (x303) {$$};
		\node[usym, right = of x303,minimum height = 3em,minimum width = 4em]     (x304) {$$};
		
		\node[usym, right = of x304,minimum height = 3em,minimum width = 4em]  		(x500) {$\bx[i+T,1]$};
		\node[usym, right = of x500,minimum height = 3em,minimum width = 4em]     (x501) {$\bx[i+T,2]$};
		\node[usym, right = of x501,minimum height = 3em,minimum width = 4em]     (x502) {$\dots$};
		\node[usym, right = of x502,minimum height = 3em,minimum width = 4em]     (x503) {$\dots$};
		\node[usym, right = of x503,minimum height = 3em,minimum width = 4em]     (x504) {$\bx[i+T,M]$};
		\node[usym, above = 3em of x500,minimum height = 3em,minimum width = 4em]     (w500) {$\bw[i+T,1]$};
		\node[usym, above = 3em of x501,minimum height = 3em,minimum width = 4em]     (w501) {$\bw[i+T,2]$};
		\node[usym, above = 3em of x502,minimum height = 3em,minimum width = 4em]     (w502) {$\dots$};
		\node[usym, above = 3em of x503,minimum height = 3em,minimum width = 4em]     (w503) {$\dots$};
		\node[usym, above = 3em of x504,minimum height = 3em,minimum width = 4em]     (w504) {$\bw[i+T,M]$};
		\node[nosym, font=\LARGE, above = 3em of w500]     												(s500) {$\bs[i+T]$};
		\node[nosym, below = 2em of x500]     												(rw500) {$\bw[i,1]$};
		\node[nosym, below = 2em of x501]     												(rw501) {$\bw[i,2]$};
		\node[nosym, below = 2.5em of x502]     												(rw502) {$\dots$};
		\node[nosym, below = 2.5em of x503]     												(rw503) {$\dots$};
		\node[nosym, below = 2em of x504]     												(rw504) {$\bw[i,M]$};
		\node[nosym, font=\LARGE, below = 3em of rw504]     												(rs504) {$\bs[i]$};
		
		\node      [right = of x504]     (x1end) {$\cdots$};
		
		\node[nosym, font=\LARGE, left = 2em of s100]     												(s000) {Source Stream};
		\node[nosym, font=\LARGE, left = 2em of w100]     												(w000) {Expanded Source Stream};
		\node[nosym, font=\LARGE, left = 2em of x100]     												(x000) {Channel Packets};
		\node[nosym, font=\LARGE, left = 2em of rs100]     												(rw000) {Source Recovery};
		\node[nosym, font=\LARGE, left = 2em of rw100]     												(rs000) {Expanded Source Recovery};
		
		\draw [->,rounded corners, very thick] (s100.south) -- (w100.north);
		\draw [->,rounded corners, very thick] (s100.south) -- (w101.north);
		\draw [->,rounded corners, very thick] (s100.south) -- (w102.north);
		\draw [->,rounded corners, very thick] (s100.south) -- (w103.north);
		\draw [->,rounded corners, very thick] (s100.south) -- (w104.north);
		
		\draw [->,rounded corners, very thick] (s200.south) -- (w200.north);
		\draw [->,rounded corners, very thick] (s200.south) -- (w201.north);
		\draw [->,rounded corners, very thick] (s200.south) -- (w202.north);
		\draw [->,rounded corners, very thick] (s200.south) -- (w203.north);
		\draw [->,rounded corners, very thick] (s200.south) -- (w204.north);
		
		\draw [->,rounded corners, very thick] (s500.south) -- (w500.north);
		\draw [->,rounded corners, very thick] (s500.south) -- (w501.north);
		\draw [->,rounded corners, very thick] (s500.south) -- (w502.north);
		\draw [->,rounded corners, very thick] (s500.south) -- (w503.north);
		\draw [->,rounded corners, very thick] (s500.south) -- (w504.north);
		
		\draw [->,rounded corners, very thick] (x500.south) -- (rw500.north);
		\draw [->,rounded corners, very thick] (x501.south) -- (rw501.north);
		\draw [->,rounded corners, very thick] (x502.south) -- (rw502.north);
		\draw [->,rounded corners, very thick] (x503.south) -- (rw503.north);
		\draw [->,rounded corners, very thick] (x504.south) -- (rw504.north);

		\draw [->,rounded corners, very thick] (rw500.south) -- (rs504.north);
		\draw [->,rounded corners, very thick] (rw501.south) -- (rs504.north);
		\draw [->,rounded corners, very thick] (rw502.south) -- (rs504.north);
		\draw [->,rounded corners, very thick] (rw503.south) -- (rs504.north);
		\draw [->,rounded corners, very thick] (rw504.south) -- (rs504.north);
		
		\braceup{x100}{x104}{1mm}{{$\bX[i,:]$}};
		\braceup{x200}{x204}{1mm}{{$\bX[i+1,:]$}};
		\braceup{x500}{x504}{1mm}{{$\bX[i+T,:]$}};

	\end{tikzpicture}}
	\caption{Each source symbol $\bs[i]$ is split into $M$ sub-packets i.e., $\bs[i] = (\bw[i,1],\bw[i,2],\dots,\bw[i,M])$. The expanded source stream is then encoded using a Maximally-Short code. 	The decoder recovers each $\bw[i,j]$ once $\by[i+T,j]$ is received which ensures that $\bs[i]$ is recovered by the end of the macro-packet $i+T$.}
	\label{fig:MS-Unequal}
\end{figure}
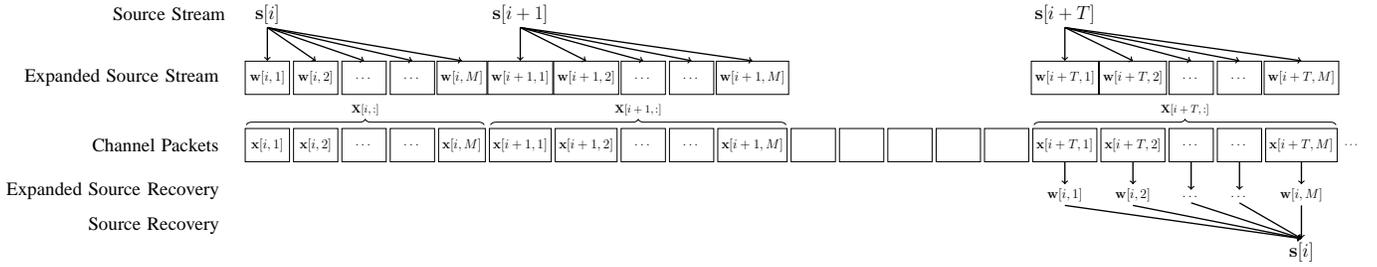
\end{center}

For the case of unequal source-channel rates, a straightforward adaptation of the MS codes is as follows. We split each symbol $\bs[i]$ into $M$ sub-symbols, one for each time-slot in the macro-packet and then apply a MS code in Lemma~\ref{lem:ms} to this expanded source stream with delay $T' = MT$ (cf. Fig.~\ref{fig:MS-Unequal}). In particular we assume that $\bs[i] \in ({\mathbb F}_q)^{kM}$ and proceed as follows.
\begin{itemize}
\item Split each $\bs[i] = (\bw[i,1],\ldots, \bw[i,M])$ where $\bw[i,j] \in {\mathbb F}_q^k$.
\item Apply a MS code in Lemma~\ref{lem:ms} for the $\cC(N=1, B, W)$ channel with delay $T' = MT$ (channel symbols) and $W \ge M(T+1)$.
\item Transmit the associated channel symbol $\bx[i,j] \in {\mathbb F}_q^n$ in slot $j$ of the macro-packet $i$.
\end{itemize}
From~\eqref{eq:B-UB} we have that
\begin{align}
R = \frac{MT}{MT+B} =\frac{T}{T+b + \frac{B'}{M}}
\label{eq:R-SCO}
\end{align} 
is achievable when $B \le MT$. Note that in the second equality in~\eqref{eq:R-SCO}, we use~\eqref{eq:B-split}. 
Note that the delay of $T' = M\cdot T$ channel symbols in the expanded stream, implies that  $\bw[i,j]$ is recovered when $\by[i+T, j]$ is received for each $j \in \{1,2,\ldots, M\}$. Thus the entire source symbol $\bs[i]$ is guaranteed to be recovered at the end of macro-packet $i+T$, thus satisfying the delay constraint. We note that the rate in~\eqref{eq:R-SCO} is only positive if $B \le MT$ and attains the capacity in Theorem~\ref{thm:capacity} in the special case when $B'=0$. If $B> MT$ the above construction is not feasible and the rate attained is zero. 

\section{Equal Source-Channel Rates}
\label{sec:midas}

In this section, we consider the case when source-arrival and channel-transmission rates are equal i.e., $M=1$. We start by establishing the upper-bound in Theorem~\ref{thm:Chan-1-UB} in Section~\ref{subsec:match-ch1-converse}. In Section~\ref{subsec:gms}, we give a generalized construction for Maximally Short codes which are optimal over the $\cC(N=1,B,W)$ channel. We provide a robust extension of these codes -- MiDAS codes -- in Section~\ref{subsec:midas} and study their optimality. In Section~\ref{subsec:midas-field}, we provide an alternative construction of MiDAS codes achieving the same tradeoff in Theorem~\ref{thm:midas} but with a smaller field-size. Finally, we compare the performance of the two constructions through an example in Section~\ref{subsec:midas-nonideal}.

\subsection{Upper-bound}
\label{subsec:match-ch1-converse}

To establish the upper bound in Theorem~\ref{thm:Chan-1-UB}, we separately consider the cases where $W \ge {T+1}$
and $W < T+1$. When $W \ge T+1$, consider a periodic erasure channel with a period of ${\tau_{P} = T+B-N+1}$ and suppose that in every such period the first
${B}$ symbols are erased (see Fig.~\ref{fig:Chan1-PEC}). While such a channel is not included in $\cC(N,B,W)$, we nonetheless show that any code for $\cC(N,B,W)$ and delay $T$ is also feasible for the proposed periodic erasure channel.\footnote{A similar converse argument involving periodic erasure channel is also presented in \cite{MartinianS04,de-sco}. For an information theoretic argument, we refer the reader to~\cite{LuiMASc,mu-sco}, but it will not be presented in this paper.}

Consider the first period that spans the interval $[0, \tau_P-1]$. We note the following
\begin{itemize}
\item The first $B-N+1$ symbols, $\{ \bs[i] \}_{0 \le i \le B-N}$, must be all recovered with delay $T$ since the recovery window $[i,i+T]$ of each such symbol only have a burst of length ${B}$ or smaller. Thus all these symbols are recovered by time $t=\tau_P-1$.
\item The recovery window of each of the $N-1$ symbols, $\{ \bs[i] \}_{B-N+1 \le i \le B-1}$ is $[i,i+T]$ which sees two bursts. The first burst spans $[i, B-1]$ and is of length $B-i$. The second burst spans $[T+B-N+1, i+T]$ and is of length $i+N-B$. Thus the total number of erased symbols in each recovery period is exactly $N$. Thus any feasible code over the $\cC(N,B,W)$ channels guarantees that each such symbol is also recovered at time $i+T$.
\item The recovery window of each of the remaining symbols in the first period, $\bs[B],\ldots, \bs[\tau_P-1]$, again sees a single-erasure burst of length $B$ at the end of the window. Hence, each of these symbols is also guaranteed to be recovered with delay $\le T$, in particular, by time $\tau_P -1$.
\end{itemize}
We have thus shown that all the symbols in the first period spanning $[0,\tau_P-1]$ can be recovered with delay $T$. We can repeat the same argument for all the remaining periods and thus the claim follows. Thus using the capacity of the periodic erasure channel, we have
\begin{align}
R \le 1- \frac{B}{T+B-N+1}.
\label{eq:Rbnd}
\end{align}

\begin{figure*}
    \centering
	\resizebox{0.75\columnwidth}{!}{
	\begin{tikzpicture}[node distance=0mm]
		\node                       (x1start) {Link:};
		\node[esym, right = of x1start]  (x100) {};
		\node[esym, right = of x100]     (x101) {};
		\node[esym, right = of x101]     (x102) {};
		\node[esym, right = of x102]     (x103) {};
		\node[esym, right = of x103]     (x104) {};
		\node[esym, right = of x104]     (x105) {};
		\node[usym, right = of x105]     (x106) {};
		\node[usym, right = of x106]     (x107) {};
		\node[usym, right = of x107]     (x108) {};
		\node[usym, right = of x108]     (x109) {};
		\node[usym, right = of x109]     (x110) {};
		\node[esym, right = of x110]     (x111) {};
		\node[esym, right = of x111]     (x112) {};
		\node[esym, right = of x112]     (x113) {};
		\node[esym, right = of x113]     (x114) {};
		\node[esym, right = of x114]     (x115) {};
		\node[esym, right = of x115]     (x116) {};
		\node[usym, right = of x116]     (x117) {};
		\node[usym, right = of x117]     (x118) {};
		\node[usym, right = of x118]     (x119) {};
		\node[usym, right = of x119]     (x120) {};
		\node[usym, right = of x120]     (x121) {};
		\node[esym, right = of x121]     (x122) {};
		\node[esym, right = of x122]     (x123) {};
		\node[esym, right = of x123]     (x124) {};
		\node[esym, right = of x124]     (x125) {};
		\node[esym, right = of x125]     (x126) {};
		\node[esym, right = of x126]     (x127) {};
		\node[usym, right = of x127]     (x128) {};
		\node[usym, right = of x128]     (x129) {};
		\node[usym, right = of x129]     (x130) {};
		\node[usym, right = of x130]     (x131) {};
		\node[usym, right = of x131]     (x132) {};
		\node      [right = of x132]     (x1end) {$\cdots$};
		\dimup{x100}{x105}{2mm}{$B$};
		\dimdn{x100}{x103}{-2mm}{$B-N+1$};
		\dimup{x106}{x110}{2mm}{$T-N+1$};
		\dimup{x111}{x116}{2mm}{$B$};
		\dimdn{x111}{x114}{-2mm}{$B-N+1$};
		\dimup{x117}{x121}{2mm}{$T-N+1$};
		\dimup{x122}{x127}{2mm}{$B$};
		\dimdn{x122}{x125}{-2mm}{$B-N+1$};
		\dimup{x128}{x132}{2mm}{$T-N+1$};
	\end{tikzpicture}}
	 \caption{The periodic erasure channel in the  proof  Theorem~\ref{thm:Chan-1-UB}. The shaded symbols are erased while the remaining ones are received by the destination.}
	\label{fig:Chan1-PEC}
\end{figure*}
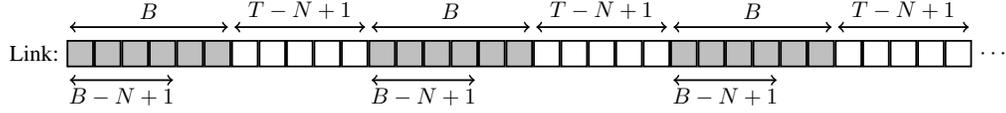

For the case when ${W < T+1}$, we consider a periodic erasure channel with a period of $\tau_P = {W+B-N}$ where in each period the first ${B}$ symbols are erased and the remaining $W-N$ symbols are not erased. Such a channel by construction is a $\cC(N,B,W)$ channel. In any window $\cW_i = [i,i+W-1]$ of length $W$ there exists either a single burst of maximum length $B$, or up to $N$ isolated erasures. Thus every erased symbol on such a channel must be recovered i.e., we have that
\begin{align}
R \le 1- \frac{B}{W+B-N}.
\label{eq:Rbnd2}
\end{align}

Rearranging~\eqref{eq:Rbnd} and~\eqref{eq:Rbnd2} and using $T_\mrm{eff} = \min(W-1, T)$, we easily recover~\eqref{eq:r-ub}. This completes the proof of the upper bound.

\subsection{Generalized MS Codes}
\label{subsec:gms}

In this section we present a generalization of the MS codes introduced in section~\ref{subsec:maximally-short}.
The proposed construction applies to any ${W\ge B+1}$, and eliminates the intermediate step of constructing a block code. 
\begin{prop}
\label{prop:gen-MS}
Let $T_\mrm{eff} \defeq \min(W-1, T)$. For the $\cC(N=1, B, W)$ channel, there exists a streaming code with delay $T$ and rate 
\begin{align}
R = 
\begin{cases}
\frac{T_\mrm{eff}}{T_\mrm{eff}+B}, & \Te \ge B \\
0, & \text{else}.
\end{cases}
\label{eq:gen-MS}
\end{align}
$\hfill\Box$
\end{prop}

The encoding steps are as follows,

\begin{itemize}
\item \textbf{Source Splitting:} Split each source symbol $\bs[i] \in {\mathbb F}_q^{k}$ into two groups $\bu[i] \in {\mathbb F}_q^{k^u}$ and $\bv[i] \in {\mathbb F}_q^{k^v}$ as follows,
\begin{align}
\bs[i] = (\underbrace{u_0[i],\ldots, u_{k^u-1}[i]}_{=\bu[i]}, \underbrace{v_0[i],\ldots, v_{k^v-1}[i]}_{=\bv[i]}),
\label{eq:s-split}
\end{align}
where $k^u+k^v=k$, i.e., $\bu[i]$ constitutes the first $k^u$ sub-symbols in $\bs[i]$ whereas $\bv[i]$ constitutes the remaining $k^v$ sub-symbols.

\item \textbf{Strongly-MDS Parity-Checks:} Apply a $(k^u+k^v, k^v, \Te)$ Strongly-MDS code of rate $R^v = \frac{k^v}{k^u+k^v}$ on the symbols $\bv[i]$ and generate parity-check symbols
\begin{align}
\bp^v[i] = \left(\sum_{j=0}^{\Te} \bv^\dagger[i-j]\cdot \bH^v_j\right)^\dagger, \quad \bp^v[i] \in {\mathbb F}_q^{k^u} \label{eq:pv-conc},
\end{align}
where the matrices $\bH_j^v \in {\mathbb F}_q^{k^v \times k^u}$ are associated with a Strongly-MDS code~\eqref{eq:systematic}.
\item \textbf{Repetition Code:} Superimpose the $\bu[\cdot]$ symbols onto $\bp^v[\cdot]$ and let \begin{align}
\bq[i] = \bp^v[i] + \bu[i-\Te]. \label{eq:MS-q}
\end{align}
\item \textbf{Channel Packet Generation:} Concatenate the generated parity-checks to the source symbols so that the channel input at time $i$ is given by $\bx[i] = \left(\bu[i],\bv[i],\bq[i]\right) \in {\mathbb F}_q^n$, where $n = 2k^u+k^v$.
\end{itemize}


In our construction discussed above, we select $k = \Te$, $k^u = B$, $k^v = \Te-B$ and $n = \Te +B$. Clearly the rate of the proposed code $R = \frac{k}{n}$ matches the expression in~\eqref{eq:gen-MS}.

For decoding, we suppose that the first erasure burst of length $B$ spans the interval $[0,B-1]$. We are then guaranteed that there are no additional erasures in the interval $[B, T_\mrm{eff}+B-1]$ for the $\cC(N=1, B,W)$ channel. We claim that each $\bs[0],\bs[1],\ldots, \bs[B-1]$ is recovered by time $t = T_\mrm{eff}, \Te+1,\ldots, T_\mrm{eff} + B-1$ respectively.

The decoder proceeds in two steps. 
\begin{itemize}
\item {\em{Simultaneously recover}} $\bv[0],\dots,\bv[B-1]$ by time $t = \Te-1$. In this step the decoder proceeds as follows. For each $j \in \{B,\ldots, \Te-1\}$, the decoder recovers the parity-check symbols $\bp^v[j]$, by subtracting the unerased $\bu[j-\Te]$ from the associated $\bq[j]=\bp^v[j] + \bu[j-\Te]$ symbols. These recovered parity-checks can then be used to recover $\bv[0],\dots,\bv[B-1]$. 

Note that using property P2 in Corollary~\ref{cor:mds} and substituting $R = R^v$ and $j=\Te-1$ we get,
\begin{align}
\label{eq:simultaneous}
(1-R^v)\Te = B,
\end{align}
and hence the recovery of $\bv[0],\dots,\bv[B-1]$ by time $t = \Te-1$ is guaranteed.
\item {\em{Sequentially recover}} $\bu[0],\dots,\bu[B-1]$ at times $\Te,\dots,\Te+B-1$, respectively. Consider the parity-checks $\bq[j] = \bu[j-\Te] + \bp^v[j]$ for $j \in \{\Te, \dots, \Te+B-1\}$, which are available to the decoder. Upon the recovery of $\bv[0],\dots,\bv[B-1]$ in the previous step, the required $\bp^v[j]$ can be computed, subtracted from $\bq[j]$, and the underlying $\bu[\cdot]$ symbols can be sequentially decoded by their deadlines.

\end{itemize}
Upon completion of the two steps stated above, the recovery of $\bs[i]$ for $i \in \{0,\ldots, B-1\}$ follows. Any subsequent burst, starting at time $t \ge \Te + B$, can be corrected in a similar fashion. Since the rate of the code is clearly given by~\eqref{eq:gen-MS}, the proof of Prop.~\ref{prop:gen-MS} is complete.

The Generalized MS code is no longer feasible when $N>1$. To see this consider two isolated erasures one at $t=0$ and the other at $t=\Te$. In this case both $\bu[0]$ as well as its repeated copy are erased. We propose a modification to these codes that can deal with any value of $N$ in Theorem~\ref{thm:midas}.

\subsection{MiDAS: Code Construction}
\label{subsec:midas}

\begin{figure}[t]
\resizebox{\columnwidth}{!}{
\begin{tikzpicture}[node distance=1mm,
  ]
  \tikzstyle{triple2} = [rectangle split, anchor=text,rectangle split parts=4]
  \tikzstyle{double2} = [rectangle split, anchor=text,rectangle split parts=2]
  \tikzstyle{triple} = [draw, rectangle split,rectangle split parts=4]
	\tikzstyle{double} = [draw, rectangle split,rectangle split parts=2]
	\tikzset{block/.style={rectangle,draw}}
	
	
	\node[triple2, minimum width=1.6cm] (start) {$k^u$ Sub-symbols
    \nodepart{second}
      $k^v$ Sub-symbols
    \nodepart{third}
     \tikz{\node[double2] {$k^u$ \nodepart{second}Sub-symbols};}
		\nodepart{fourth}
		$k^s$ Sub-symbols
  };
	
  \node[triple,  right = of start,fill=light-gray,minimum width=1.6cm] (p1) {$\bu[0]$
    \nodepart{second}
      $\bv[0]$
    \nodepart{third}
      \tikz{\node[double2] {$\bu[-\Te]$ \nodepart{second}$+\bp^v[0]$};}
			\nodepart{fourth}
			$\bp^u[0]$
  };
  
  \node[triple, right = of p1,fill=light-gray,minimum width=1.6cm] (p11) {$\bu[1]$
    \nodepart{second}
      $\bv[1]$
    \nodepart{third}
      \tikz{\node[double2] {$\bu[-\Te+1]$ \nodepart{second}$+\bp^v[1]$};}
			\nodepart{fourth}
			$\bp^u[1]$
  };
  
  \node[triple, font=\color{light-gray}, right = of p11,fill=light-gray,minimum width=1.6cm] (p2) {$\bu[0]$
    \nodepart{second}
      $\bv[0]$
    \nodepart{third}
      \tikz{\node[double2] {$\bu[-\Te]$ \nodepart{second}$+\bp^v[0]$};}
			\nodepart{fourth}
			$\bp^u[0]$
  };
	
  \node[triple,  right = of p2,fill=light-gray,minimum width=1.6cm] (p3) {$\bu[B-1]$
    \nodepart{second}
      $\bv[B-1]$
    \nodepart{third}
      \tikz{\node[double2] {$\bu[-\Te+B-1]$ \nodepart{second}$+\bp^v[B-1]$};}
			\nodepart{fourth}
			$\bp^u[B-1]$
  };
  
  \node[triple,  right = of p3,minimum width=1.6cm] (p4) {$\bu[B]$
    \nodepart{second}
      $\bv[B]$
    \nodepart{third}
      \tikz{\node[double2] {$\bu[-\Te+B]$ \nodepart{second}$+\bp^v[B]$};}
			\nodepart{fourth}
			$\bp^u[B]$
  };
  
  \node[triple, font=\color{white}, right = of p4,minimum width=1.6cm] (p5) {$\bu[\Te]$
    \nodepart{second}
      $\bv[\Te]$
    \nodepart{third}
      \tikz{\node[double2] {$\bu[0]+$ \nodepart{second}$\bp^v[\Te]$};}
			\nodepart{fourth}
			$\bp^u[\Te]$
  };
	
  \node[triple,  right = of p5,minimum width=1.6cm] (p6) {$\bu[\Te-1]$
    \nodepart{second}
      $\bv[\Te-1]$
    \nodepart{third}
      \tikz{\node[double2] {$\bu[-1]+$ \nodepart{second}$\bp^v[\Te-1]$};}
			\nodepart{fourth}
			$\bp^u[\Te-1]$
  };
  
  \node[triple, right = of p6,minimum width=1.6cm] (p7) {$\bu[\Te]$
    \nodepart{second}
      $\bv[\Te]$
    \nodepart{third}
      \tikz{\node[double2] {$\bu[0]+$ \nodepart{second}$\bp^v[\Te]$};}
			\nodepart{fourth}
			$\bp^u[\Te]$
  };
  
  \bracedn{p1}{p3}{-2mm}{\footnotesize{Erased Packets}};
  \bracedn{p4}{p6}{-2mm}{\footnotesize{Recover $\bv[0],\cdots ,\bv[B-1]$ using $\bp^v[\cdot]$}};
  \bracedn{p7}{p7}{-2mm}{\footnotesize{Recover $\bu[0]$}};

\end{tikzpicture}}
\caption{A window of $\Te+1$ channel packets showing the decoding steps of a MiDAS code when an erasure burst takes place. In the case of isolated erasures, the $\bv[\cdot]$ and $\bu[\cdot]$ symbols are recovered separately using the $\bp^v[\cdot]$ parities in the window $[0,\Te-1]$ and $\bp^u[\cdot]$ parities in the window $[0,\Te]$, respectively. 
}
\label{fig:midas-const}
\end{figure}

Our proposed construction is based on a layered approach. As illustrated in Fig.~\ref{fig:midas-const}, we first construct a Generalized MS code for $\cC(N=1,B,W)$ channel and then concatenate an additional layer of parity symbols when $N > 1$. We again assume that $\bs[i] \in {\mathbb F}_q^{k}$ and split it into two groups $\bu[i]$ and $\bv[i]$ as in~\eqref{eq:s-split} and generate the parity-checks $\bq[i]$ as in~\eqref{eq:MS-q}. The resulting code up to this point can only correct burst erasures. 
We further apply a $(k^u+k^s,k^u,\Te)$ Strongly-MDS code of rate $R^u = \frac{k^u}{k^u+k^s}$ to the $\bu[i]$ symbols and generate additional parity-check symbols,
\begin{align}
\bp^u[i] = \left(\sum_{j=0}^{\Te} \bu^\dagger[i-j]\cdot \bH^u_j \right)^\dagger, \qquad \bp^u[i] \in {\mathbb F}_q^{k^s} \label{eq:pu-conc},
\end{align}
where $\bH^u_j \in {\mathbb F}_q^{k^u \times k^s}$ are matrices associated with a Strongly-MDS code~\eqref{eq:systematic}. We simply concatenate the parity-checks $\bq[i]$ and $\bp^u[i]$ with the source symbols i.e., ${\bx[i] = \left(\bu[i], \bv[i], \bq[i], \bp^u[i]\right)}$. Fig.~\ref{fig:midas-const} illustrates the layered approach in our code construction.
Note that $\bx[i] \in {\mathbb F}_q^{n}$, where $n = 2k^u+k^v+k^s$ and the associated rate is given by $R = \frac{k}{n}$.

In our construction, we select $k^u = B$, $k^v = \Te-B$, $k = k^u+k^v = \Te$ and,
\begin{align}
k^s = \frac{N}{\Te-N+1}k^u,\label{eq:K-opt}
\end{align}

\begin{remark}
We note that if the value of $k^s$ in~\eqref{eq:K-opt} is non-integer, extra source splitting by a certain factor of $m$ is needed. In particular, we set $k^u = mB$, $k^v = m(\Te-B)$, $k = k^u+k^v = m\Te$ and $k^s = \frac{N}{\Te-N+1} k^u = \frac{N}{\Te-N+1} mB$. It can be clearly seen that choosing $m = \Te-N+1$ is sufficient for $k^s$ to be an integer.
\end{remark}

\subsubsection*{Decoder Analysis}
In the analysis of the decoder, we consider the interval $[0,\Te]$ and show that the decoder can recover $\bs[0]$ by time $t=\Te$ if there is either an erasure burst of length $B$ or smaller, or up to $N$ isolated erasures in this interval. Once we show the recovery of $\bs[0]$ by time $t=\Te$, we can cancel its effect from all future parity-check symbols if necessary. The same argument can then be used to show that $\bs[1]$ can be recovered by time ${\Te+1}$ if there are no more than $N$ isolated erasures or a single burst erasure of maximum length $B$ in the interval $[1,\Te+1]$. Recursively continuing this argument we are guaranteed the recovery of each $\bs[i]$ by time ${i+\Te}$.

If there is a burst of length $B$ in the interval $[0,\Te]$ our construction of $\bq[\cdot]$ already guarantees the recovery of $\bs[0]$ by time $t=\Te$ (cf.~Section~\ref{subsec:gms}). Thus we only need to consider the case when there are $N$ isolated erasures in the interval $[0,\Te]$. We show that the decoder is guaranteed to recover $\bv[0]$ at time $t=\Te-1$
using the parity-checks $\bq[\cdot]$ and $\bu[0]$ at time $t=\Te$ using the parity-checks $\bp^u[\cdot]$. 

The recovery of $\bv[0]$ by time $\Te-1$ follows in a fashion similar to the {\emph{simultaneous recovery}} step above~\eqref{eq:simultaneous} in the previous section. However, we use P1 in corollary~\ref{cor:mds} instead. Recall from~\eqref{eq:MS-q} that $\bq[i] = \bp^v[i] + \bu[i-\Te],$ where $\bp^v[i]$ are the parity-checks of the Strongly-MDS code~\eqref{eq:pv-conc}. Since the interfering $\bu[\cdot]$ symbols in the interval $[0,\Te-1]$ are not erased, they can be canceled out by the decoder from $\bq[\cdot]$ and the corresponding parity-checks $\bp^v[\cdot]$ are recovered at the decoder. Since the code $(\bv[i], \bp^v[i])$ is a Strongly-MDS code of rate $R^v = \frac{\Te-B}{\Te}$, applying property P1 in Corollary~\ref{cor:mds} the number of isolated erasures under which the recovery of $\bv[0]$ is possible is given by $N^v = (1-R^v)\Te= B$. Since $N \le B$ holds, the recovery of $\bv[0]$ by time ${t=\Te-1}$ is guaranteed by the code construction.

For recovering $\bu[0]$ at time $t=\Te$, we use the $\bp^u[\cdot]$ parity-checks in the interval $[0,\Te]$. Note that the associated code
$(\bu[i], \bp^u[i])$ is a Strongly-MDS code with rate $R^u = \frac{k^u}{k^u+k^s}$ and hence it follows from P1 in Corollary~\ref{cor:mds} that the number of isolated erasures under which the recovery of $\bu[0]$ is possible is given by
\begin{align}
(1-R^u)(\Te+1) = \frac{k^s}{k^s+k^u}(\Te+1) = N,
\end{align}
where we substitute~\eqref{eq:K-opt} in the last equality. The rate of the code satisfies
\begin{align}
R &= \frac{k^u+k^v}{2k^u+k^v+k^s} = \frac{\Te}{\Te+B + B\frac{N}{\Te-N+1}}\label{eq:K-sub}\\
&> \frac{\Te}{\Te+B + B\frac{N}{\Te-N}} \notag\\
&= \frac{\Te-N}{\Te-N+B}\label{eq:R-lb-final}
\end{align}
where~\eqref{eq:K-sub} follows by substituting in~\eqref{eq:K-opt}. Rearranging~\eqref{eq:R-lb-final}
we have that
\begin{align}
\frac{R}{1-R}B + N > \Te.\label{eq:midas-lb}
\end{align}
The proof of Theorem~\ref{thm:midas} is thus completed.
$\hfill\blacksquare$

\subsubsection*{Example - MiDAS $(N,B,T) = (2,3,4)$ and $W \ge T+1 = 5$}
\begin{table}[tp]
	\centering
\scriptsize{
		\begin{tabular}{c|c|c|c|c|c}
			& $[i]$ 			& $[i+1]$ 			& $[i+2]$ 			& $[i+3]$ 			& $[i+4]$ 			\\ \hline
			\multirow{3}{*}{$k^u = 3$}
			& $u_{0}[i]$ 	& $u_{0}[i+1]$	& $u_{0}[i+2]$	& $u_{0}[i+3]$	& $u_{0}[i+4]$ 	\\
			& $u_{1}[i]$ 	& $u_{1}[i+1]$	& $u_{1}[i+2]$	& $u_{1}[i+3]$	& $u_{1}[i+4]$ 	\\
			& $u_{2}[i]$ 	& $u_{2}[i+1]$	& $u_{2}[i+2]$	& $u_{2}[i+3]$	& $u_{2}[i+4]$ 	\\\hline
			\multirow{1}{*}{$k^v = 1$}
			& $v_{0}[i]$ 	& $v_{0}[i+1]$	& $v_{0}[i+2]$	& $v_{0}[i+3]$	& $v_{0}[i+4]$ 	\\\hline
			\multirow{3}{*}{$k^u = 3$}
			& $u_0[i-4]+p^v_{0}[i]$	& $u_0[i-3]+p^v_{0}[i+1]$	& $u_0[i-2]+p^v_{0}[i+2]$		& $u_0[i-1]+p^v_{0}[i+3]$	& $u_0[i]+p^v_{0}[i+4]$	\\
			& $u_1[i-4]+p^v_{1}[i]$	& $u_1[i-3]+p^v_{1}[i+1]$	& $u_1[i-2]+p^v_{1}[i+2]$		& $u_1[i-1]+p^v_{1}[i+3]$	& $u_1[i]+p^v_{1}[i+4]$	\\
			& $u_2[i-4]+p^v_{2}[i]$	& $u_2[i-1]+p^v_{2}[i+1]$	& $u_2[i-2]+p^v_{2}[i+2]$		& $u_2[i-1]+p^v_{2}[i+3]$	& $u_2[i]+p^v_{2}[i+4]$	\\\hline
			\multirow{2}{*}{$k^s = 2$}
			& $p^u_{0}[i]$ 	& $p^u_{0}[i+1]$	& $p^u_{0}[i+2]$	& $p^u_{0}[i+3]$	& $p^u_{0}[i+4]$ 	\\
			& $p^u_{1}[i]$ 	& $p^u_{1}[i+1]$	& $p^u_{1}[i+2]$	& $p^u_{1}[i+3]$	& $p^u_{1}[i+4]$ 	\\\hline
		\end{tabular}}
		\caption{MiDAS code construction for $(N,B) = (2,3)$, a delay of $T=4$ and rate $R = 4/9$.}
		\label{tab:MIDAS_234}
		\vspace{-2em}
\end{table}

Table~\ref{tab:MIDAS_234} illustrates a MiDAS construction for ${(N,B) = (2,3)}$ and ${T=4}$ and $\Te = T$. The encoding steps are as follows,
\begin{itemize}
\item Split each source symbol $\bs[i]$ into $k = T = 4$ sub-symbols. The first $k^u = B = 3$ sub-symbols are $\bu[i] = (u_0[i],u_1[i],u_2[i])$ while the last $k^v = T-B = 1$ sub-symbol is $v_0[i]$.
\item Apply a $(k^u+k^v,k^v,T) = (4,1,4)$ Strongly-MDS code of rate $R^v = \frac{1}{4}$ to the $v$ sub-symbols generating the $B = 3$ parity-check sub-symbols,
\begin{align}
\label{eq:pv}
\bp^v[i] = (p^v_0[i],p^v_1[i],p^v_2[i]) = \sum_{j=0}^{4} v_0[i-j] \bH^v_j.
\end{align}

\item Combine the $\bu[\cdot]$ symbols with $\bp^v[\cdot]$ symbols and generate $\bq[i] = \bp^v[i]+ \bu[i-T]$. 

\item Apply a $(k^u+k^s,k^u,T) = (5,3,4)$ Strongly-MDS code of rate $R^u = \frac{3}{5}$ to the $u$ symbols generating $k^s = \frac{N}{T-N+1}k^u = 2$ parity-check sub-symbols,
\begin{align}
\bp^u[i] = (p^u_0[i],p^u_1[i]) = \sum_{j=0}^{4} \begin{bmatrix} u_0[i-j] & u_1[i-j] & u_2[i-j] \end{bmatrix} \bH^u_j.
\end{align}.
\end{itemize}
The channel symbol at time $i$ is given by,
\begin{align}
\bx[i] = \left( \bu[i], \bv[i], \bq[i], \bp^u[i] \right),
\end{align}
whose rate is $R = \frac{k^u+k^v}{2k^u+k^v+k^s} = \frac{T}{T+B+\frac{NB}{T-N+1}} = \frac{4}{9}$.

For decoding, first assume that an erasure burst spans the interval $[i, i+2]$. We first recover $p_0^v[i+3],p_1^v[i+3], p_2^v[i+3]$ by subtracting $u_0[i-1],u_1[i-1],u_2[i-1]$ from the parity-check sub-symbols $q_0[i+3],q_1[i+3], q_2[i+3]$ respectively. In the interval $[i,i+T-1] = [i,i+3]$, the channel introduces a burst of length $3$. Thus, the $(4,1,4)$ Strongly-MDS code is capable of recovering the three erased symbols $v_0[i]$, $v_0[i+1]$ and $v_0[i+2]$ by time $i+3$ since $(1-R^v)(T) = 3$. Once all the erased $\bv[t]$ are recovered, we can compute the parity-check symbols $\bp^v[t]$ for $t \in \{i+4,i+5,i+6\}$ and subtract them from the corresponding $\bq[t]$ to recover $\bu[i],\bu[i+1],\bu[i+2]$ at time $i+4,i+5,i+6$ respectively i.e., within a delay of $T = 4$.

In the case of isolated erasures, we consider a channel introducing $N=2$ isolated erasures in the interval $[i,i+4]$ of length $T+1=5$. We first recover the unerased parity-check symbols $\bp^v[\cdot]$ in the interval $[i,i+3]$ by subtracting the corresponding $\bu[\cdot]$ symbols. The $(4,1,4)$ is then capable of recovering $v_0[i]$ by time $i+T-1 = i+3$ since $(1-R^v)T = 3 > 2 = N$. Also, $\bu[0]$ can be recovered by time $i+4$ using the $(5,3,4)$ Strongly-MDS code in the interval $[i,i+4]$ since $(1-R^u)(T+1) = 2 = N$.

\subsection{MiDAS Codes with Improved Field-Size}
\label{subsec:midas-field}

Our constructions in section~\ref{subsec:midas} are based on Strongly-MDS codes~\cite{strongly-mds-2, strongly-mds}. Such codes are guaranteed to exist only when the underlying field-sizes are very large. In particular, the field-size must increase exponentially in $\Te$ except in some special cases~\cite{strongly-mds}.
In this section we suggest an alternative construction that uses block-MDS codes instead of Strongly-MDS codes. This construction requires a field-size that only increases as $\cO(\Te^3)$. While this alternate construction also attains the tradeoff in Theorem~\ref{thm:midas}, it does come at a price. It incurs some performance loss in simulations and is less robust to non-ideal erasure patterns as discussed in Section~\ref{subsec:midas-nonideal}.

\begin{prop}
\label{prop:MiDAS-MDS}
For the channel $\cC(N,B,W)$ and delay $T$, there exists a streaming code of rate $R$ that satisfies~\eqref{b-achiev} in Theorem~\ref{thm:midas} with a field-size that increases as $\cO(\Te^3)$.
$\hfill\Box$
\end{prop}

We start by giving two examples and then discuss the general code construction. The key step is to replace the Strongly-MDS code in~\eqref{eq:pv-conc} and~\eqref{eq:pu-conc} by two block MDS codes applied diagonally to the $\bv[\cdot]$ and $\bu[\cdot]$ symbols.

\subsubsection{Example -  MiDAS $(N,B,T) = (2,3,4)$ and $W \ge T+1 = 5$}

\begin{table}[tp]
	\centering
\scriptsize{
		\begin{tabular}{c|c|c|c|c|c}
			& $[i]$ 			& $[i+1]$ 			& $[i+2]$ 			& $[i+3]$ 			& $[i+4]$ 			\\ \hline
			\multirow{3}{*}{$k^u = 3$}
			& \fboxrule=1pt \fbox{$u_{0}[i]$} 	& $u_{0}[i+1]$	& $u_{0}[i+2]$	& $u_{0}[i+3]$	& $u_{0}[i+4]$ 	\\
			& $u_{1}[i]$ 	& \fboxrule=1pt \fbox{$u_{1}[i+1]$}	& $u_{1}[i+2]$	& $u_{1}[i+3]$	& $u_{1}[i+4]$ 	\\
			& $u_{2}[i]$ 	& $u_{2}[i+1]$	& \fboxrule=1pt \fbox{$u_{2}[i+2]$}	& $u_{2}[i+3]$	& $u_{2}[i+4]$ 	\\\hline
			\multirow{1}{*}{$k^v = 1$}
			& \fcolorbox{lightgray}{lightgray}{$v_{0}[i]$} 	& $v_{0}[i+1]$	& $v_{0}[i+2]$	& $v_{0}[i+3]$	& $v_{0}[i+4]$ 	\\\hline
			\multirow{3}{*}{$k^u = 3$}
			& $u_0[i-4]+p^v_{0}[i]$	& $u_0[i-3]+$ \fcolorbox{lightgray}{lightgray}{$p^v_{0}[i+1]$}	& $u_0[i-2]+p^v_{0}[i+2]$		& $u_0[i-1]+p^v_{0}[i+3]$	& $u_0[i]+p^v_{0}[i+4]$	\\
			& $u_1[i-4]+p^v_{1}[i]$	& $u_1[i-3]+p^v_{1}[i+1]$	& $u_1[i-2]+$ \fcolorbox{lightgray}{lightgray}{$p^v_{1}[i+2]$}		& $u_1[i-1]+p^v_{1}[i+3]$	& $u_1[i]+p^v_{1}[i+4]$	\\
			& $u_2[i-4]+p^v_{2}[i]$	& $u_2[i-1]+p^v_{2}[i+1]$	& $u_2[i-2]+p^v_{2}[i+2]$		& $u_2[i-1]+$\fcolorbox{lightgray}{lightgray}{$p^v_{2}[i+3]$}	& $u_2[i]+p^v_{2}[i+4]$	\\ \hline
			\multirow{2}{*}{$k^s = 2$}
			& $p^u_{0}[i]$ 	& $p^u_{0}[i+1]$	& $p^u_{0}[i+2]$	& \fboxrule=1pt \fbox{$p^u_{0}[i+3]$}	& $p^u_{0}[i+4]$ 	\\
			& $p^u_{1}[i]$ 	& $p^u_{1}[i+1]$	& $p^u_{1}[i+2]$	& $p^u_{1}[i+3]$	& \fboxrule=1pt \fbox{$p^u_{1}[i+4]$} 	\\\hline
		\end{tabular}}
		\caption{MiDAS code construction for $(N,B) = (2,3)$, a delay of $T=4$ and rate $R = 4/9$ with a  block MDS constituent code.}
		\label{tab:MIDAS_MDS_234}
		\vspace{-2em}
\end{table}

Table~\ref{tab:MIDAS_MDS_234} illustrates a MiDAS construction using MDS as constituent codes. The rate of this code is $R = \frac{T}{T+B+\frac{NB}{T-N+1}} = \frac{4}{9}$ from~\eqref{eq:K-sub}. Note that this code has the same parameters as in Table~\ref{tab:MIDAS_234} in Section~\ref{subsec:midas}. The encoding steps, stated below, are also similar except that the Strongly-MDS codes are replaced with block MDS codes.
\begin{itemize}
\item Split each source symbol $\bs[i]$ into $k = T = 4$ sub-symbols. The first $k^u = B = 3$ sub-symbols are $\bu[i] = (u_0[i],u_1[i],u_2[i])$, while the last $k^v = T-B = 1$ sub-symbol is $v_0[i]$.
\item Apply a $(T,T-B) = (4,1)$ MDS code\footnote{This can be a simple repetition code i.e., $p^v_0[i+1] = p^v_1[i+2] = p^v_2[i+3] = v_0[i]$.} to the $v$ sub-symbols generating $B = 3$ parity-check sub-symbols $\bp^v[i] = (p^v_0[i],p^v_1[i],p^v_2[i])$. Hence, at time $i$ the generated codeword is,
\begin{align}
\bc^v[i] &= (v_0[i], p^v_0[i+1], p^v_1[i+2], p^v_2[i+3]) \label{eq:cv}
\end{align}
and is shown using the shaded boxes in Table~\ref{tab:MIDAS_MDS_234}.

\item Combine the $\bu[\cdot]$ symbols with $\bp^v[\cdot]$ symbols and generate $\bq[t] = \bp^v[t]+ \bu[t-T]$. 

\item Apply a $(T+1,T-N+1) = (5,3)$ MDS code to the $u$ symbols generating $N = 2$ parity-check sub-symbols $\bp^u[i] = (p^u_0[i],p^u_1[i])$. The codeword starting at time $i$ is given by,
\begin{align}
\bc^u[i] &= (u_0[i], u_1[i+1], u_2[i+2], p_0^u[i+3], p_1^u[i+4] ) \label{eq:cu}
\end{align}
and is marked by the unshaded boxes in Table~\ref{tab:MIDAS_MDS_234} for convenience.
\end{itemize}
The channel symbol at time $i$ is given by,
\begin{align}
\bx[i] = \left( \bu[i], \bv[i], \bq[i], \bp^u[i] \right),
\end{align}
whose rate is $R = \frac{3+1}{3+1+3+2} = \frac{4}{9}$ which is consistent with~\eqref{eq:K-sub}.

For decoding, first assume that an erasure burst spans the interval $[i, i+2]$. We first recover $p_0^v[i+3],p_1^v[i+3], p_2^v[i+3]$ at time $t=i+3$ from the parity-check symbols $q_0[i+3],q_1[i+3], q_2[i+3]$. We can use the underlying MDS codes to recover $v_0[i], v_1[i+1], v_2[i+2]$ at time $t=i+3$ by considering $\bc^v[i],\bc^v[i+1],\bc^v[i+2]$ respectively. Once all the erased $\bv[t]$ are recovered, we recover $\bu[i]$ at time $t=i+4$, $\bu[i+1]$ at time $t=i+5$ and $\bu[i+2]$ at time $t=i+6$. 

In the case of isolated erasures, we assume a channel introducing $N=2$ isolated erasures in the interval $[0,4]$ of length $T+1=5$. Note that the codeword $\bc^v[i]$ in~\eqref{eq:cv} terminates at time $t=i+3$. Thus there are no more than $N=2$ erasures on it and thus the recovery of $v_0[i]$ is guaranteed at time $t=i+3$. Likewise the codewords $\bc^u[i-2],\bc^u[i-1],\bc^u[i]$ in~\eqref{eq:cu} combining $u_2[i],u_1[i],u_0[i]$, respectively, terminate at time $t=i+4$ and there are no more than $N=2$ erasures on any of them. Thus the recovery of $u_j[i]$ for $j = 0,1,2$ is guaranteed at time $t=i+4$.

However, splitting each source symbol into $k = T$ sub-symbols is not enough in general. In particular, applying a $(T,T-B)$ MDS code to the $\bv[\cdot]$ symbols requires that the $\bv[\cdot]$ symbols are split into a multiple of $T-B$ sub-symbols. Similarly, applying a $(T+1,T-N+1)$ MDS code to the $\bu[\cdot]$ symbols requires splitting them into a multiple of $T-N+1$ sub-symbols. On the other hand, achieving the tradeoff in~\eqref{b-achiev} requires that the ratio between the size of $\bu[\cdot]$ to $\bv[\cdot]$ to be $\frac{B}{T-B}$. Thus, splitting the $\bu[\cdot]$ symbols to $B(T-N+1)$ sub-symbols and splitting the $\bv[\cdot]$ symbols into $(T-N+1)(T-B)$ sub-symbols fulfills all the previous constraints. The following example illustrates this case.

\subsubsection{Example - MiDAS $(N,B,T) = (2,3,5)$ and $W \ge T+1 = 6$}
\begin{table}[tp]
	\centering
\scriptsize{
		\begin{tabular}{c|c|c|c|c|c|c}
			& $[i]$ 				& $[i+1]$ 			& $[i+2]$ 			& $[i+3]$ 			& $[i+4]$ 		& $[i+5]$ 			\\ \hline
			\multirow{12}{*}{$k^u = 12$}
			& \fboxrule=1pt \fbox{$u_{0}[i]$} 	& $u_{0}[i+1]$	& $u_{0}[i+2]$	& $u_{0}[i+3]$	& $u_{0}[i+4]$	& $u_{0}[i+5]$ 	\\
			& $u_{1}[i]$ 		& $u_{1}[i+1]$	& $u_{1}[i+2]$	& $u_{1}[i+3]$	& $u_{1}[i+4]$	& $u_{1}[i+5]$ 	\\
			& $u_{2}[i]$ 		& $u_{2}[i+1]$	& $u_{2}[i+2]$	& $u_{2}[i+3]$	& $u_{2}[i+4]$	& $u_{2}[i+5]$ 	\\
			& $u_{3}[i]$ 		& \fboxrule=1pt \fbox{$u_{3}[i+1]$}	& $u_{3}[i+2]$	& $u_{3}[i+3]$	& $u_{3}[i+4]$	& $u_{3}[i+5]$ 	\\
			& $u_{4}[i]$ 		& $u_{4}[i+1]$	& $u_{4}[i+2]$	& $u_{4}[i+3]$	& $u_{4}[i+4]$	& $u_{4}[i+5]$ 	\\
			& $u_{5}[i]$ 		& $u_{5}[i+1]$	& $u_{5}[i+2]$	& $u_{5}[i+3]$	& $u_{5}[i+4]$	& $u_{5}[i+5]$ 	\\
			& $u_{6}[i]$ 		& $u_{6}[i+1]$	& \fboxrule=1pt\fbox{$u_{6}[i+2]$}	& $u_{6}[i+3]$	& $u_{6}[i+4]$	& $u_{6}[i+5]$ 	\\
			& $u_{7}[i]$ 		& $u_{7}[i+1]$	& $u_{7}[i+2]$	& $u_{7}[i+3]$	& $u_{7}[i+4]$	& $u_{7}[i+5]$ 	\\
			& $u_{8}[i]$ 		& $u_{8}[i+1]$	& $u_{8}[i+2]$	& $u_{8}[i+3]$	& $u_{8}[i+4]$	& $u_{8}[i+5]$ 	\\ 
			& $u_{9}[i]$ 		& $u_{9}[i+1]$	& $u_{9}[i+2]$	& \fboxrule=1pt\fbox{$u_{9}[i+3]$}	& $u_{9}[i+4]$ 	& $u_{9}[i+5]$	\\ 
			& $u_{10}[i]$ 	& $u_{10}[i+1]$	& $u_{10}[i+2]$	& $u_{10}[i+3]$	& $u_{10}[i+4]$ & $u_{10}[i+5]$	\\ 
			& $u_{11}[i]$ 	& $u_{11}[i+1]$	& $u_{11}[i+2]$	& $u_{11}[i+3]$	& $u_{11}[i+4]$ & $u_{11}[i+5]$	\\ \hline
			\multirow{8}{*}{$k^v = 8$}
			& \fcolorbox{lightgray}{lightgray}{$v_{0}[i]$} 	& $v_{0}[i+1]$	& $v_{0}[i+2]$	& $v_{0}[i+3]$	& $v_{0}[i+4]$	& $v_{0}[i+5]$ 	\\
			& $v_{1}[i]$ 	& $v_{1}[i+1]$	& $v_{1}[i+2]$	& $v_{1}[i+3]$	& $v_{1}[i+4]$	& $v_{1}[i+5]$ 	\\
			& $v_{2}[i]$ 	& $v_{2}[i+1]$	& $v_{2}[i+2]$	& $v_{2}[i+3]$	& $v_{2}[i+4]$ 	& $v_{2}[i+5]$	\\
			& $v_{3}[i]$ 	& $v_{3}[i+1]$	& $v_{3}[i+2]$	& $v_{3}[i+3]$	& $v_{3}[i+4]$ 	& $v_{3}[i+5]$	\\
			& $v_{4}[i]$ 	& \fcolorbox{lightgray}{lightgray}{$v_{4}[i+1]$}	& $v_{4}[i+2]$	& $v_{4}[i+3]$	& $v_{4}[i+4]$ 	& $v_{4}[i+5]$	\\
			& $v_{5}[i]$ 	& $v_{5}[i+1]$	& $v_{5}[i+2]$	& $v_{5}[i+3]$	& $v_{5}[i+4]$ 	& $v_{5}[i+5]$	\\
			& $v_{6}[i]$ 	& $v_{6}[i+1]$	& $v_{6}[i+2]$	& $v_{6}[i+3]$	& $v_{6}[i+4]$ 	& $v_{6}[i+5]$	\\
			& $v_{7}[i]$ 	& $v_{7}[i+1]$	& $v_{7}[i+2]$	& $v_{7}[i+3]$	& $v_{7}[i+4]$ 	& $v_{7}[i+5]$	\\	\hline
			\multirow{12}{*}{$k^u = 12$}
			& $p^v_{0}[i]$ 	& $p^v_{0}[i+1]$	& \fcolorbox{lightgray}{lightgray}{$p^v_{0}[i+2]$}	& $p^v_{0}[i+3]$	& $p^v_{0}[i+4]$	& $p^v_{0}[i+5]$ 	\\
			& $p^v_{1}[i]$ 	& $p^v_{1}[i+1]$	& $p^v_{1}[i+2]$	& $p^v_{1}[i+3]$	& $p^v_{1}[i+4]$	& $p^v_{1}[i+5]$ 	\\
			& $p^v_{2}[i]$ 	& $p^v_{2}[i+1]$	& $p^v_{2}[i+2]$	& $p^v_{2}[i+3]$	& $p^v_{2}[i+4]$ 	& $p^v_{2}[i+5]$	\\
			& $p^v_{3}[i]$ 	& $p^v_{3}[i+1]$	& $p^v_{3}[i+2]$	& $p^v_{3}[i+3]$	& $p^v_{3}[i+4]$ 	& $p^v_{3}[i+5]$	\\
			& $p^v_{4}[i]$ 	& $p^v_{4}[i+1]$	& $p^v_{4}[i+2]$	& \fcolorbox{lightgray}{lightgray}{$p^v_{4}[i+3]$}	& $p^v_{4}[i+4]$ 	& $p^v_{4}[i+5]$	\\
			& $p^v_{5}[i]$ 	& $p^v_{5}[i+1]$	& $p^v_{5}[i+2]$	& $p^v_{5}[i+3]$	& $p^v_{5}[i+4]$ 	& $p^v_{5}[i+5]$	\\
			& $p^v_{6}[i]$ 	& $p^v_{6}[i+1]$	& $p^v_{6}[i+2]$	& $p^v_{6}[i+3]$	& $p^v_{6}[i+4]$ 	& $p^v_{6}[i+5]$	\\
			& $p^v_{7}[i]$ 	& $p^v_{7}[i+1]$	& $p^v_{7}[i+2]$	& $p^v_{7}[i+3]$	& $p^v_{7}[i+4]$ 	& $p^v_{7}[i+5]$	\\
			& $p^v_{8}[i]$	& $p^v_{8}[i+1]$	& $p^v_{8}[i+2]$	& $p^v_{8}[i+3]$	& \fcolorbox{lightgray}{lightgray}{$p^v_{8}[i+4]$} 	& $p^v_{8}[i+5]$	\\
			& $p^v_{9}[i]$ 	& $p^v_{9}[i+1]$	& $p^v_{9}[i+2]$	& $p^v_{9}[i+3]$	& $p^v_{9}[i+4]$ 	& $p^v_{9}[i+5]$	\\
			& $p^v_{10}[i]$ 	& $p^v_{10}[i+1]$	& $p^v_{10}[i+2]$	& $p^v_{10}[i+3]$	& $p^v_{10}[i+4]$ 	& $p^v_{10}[i+5]$	\\
			& $p^v_{11}[i]$ 	& $p^v_{11}[i+1]$	& $p^v_{11}[i+2]$	& $p^v_{11}[i+3]$	& $p^v_{11}[i+4]$ 	& $p^v_{11}[i+5]$	\\	\hline
			\multirow{6}{*}{$k^s = 6$}
			& $p^u_{0}[i]$ 	& $p^u_{0}[i+1]$	& $p^u_{0}[i+2]$	& $p^u_{0}[i+3]$	& \fboxrule=1pt \fbox{$p^u_{0}[i+4]$} 	& $p^u_{0}[i+5]$	\\
			& $p^u_{1}[i]$ 	& $p^u_{1}[i+1]$	& $p^u_{1}[i+2]$	& $p^u_{1}[i+3]$	& $p^u_{1}[i+4]$ 	& $p^u_{1}[i+5]$	\\
			& $p^u_{2}[i]$ 	& $p^u_{2}[i+1]$	& $p^u_{2}[i+2]$	& $p^u_{2}[i+3]$	& $p^u_{2}[i+4]$ 	& $p^u_{2}[i+5]$	\\
			& $p^u_{3}[i]$ 	& $p^u_{3}[i+1]$	& $p^u_{3}[i+2]$	& $p^u_{3}[i+3]$	& $p^u_{3}[i+4]$	& \fboxrule=1pt \fbox{$p^u_{3}[i+5]$}	\\
			& $p^u_{4}[i]$ 	& $p^u_{4}[i+1]$	& $p^u_{4}[i+2]$	& $p^u_{4}[i+3]$	& $p^u_{4}[i+4]$ 	& $p^u_{4}[i+5]$	\\
			& $p^u_{5}[i]$ 	& $p^u_{5}[i+1]$	& $p^u_{5}[i+2]$	& $p^u_{5}[i+3]$	& $p^u_{5}[i+4]$ 	& $p^u_{5}[i+5]$	\\	\hline
		\end{tabular}}
		\caption{MiDAS code construction for $(N,B) = (2,3)$, a delay of $T=5$ and rate $R = 10/19$ with a block MDS constituent code. We note that each of the parity-check sub-symbols $p^v_j[t]$ is combined with $u_j[t-5]$ for $j = \{0,1,\dots,11\}$ but are omitted for simplicity.}
		\label{tab:MIDAS_235}
		\vspace{-2em}
\end{table}

Table~\ref{tab:MIDAS_235} illustrates a MiDAS construction using MDS as constituent codes. The rate of this code is $R = \frac{T}{T+B+\frac{NB}{T-N+1}} = \frac{10}{19}$. The encoding steps are as follows.
\begin{itemize}
\item Split each source symbol $\bs[i]$ into $k = (T-N+1)T = 20$ sub-symbols. The first $k^u = (T-N+1)B = 12$ of which are $(u_0[i],\dots,u_{11}[i])$ while the last $k^v = (T-N+1)(T-B) = 8$ are $(v_0[i],\dots,v_7[i])$.
\item Apply a $(T,T-B) = (5,2)$ MDS code to the $v$ sub-symbols with an interleaving factor of $T-N+1=4$. Hence, at time $i$ four codewords are generated as follows,
\begin{align}
\bc_0^v[i] &= (v_0[i], v_4[i+1], p^v_0[i+2], p^v_4[i+3], p^v_8[i+4]) \label{eq:cv-1}\\
\bc_1^v[i] &= (v_1[i], v_5[i+1], p^v_1[i+2], p^v_5[i+3], p^v_9[i+4]) \label{eq:cv-2} \\
\bc_2^v[i] &= (v_2[i], v_6[i+1], p^v_2[i+2], p^v_6[i+3], p^v_{10}[i+4]) \label{eq:cv-3} \\
\bc_3^v[i] &= (v_3[i], v_7[i+1], p^v_3[i+2], p^v_7[i+3], p^v_{11}[i+4]) \label{eq:cv-4}
\end{align}
The codeword $\bc_0^v[i]$ is shown using the shaded boxes in Table~\ref{tab:MIDAS_235}. 
According to~\eqref{eq:cv-1},~\eqref{eq:cv-2},~\eqref{eq:cv-3} and~\eqref{eq:cv-4}, $(T-N+1)B = 12$ parity-check sub-symbols are generated, namely $(p^v_0[i],\dots,p^v_{11}[i])$.

\item Combine the $\bu[\cdot]$ symbols with $\bp^v[\cdot]$ symbols and generate $\bq[t] = \bp^v[t]+ \bu[t-T]$. 
For simplicity we do not show these in Table~\ref{tab:MIDAS_235}.

\item Apply a $(T+1,T-N+1) = (6,4)$ MDS code to the $u$ symbols with an interleaving factor of $B=3$ generating $BN = 6$ parity-check sub-symbols $(p^u_0[i],\dots,p^u_5[i])$. The resulting codewords are as follows,
\begin{align}
\bc_0^u[i] &= (u_0[i], u_3[i+1], u_6[i+2], u_9[i+3], p_0^u[i+4], p_3^u[i+5] ) \label{eq:cu-1}\\
\bc_1^u[i] &= (u_1[i], u_4[i+1], u_7[i+2], u_{10}[i+3], p_1^u[i+4], p_4^u[i+5] ) \label{eq:cu-2}\\
\bc_2^u[i] &= (u_2[i], u_5[i+1], u_8[i+2], u_{11}[i+3], p_2^u[i+4], p_5^u[i+5] ) \label{eq:cu-3}
\end{align}
The codeword $\bc^u_0[i]$ is marked by the unshaded boxes in Table~\ref{tab:MIDAS_235} for convenience.
\end{itemize}
The channel symbol at time $i$ is given by,
\begin{align}
\bx[i] = \left( \bu[i], \bv[i], \bq[i], \bp^u[i] \right),
\end{align}
whose rate is $R = \frac{12+8}{12+8+12+6} = \frac{10}{19}$.

For decoding, first assume that an erasure burst spans the interval $[i, i+2]$. The decoding steps are as follows,
\begin{itemize}
\item Recover $\bp^v[t] = (p_0^v[t],\ldots, p_{11}^v[t])$ for $t = \{i+3,i+4\}$ by subtracting $\bu[t-5]$ from $\bq[t]$. 
\item Recover $\bv[i]$, $\bv[i+1]$ and $\bv[i+2]$ using the underlying $(5,2)$ MDS codes as follows. For $j \in \{0,\dots,3\}$,
\begin{itemize}
\item $\bc_j^v[i-1] = (v_j[i-1],v_{j+4}[i],p_j^v[i+1],p_{j+4}^v[i+2],p_{j+8}^v[i+3])$ has $3$ erasures at $i$, $i+1$ and $i+2$. Hence, the $v_{j+4}[i]$ sub-symbols are recovered by time $i+3$.
\item $\bc_j^v[i] = (v_j[i],v_{j+4}[i+1],p_j^v[i+2],p_{j+4}^v[i+3],p_{j+8}^v[i+4])$ has $3$ erasures at $i$, $i+1$ and $i+2$. Hence, the $v_{j}[i]$ and $v_{j+4}[i+1]$ sub-symbols are recovered by time $i+4$.
\item $\bc_j^v[i+1] = (v_j[i+1],v_{j+4}[i+2],p_j^v[i+3],p_{j+4}^v[i+4],p_{j+8}^v[i+5])$ has $3$ erasures at $i+1$, $i+2$ and $i+5$\footnote{We note that the parity-check sub-symbols $p_{j+8}^v[i+5]$ for $j \in \{0,\dots,3\}$ are counted as erasures since they are combined with $u_{j+8}[i]$ which are erased.}. Hence, the $v_{j}[i+1]$ and $v_{j+4}[i+2]$ sub-symbols are recovered by time $i+4$.
\item $\bc_j^v[i+2] = (v_j[i+2],v_{j+4}[i+3],p_j^v[i+4],p_{j+4}^v[i+5],p_{j+8}^v[i+6])$ has $3$ erasures at $i+2$, $i+5$ and $i+6$. Hence, the $v_{j}[i+2]$ sub-symbols are recovered by time $i+4$.
\end{itemize}
In other words, all the erased $\bv[\cdot]$ symbols are recovered by time $i+4$.
\item Compute the parity-check symbols $\bp^v[t]$ for $t \in \{i+5,i+6,i+7\}$ as they only combine $\bv[\cdot]$ symbols that are either unerased or recovered in the previous step. These parity-check symbols can be subtracted from the corresponding $\bp[t]$ symbols to recover $\bu[i-T]$ symbols within a delay of $T = 5$. In other words, we recover $\bu[i]$ at time $t=i+5$, $\bu[i+1]$ at time $t=i+6$ and $\bu[i+2]$ at time $t=i+7$.
\end{itemize}

In the case of isolated erasures, we assume a channel introducing $N=2$ isolated erasures in a the interval $[0,5]$ of length $T+1=6$. Note that the codewords $\bc_j^v[i]$ in~\eqref{eq:cv-1}-\eqref{eq:cv-4} terminate at time $t=i+4$. Thus there are no more than $N=2$ erasures on either of them and thus the recovery of $v_j[i]$ is guaranteed at time $i+4$. Likewise the codewords $\bc_j^u[i]$ in~\eqref{eq:cu-1}-\eqref{eq:cu-3} terminate at time $t=i+5$ and there are no more than $N=2$ erasures on any of them. Thus the recovery of $u_j[i]$ is guaranteed at time $t=i+5$.

\subsubsection{Code Construction}

The general construction achieving Prop.~\ref{prop:MiDAS-MDS} is as follows.

\begin{itemize}
\item \textbf{Source Splitting:} We assume that each source symbol $\bs[i] \in \mathbb{F}_q^k$ and partition the $k$ sub-symbols into two groups $\bu_{\mrm{vec}}[i] \in \mathbb{F}_q^{k^u}$ and $\bv_{\mrm{vec}}[i] \in \mathbb{F}_q^{k^v}$ as follows,
\begin{align}
\bs[i] = (s_0[i],\dots,s_{k-1}[i]) = (\underbrace{u_0[i],\dots,u_{k^u-1}[i]}_{\bu_{\mrm{vec}}[i]},\underbrace{v_{0}[i],\dots,v_{k^v-1}[i]}_{\bv_{\mrm{vec}}[i]})
\end{align}
where we select
\begin{align}
k^u = (\Te-N+1)B, \quad k^v =(\Te-N+1)(\Te-B).
\end{align}

\item \textbf{MDS Parity-Checks for $\bv[\cdot]$ symbols:} Construct $\Te-N+1$ systematic MDS codes of parameters $(\Te,\Te-B)$ starting at time $i$ whose associated codewords are, 
\begin{align}
\bc_j^v[i] = \left[
\begin{array}{c}
 	v_j[i] \\
 	v_{j+(\Te-N+1)}[i+1] \\
 	v_{j+2(\Te-N+1)}[i+2] \\
 	\vdots \\
	v_{j+(\Te-N+1)(\Te-B-1)}[i+\Te-B-1] \\
	p_j^v[i+\Te-B] \\
	p_{j+(\Te-N+1)}^v[i+\Te-B+1] \\
	\vdots \\
	p_{j+(\Te-N+1)(B-1)}^v[i+\Te-1]
\end{array}
\right],
\end{align}
for $j \in \{ 0,1,\dots,\Te-N \}$. Notice that each codeword $\bc_j^v[i]$ spans the interval $[i, i+\Te-1]$ and the adjacent sub-symbols have an interleaving factor of ${\Te-N+1}$. The resulting parity-check symbols at time $i$ are expressed as: $\bp^v[i] = (p^v_{0}[i],\dots,p^v_{(\Te-N+1)B-1}[i])$

\item \textbf{Repetition of $\bu[\cdot]$ symbols:} Combine the $\bu[\cdot]$ symbols with the parity-check symbols $\bp^v[\cdot]$ after applying a shift of $\Te$, i.e., $\bq[i] = \bp^v[i] + \bu[i-\Te]$.

\item \textbf{MDS Parity-Checks for $\bu[\cdot]$ symbols:} Construct $B$ systematic MDS codes of parameters $(\Te+1,\Te-N+1)$ at time $i$ whose associated codewords are,
\begin{align}
\bc_j^u[i] = \left[
\begin{array}{c}
 	u_j[i] \\
 	u_{j+B}[i+1] \\
 	u_{j+2B}[i+2] \\
 	\vdots \\
	u_{j+B(\Te-N)}[i+\Te-N] \\
	p_j^u[i+\Te-N+1] \\
	p_{j+B}^u[i+\Te-N+2] \\
	\vdots \\
	p_{j+B(N-1)}^u[i+\Te]
\end{array}
\right],
\end{align}
for $j \in \{ 0,1,\dots,B-1 \}$. Notice that each codeword $\bc_j^u[i]$ spans the interval $[i,i+\Te]$ and consists of 
sub-symbols with an interleaving factor of $B$. The resulting parity-check symbols at time $i$ are denoted by $\bp^u[i] = (p^u_{0}[i],\dots,p^u_{BN-1}[i])$.

\item \textbf{Concatenation of Parity-Checks:} Concatenate the parity-check symbols $\bp^u[\cdot]$ and $\bq[\cdot]$, i.e., the channel input at time $i$ is given by,
\begin{align}
\label{eq:MIDAS-MDS}
\bx[i] = \left(\bu[i],\bv[i],\bq[i],\bp^u[i] \right) .
\end{align}
Note that the rate of the code equals
\begin{align}
R = \frac{(\Te-N+1)\Te}{(\Te-N+1)\Te+B(\Te+1)} = \frac{\Te}{\Te + \frac{B(\Te+1)}{\Te-N+1}}
\end{align}
which is identical to the expression in~\eqref{eq:K-sub}.
\end{itemize}

The decoding steps are similar to that discussed in the previous examples and is provided in Appendix~\ref{app:MiDAS-MDS}.

\subsubsection{Field-Size Computation}
To compute the required field-size, note that splitting each source symbol into $(\Te-N+1)\Te$ sub-symbols requires that each source symbol consist of $q_1 = (\Te-N+1)\Te$ sub-symbols. We therefore need to determine the field-size of each sub-symbol. Using the well-known fact that an $(n,k)$ MDS code exists for any prime field-size greater than $n$, we note that the field-size needed for both $(\Te,\Te-B)$ and $(\Te+1,\Te-N+1)$ MDS codes to simultaneously exist is $q_2 = \Theta(\Te)$. Thus the required field-size is $q = q_1 \cdot q_2$ which is of the order $\cO(\Te^3)$.

\subsection{Non-Ideal Erasure Patterns}
\label{subsec:midas-nonideal}

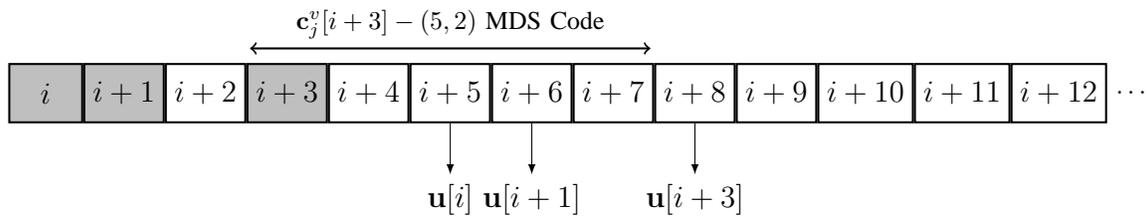
\begin{figure}[!tb]
	\centering
	\resizebox{0.85\columnwidth}{!}{
	\begin{tikzpicture}[node distance=0mm]
		\node[esym, minimum height = 8mm, minimum width = 10mm]  (x100) {$i$};
		\node[esym, minimum height = 8mm, minimum width = 10mm, right = of x100]     (x101) {$i+1$};
		\node[usym, minimum height = 8mm, minimum width = 10mm, right = of x101]     (x102) {$i+2$};
		\node[esym, minimum height = 8mm, minimum width = 10mm, right = of x102]     (x103) {$i+3$};
		\node[usym, minimum height = 8mm, minimum width = 10mm, right = of x103]     (x104) {$i+4$};
		\node[usym, minimum height = 8mm, minimum width = 10mm, right = of x104]     (x105) {$i+5$};
		\node[usym, minimum height = 8mm, minimum width = 10mm, right = of x105]     (x106) {$i+6$};
		\node[usym, minimum height = 8mm, minimum width = 10mm, right = of x106]     (x107) {$i+7$};
		\node[usym, minimum height = 8mm, minimum width = 10mm, right = of x107]     (x108) {$i+8$};
		\node[usym, minimum height = 8mm, minimum width = 10mm, right = of x108]     (x109) {$i+9$};
		\node[usym, minimum height = 8mm, minimum width = 10mm, right = of x109]     (x110) {$i+10$};
		\node[usym, minimum height = 8mm, minimum width = 10mm, right = of x110]     (x111) {$i+11$};
		\node[usym, minimum height = 8mm, minimum width = 10mm, right = of x111]     (x112) {$i+12$};
		\node[nosym, below = 2em of x105]     (x225) {$\bu[i]$};
		\node[nosym, below = 2em of x106]     (x226) {$\bu[i+1]$};
		\node[nosym, below = 2em of x108]     (x228) {$\bu[i+3]$};
		\node      [right = of x112]     (x1end) {$\cdots$};
		
		\draw[-latex] (x105.south) -| (x225.north);
		\draw[-latex] (x106.south) -| (x226.north);
		\draw[-latex] (x108.south) -| (x228.north);
		
		\dimup{x103}{x107}{2mm}{$\bc_j^v[i+3] - (5,2)$ MDS Code};
		
	\end{tikzpicture}}
	\caption{A Non-ideal erasure pattern in Section~\ref{subsec:midas-nonideal}.}
	\label{fig:MIDASPRC_NonIdeal}
\end{figure}

Even though the construction in Section~\ref{subsec:midas-field} attains the same optimal tradeoff over the deterministic erasure channel model with a smaller field-size, their performance is more sensitive compared to the construction in Section~\ref{subsec:midas} when non-ideal erasure patterns are considered. To illustrate this we focus on the case when $N=2$, $B=3$, $T=5$ and $W \ge 6$ in our discussion. The MiDAS construction with block MDS constituent code for these parameters is illustrated in Table~\ref{tab:MIDAS_235}. The MiDAS codes using strongly-MDS codes has a similar structure except that the parity checks $p_j^v[\cdot]$ and $p_j^u[\cdot]$ are generated using the strongly-MDS code.

We consider an erasure pattern that introduces a burst of length $2$ in the interval $[i,i+1]$ and an additional isolated erasure at time $i+3$. Clearly such a pattern violates a $\cC(N=2,B=3,W=6)$. Nonetheless, we argue that the MiDAS codes are able to completely recover from this erasure pattern but the alternative construction using block MDS codes in Table~\ref{tab:MIDAS_235} cannot.

In particular note that the the parity symbols $\bp[i+2]$ and $\bp[i+4]$ contribute a total of $24$ sub-symbols which suffice to recover $\bv[i],\bv[i+1]$ and $\bv[i+3]$, each of which involves $8$ sub-symbols. Thus by time $i+4$ all the erased $\bv[\cdot]$ sub-symbols are recovered and we can proceed to recover $\bu[i],\bu[i+1]$ and $\bu[i+3]$ at time $i+5,i+6$ and $i+8$, respectively, i.e., a delay of $T = 5$ symbols.

In the MiDAS construction with MDS constituent codes, illustrated in Table~\ref{tab:MIDAS_235}, we either use $\bc^u[\cdot]$ or $\bc^v[\cdot]$ codewords to recover $\bu[i]$.
\begin{itemize}
\item {Using $\bc^u[\cdot]$ codewords:} 
Here, a $(T+1,T-N+1) = (6,4)$ block MDS code is applied to each of the $\bu[\cdot]$ symbols. Each of the codewords $\bc^u_j[i]$ for $j \in \{0,1,2\}$ in~\eqref{eq:cu-1},~\eqref{eq:cu-2}, and~\eqref{eq:cu-3} has $3$ erasures at $i$, $i+1$ and $i+3$ and hence the recovery of $\bu[0]$ is impossible. 
\item {Using $\bc^v[\cdot]$ codewords:} 
Also, the $\bv[\cdot]$ symbols are protected using a $(T,T-B) = (5,2)$ MDS codes. Let us consider the codewords $\bc^v_j[i+3] = (v_j[i+3],v_{j+4}[i+4],p_j[i+5],p_{j+4}[i+6],p_{j+8}[i+7])$ for $j \in \{0,1,2,3\}$ in~\eqref{eq:cv-1},~\eqref{eq:cv-2},~\eqref{eq:cv-3} and~\eqref{eq:cv-4}. Each of these codewords has an erasure at time $i+3$ and the parity-check symbols $p_j[i+5]$ and $p_{j+4}[i+6]$ are combined with $u_j[i]$ and $u_{j+4}[i+1]$ which are erased by the channel. Thus, a total of $3$ erasures at times $i+3$, $i+5$ and $i+6$, which implies that $v_j[i+3]$ can be recovered at time $i+7$. Now, the decoder can compute $p_j[i+5]$ and $p_{j+4}[i+6]$ and subtract them from $q_j[i+5]$ and $q_{j+4}[i+6]$ to recover $u_j[i]$ and $u_j[i+1]$ with a delay of $7$ and $6$, respectively, i.e., exceeds the delay of $T=5$. 
\end{itemize}

 
We will also see performance loss from using MDS block codes instead of Strongly-MDS codes in our simulation results.


\section{Unequal Source-Channel Rates}
\label{sec:mismatch}

In this section, we study the case when the source-arrival and channel-transmission rates are unequal i.e., $M > 1$. We start by interpreting the capacity expression in Theorem~\ref{thm:capacity} in Section~\ref{subsec:expr}. In Section~\ref{subsec:mismatch-construction}, we provide the code construction achieving such capacity. The decoding analysis is discussed in Section~\ref{subsec:mismatch-decoding}. We illustrate both the encoding and decoding steps through a numerical example in Section~\ref{sec:code_examples}. We then provide the converse proof of Theorem~\ref{thm:capacity} in Section~\ref{subsec:mismatch-converse}. Finally, we present constructions that are robust against isolated erasures in Section~\ref{subsec:robust-ext}.

\subsection{Capacity Expression}
\label{subsec:expr}

We note that $C=0$, if $T < b$\footnote{Recall from~\eqref{eq:B-split} that we express $B = bM + B'$ where $B' \in [0, M-1]$.}. This follows since an erasure burst of length $B$ can span all underlying channel symbols in macro-packets $[i,i+T]$ thus making the recovery of $\bs[i]$ by macro-packet $i+T$ impossible. This trivial case will therefore not be discussed further in the paper. When $T=b$, the capacity in Theorem~\ref{thm:capacity} reduces to the following,
\begin{align}
C = \begin{cases}
\frac{1}{2}, & 0 \le B' \le \frac{M}{2},~T=b,\\
\frac{M-B'}{M}, & \frac{M}{2} < B' \le M-1,~ T=b.
\end{cases}
\label{eq:minT-cap}
\end{align}
In the special case of $T=b$, during the recovery of $\bs[i]$
we can only use the unerased symbols in $\bY[i,:]$ and $\bY[i+b,:]$ as all the intermediate macro-packets are completely erased. It turns out that a simple repetition code that uses $\min\left({M-B'}, \frac{M}{2}\right)$ information symbols and an identical number of parity-check symbols in each macro-packet achieves the capacity when ${T=b}$. 

When $T > b$ the capacity in Theorem~\ref{thm:capacity} reduces to one of the following cases,
\begin{align}
C = \begin{cases}
\frac{T}{T+b}, & 0 \le B' \le \frac{b}{T+b}M,\\
\frac{M(T+b+1)-B}{M(T+b+1)}, & \frac{b}{T+b}M < B' \le M-1.
\end{cases}\label{eq:genT-cap}
\end{align}
Note that, for each value of $b$, the capacity remains constant as $B'$ is increased in the interval
$\left[0, \frac{b}{T+b}M\right]$. 

\subsection{Code Construction}
\label{subsec:mismatch-construction}

\begin{figure}
 \begin{minipage}[t]{0.5\linewidth}
  \centering
  \includegraphics[width = \columnwidth]{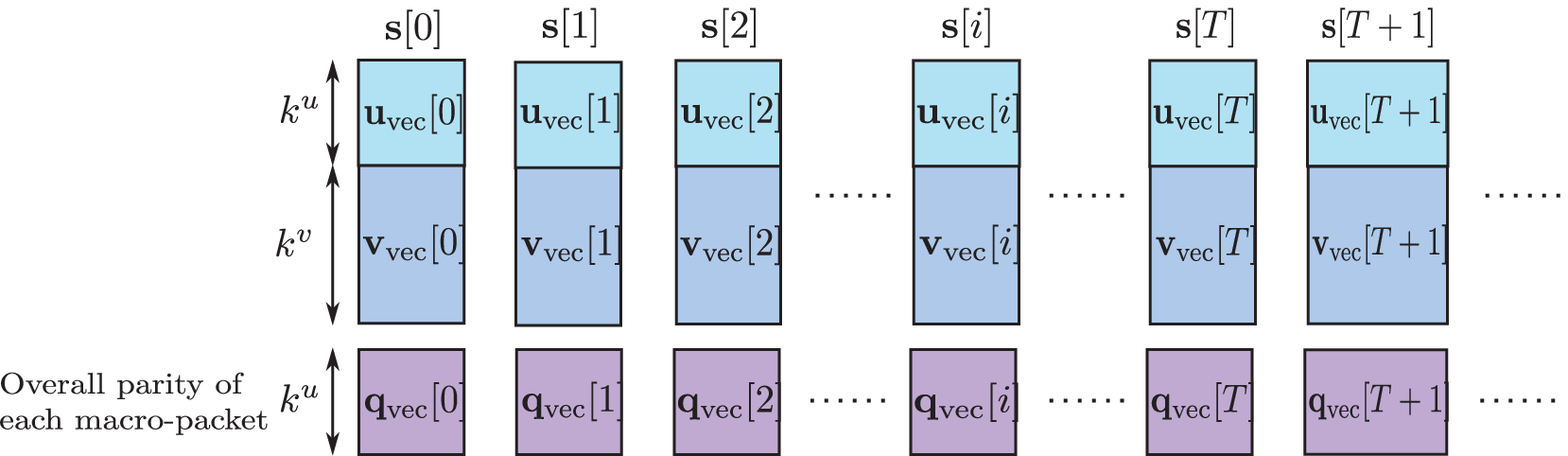}
  \caption{Construction of Parity-Check Symbols}
\label{fig:layered-1}
 \end{minipage}
 \hspace{0.5cm}
 \begin{minipage}[t]{0.5\linewidth}
  \centering
  \includegraphics[height = 1.2in]{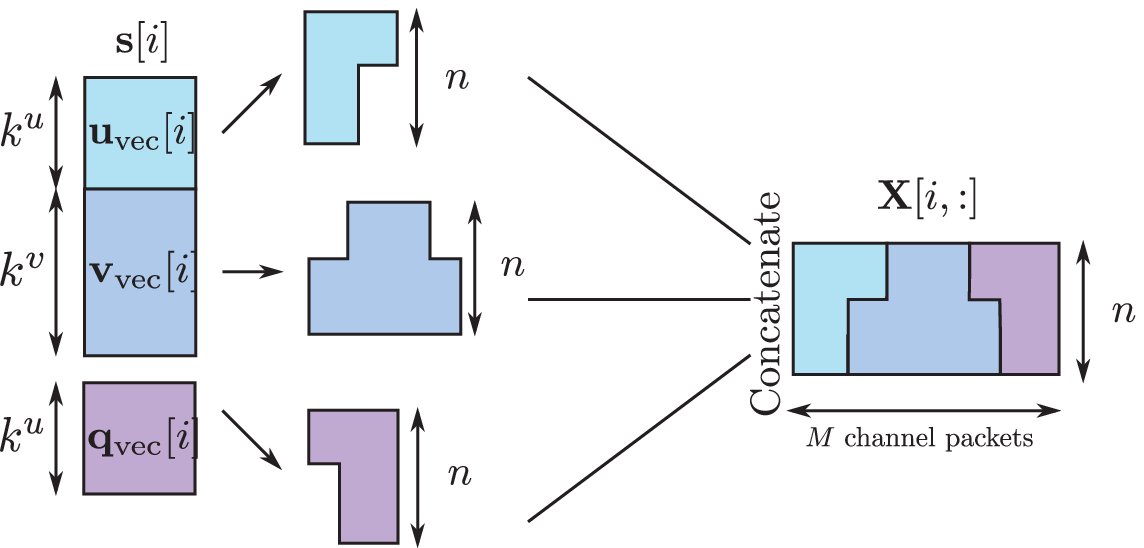}
  \caption{Reshaping of Channel Packets}
\label{fig:reshape}
 \end{minipage}
\end{figure}

The proposed construction is a natural generalization of the Generalized Maximally Short codes in section~\ref{subsec:gms}. This is illustrated in Fig.~\ref{fig:layered-1}. However an additional step of {\em reshaping} as illustrated in Fig.~\ref{fig:reshape} is needed. We separately consider three cases below.

\subsubsection{Encoding: $T \ge b$ and $B' \le \frac{b}{T+b}M$}
We let 
\begin{align}
n = T+b, \quad k = MT, \label{eq:nk-params-1}
\end{align}
throughout this case. Note that the rate $R=\frac{k}{Mn}$ reduces to the first case in both~\eqref{eq:minT-cap} and~\eqref{eq:genT-cap}.

\begin{itemize}

\item \textbf{Source Splitting:} 
We assume that each source symbol
$\bs[i] \in \mathbb{F}_q^k$ and partition the $k$ sub-symbols into two groups $\bu_{\mrm{vec}}[i] \in \mathbb{F}_q^{k^u}$ and $\bv_{\mrm{vec}}[i] \in \mathbb{F}_q^{k^v}$ as follows,
\begin{align}
\bs[i] = (s_0[i],\dots,s_{k-1}[i]) = (\underbrace{u_0[i],\dots,u_{k^u-1}[i]}_{\bu_{\mrm{vec}}[i]},\underbrace{v_{0}[i],\dots,v_{k^v-1}[i]}_{\bv_{\mrm{vec}}[i]})
\label{eq:s-part}
\end{align}
where we select
\begin{align}
k^u = Mb, \quad k^v = M(T-b). \label{eq:kuv-params-1}
\end{align}

\item \textbf{Strongly-MDS Parity-Checks:} Apply a $(k^v+k^u,k^v,T)$ Strongly-MDS code of rate $\frac{k^v}{k^v+k^u}$ to the sub-stream of $\bv_{\mrm{vec}}[\cdot]$ symbols generating $k^u$ parity-check symbols, $(p_0[i],\dots,p_{k^u-1}[i]) = \bp_{\mrm{vec}}[i] \in \mathbb{F}_q^{k^u}$ for each macro-packet. In particular we have that
\begin{align}
\bp_\mrm{vec}[i] = \left(\sum_{j=0}^T \bv_{\mrm{vec}}^\dagger[i-j] \cdot \bH_j \right)^\dagger \label{eq:q-def}
\end{align}
where $\bH_j \in \mathbb{F}_q^{k^v \times k^u}$ are the sub-matrices associated with the Strongly-MDS code~\eqref{eq:systematic}. 

\item \textbf{Parity-Check Generation:} Combine the $\bu_{\mrm{vec}}[\cdot]$ symbols with the $\bp_{\mrm{vec}}[\cdot]$ parity-checks after applying a shift of $T$ to the former i.e.,
\begin{align}
\bq_{\mrm{vec}}[i] = \bp_{\mrm{vec}}[i] + \bu_{\mrm{vec}}[i-T], \label{eq:p-def}
\end{align}
where $\bq_{\mrm{vec}}[i] \in \mathbb{F}_q^{k^u}$.

\item \textbf{Re-shaping:} In order to construct the macro-packet, we reshape $\bu_{\mrm{vec}}[i]$, $\bv_{\mrm{vec}}[i]$ and $\bq_{\mrm{vec}}[i]$ into groups each of $n$ sub-symbols generating the following matrices: 
\begin{equation}
\begin{aligned}
\label{eq:UVP}
\bU[i,:] &= \left[
\begin{array}{c|c|c|c}
 \bu[i,1] & \cdots & \bu[i,r] & \begin{array}{c}	\bu[i,r+1] \\\hline {\bf 0} \end{array} 
\end{array} \right] \in \mathbb{F}_q^{n \times r+1}\\
\bV[i,:] &= \left[
\begin{array}{c|c|c|c|c}
 \begin{array}{c} {\bf 0} \\\hline \bv[i,1] \end{array} & \bv[i,2] & \cdots & \bv[i,M-2r-1] & \begin{array}{c}	{\bf 0} \\\hline \bv[i,M-2r] \end{array} 
\end{array} \right] \in \mathbb{F}_q^{n \times M-2r} \\
\bQ[i,:] &= \left[
\begin{array}{c|c|c|c}
 \begin{array}{c}	\bq[i,r+1] \\\hline {\bf 0} \end{array} & \bq[i,r] & \cdots & \bq[i,1]
\end{array} \right] \in \mathbb{F}_q^{n \times r+1},
\end{aligned} \end{equation}
where
\begin{align}
&\begin{array}{cc}
\bu_{\mrm{vec}}[i] = \left[
\begin{array}{c}
\bu[i,1] \\
\bu[i,2] \\
\vdots \\
\bu[i,r] \\
\bu[i,r+1] \\
\end{array}
\right],
\qquad 
\bv_{\mrm{vec}}[i] = \left[
\begin{array}{c}
\bv[i,1] \\
\bv[i,2] \\
\vdots \\
\bv[i,M-2r-1] \\
\bv[i,M-2r] \\
\end{array}
\right],
\end{array}
\qquad
\bq_{\mrm{vec}}[i] = \left[
\begin{array}{c}
\bq[i,1] \\
\bq[i,2] \\
\vdots \\
\bq[i,r] \\
\bq[i,r+1] 
\end{array}\right]
\label{eq:uvp-split}
\end{align}

In~\eqref{eq:UVP} we define $r \in \mathbb{N}^0$ and $r' \in \{0,1,\ldots, n-1\}$ via \begin{align}k^u = r\cdot n+r'.\label{eq:rr'}\end{align} Note that 
$\bu[i,l] \in \mathbb{F}_q^n$ for each $l\in\{1,\ldots, r\}$ and $\bu[i, r+1] \in \mathbb{F}_q^{r'}$.
The splitting of $\bq_\mrm{vec}[i]$ into $\bq[i,j]$ in~\eqref{eq:UVP} follows in an analogous manner. We can express 
\begin{align}
\bq[i,j] = \bu[i-T, j] + \bp[i,j], \quad j=1,2,\ldots, r+1 \label{eq:pij-def}
\end{align}
where $\bp[i,j]$ is a sub-sequence of $\bp_\mrm{vec}[i]$ defined in a similar manner.
In the splitting of $\bv_\mrm{vec}[i]$ into $\bv[i,j]$ we note that $\bv[i,1], \bv[i,M-2r] \in \mathbb{F}_q^{n-r'}$ 
whereas $\bv[i,j] \in \mathbb{F}_q^n$ for $2\le j \le {M-2r-1}$. It can be easily verified that $M-2r > 0$ for our selected code parameters. When $M-2r=1$ the structure of $\bV[i,:]$ is as follows,
\begin{align}
\bV[i,:] = \left[ 
\begin{array}{c} 
{\bf 0} \\ \bv[i,1] \\ {\bf 0}
\end{array}
\right],
\end{align}
where $\bv[i,1] \in \mathbb{F}_q^{n-2r'}$.

\item \textbf{Macro-Packet Generation:}
Concatenate $\bU[i,:]$, $\bV[i,:]$ and $\bQ[i,:]$ to construct the channel macro-packet $\bX[i,:]$ as follows,
\footnotesize
\begin{align}
\label{eq:X}
&\bX[i,:] = \left[\bx[i,1]|\dots|\bx[i,M]\right] = \nonumber \\ 
&\left[
\begin{array}{c|c|c|c|c|c|c|c|c|c|c}
 \bu[i,1] & \cdots & \bu[i,r] & \begin{array}{c}	\bu[i,r+1] \\ \bv[i,1] \end{array} & \bv[i,2] & \cdots & \bv[i,M-2r-1] & \begin{array}{c} \bq[i,r+1] \\ \bv[i,M-2r] \end{array} & \bq[i,r] & \cdots & \bq[i,1]
\end{array} \right],
\end{align}
\normalsize
for $M-2r > 1$ while for the special case of $M-2r = 1$, the channel macro-packet is given by,
\begin{align}
\label{eq:Xsc}
\bX[i,:] &= \left[\bx[i,1]|\dots|\bx[i,M]\right] \nonumber \\
&= \left[
\begin{array}{c|c|c|c|c|c|c}
 \bu[i,1] & \cdots & \bu[i,r] & \begin{array}{c}	\bu[i,r+1] \\ \bv[i,1] \\ \bq[i,r+1] \end{array} & \bq[i,r] & \cdots & \bq[i,1]
\end{array} \right]
\end{align}
Note that the channel macro-packet at time $i$ is denoted by $\bX[i,:] \in \mathbb{F}_q^{n \times M}$ and the $j$-th channel packet in $\bX[i,:]$ by $\bx[i,j] \in \mathbb{F}_q^n$ for $j \in \{1,\dots,M\}$. 
\end{itemize}

\begin{remark}
We note that in the minimum delay case, i.e., $T=b$, this construction degenerates into a repetition code as $k^v = M(T-b) = 0$ in this case (cf.~\eqref{eq:kuv-params-1}). The corresponding rate of such repetition code is $R= \frac{k^u}{2k^u} = \frac{1}{2}$ which meets the capacity expression in the first case in~\eqref{eq:minT-cap}. The construction achieving the second case with $T=b$ and $B' > \frac{M}{2}$ is discussed later in this section.
\end{remark}

This completes the description of the encoding function for the first case in~\eqref{eq:capacity}. Fig.~\ref{fig:encoder} illustrates the overall encoder structure. 

\subsubsection{Encoding: $T > b$ and $B' > \frac{b}{T+b}M$}

We begin by choosing the following values of $n$ and $k$,
\begin{align}
n = T+b+1, \quad k = M(T+b+1)-B \label{eq:nk-params-2} 
\end{align}
and note that the rate $R = \frac{k}{Mn}$ reduces to the second case in~\eqref{eq:genT-cap}.

\begin{itemize}
\item Split each source $\bs[i] \in \mathbb{F}_q^k$ into $k$ sub-symbols and divide them into two groups $\bu_{\mrm{vec}}[i] \in \mathbb{F}_q^{k^u}$ and $\bv_{\mrm{vec}}[i] \in \mathbb{F}_q^{k^v}$ as in~\eqref{eq:s-part}. This time we select
\begin{align}
k^u = B=Mb+B', \quad k^v = M(T+b+1)-2B \label{eq:kuv-params-2}
\end{align}

\item Apply a $(k^v+k^u,k^v,T)$ Strongly-MDS code of rate $\frac{k^v}{k^v+k^u}$ to the sub-stream of $\bv_{\mrm{vec}}[\cdot]$ symbols generating $k^u$ parity-check symbols, $(p_0[i],\dots,p_{k^u-1}[i]) = \bp_{\mrm{vec}}[i] \in \mathbb{F}_q^{k^u}$ for each macro-packet as in~\eqref{eq:q-def}.

\item Combine the $\bu_{\mrm{vec}}[\cdot]$ symbols with the $\bp_{\mrm{vec}}[\cdot]$ parity-checks after applying a shift of $T$ to the former i.e., $\bq_{\mrm{vec}}[i] = \bp_{\mrm{vec}}[i] + \bu_{\mrm{vec}}[i-T]$.

\item Reshape the $\bu_\mrm{vec}[i],$ $\bv_\mrm{vec}[i]$ and $\bq_\mrm{vec}[i]$ vectors into matrices $\bU[i,:],$ $\bV[i,:]$
and $\bQ[i,:]$ as in~\eqref{eq:UVP}. In particular we let $r$ and $r'$ be such that $k^u = r\cdot n + r'$ as in~\eqref{eq:rr'}.
As in~\eqref{eq:uvp-split} we split $\bu_\mrm{vec}[i]$ into $\{\bu[i,j]\}_{1\le j \le (r+1)}$ where $\bu[i,j]\in {\mathbb F}_q^n$ for $1 \le j \le r$ and $\bu[i,r+1] \in {\mathbb F}_q^{r'}$ holds. In a similar manner we split $\bq_\mrm{vec}[i]$ into vectors $\{\bq[i,j]\}_{1\le j \le (r+1)}$ where $\bq[i,j] \in {\mathbb F}_q^n$ for $1 \le j \le r$ and $\bq[i,r+1] \in {\mathbb F}_q^{r'}$ holds. Finally we split $\bv_\mrm{vec}[i]$ into $\{\bv[i,j]\}_{1\le j \le {(M-2r)}}$ where $\bv[i,1], \bv[i, M-2r] \in \mathbb{F}_q^{n-r'}$ and $\bv[i,j] \in \mathbb{F}_q^n$ for $2\le j \le (M-2r-1)$.

\item Generate the Macro-Packet $\bX[i,:]$ by concatenating $\bU[i,:]$, $\bV[i,:]$ and $\bQ[i,:]$ as in~\eqref{eq:X} and~\eqref{eq:Xsc}.
\end{itemize}

\subsubsection{Encoding: $T = b$ and $B' > \frac{M}{2}$}
A simple repetition scheme is used. We split each source packet into $M-B'$ symbols i.e., $\mathbf{s}[i] = (s_0[i],\dots,s_{M-B-1'}[i])$ and assign the channel packets as follows,
\begin{align}
\mathbf{x}[i,j] = \left\{
\begin{array}{ll}
s_{j-1}[i] & j \in [1,M-B'] \\
0 & j \in [M-B'+1,B'] \\
s_{j-B'-1}[i-T] & j \in [B'+1,M].
\end{array}
\right.
\end{align}
The rate of such code is clearly $R = \frac{M-B'}{M}$ as stated in the second case in~\eqref{eq:minT-cap}.
In this case, by inspection we can check that the code described above is decodable within the decoding delay $T = b$.
Thus we will only focus on the previous two cases in our decoding analysis.

\begin{remark}
We note that the above construction achieves the capacity in Theorem~\ref{thm:capacity} for $W \ge M(T+1)$. For the case when $W < M(T+1)$, the same construction can be used with replacing the delay $T$ with the effective delay $\Te = \left\lfloor \frac{W}{M} \right\rfloor - 1$. 
\end{remark}

\begin{figure*}
\begin{center} 
\includegraphics[width=\columnwidth]{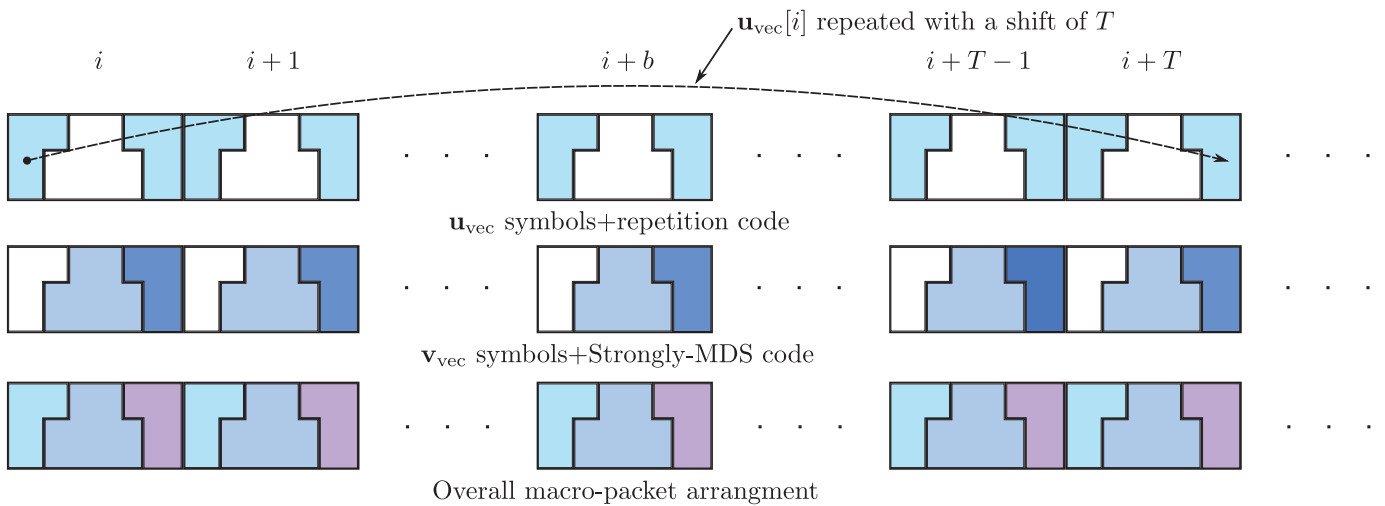}
\caption{Encoding of source symbols into macro-packets. Each source symbol is split into two groups. A repetition code is applied to one group with a delay of $T$ macro-packets as shown in the first figure. A Strongly-MDS code is applied to the second group as shown in the next figure. The resulting parity-checks are then combined as indicated in the last figure with purple parity-checks.}
\label{fig:encoder}
\end{center}
\end{figure*}

\begin{figure*}
\begin{center} 
\includegraphics[width=\columnwidth]{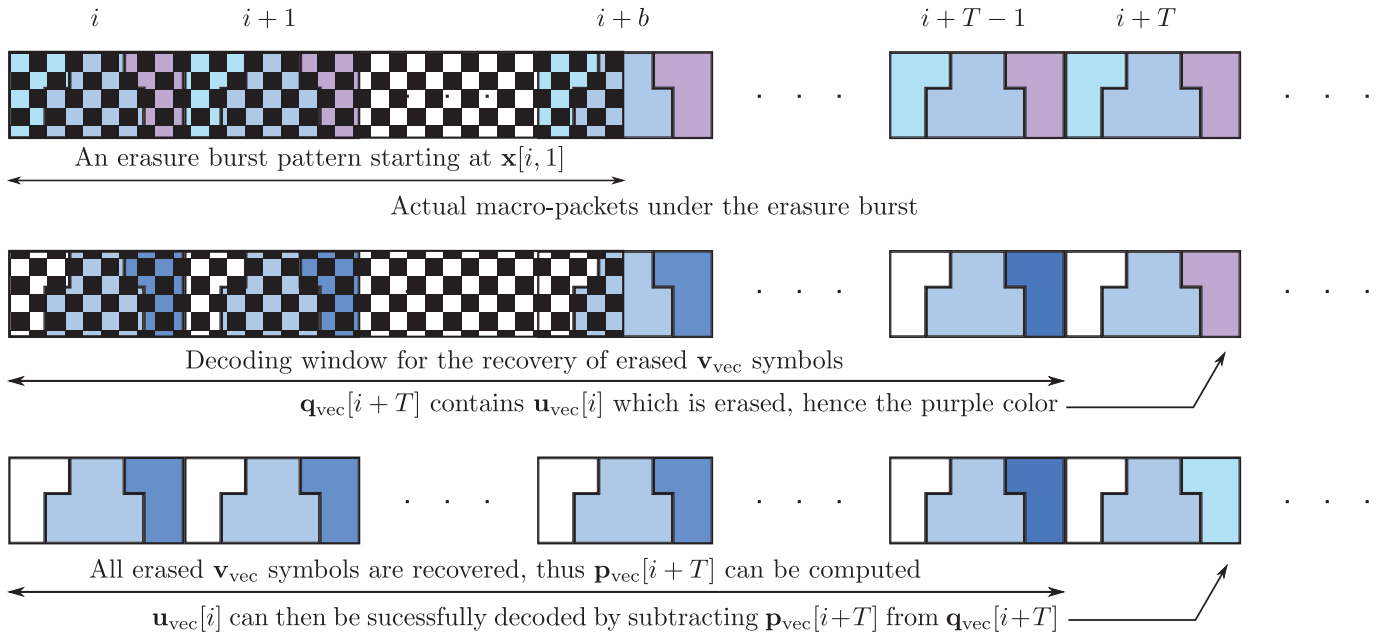}
\caption{Decoding for the burst pattern starting from $\bx[i,1]$. The first figure shows an erasure burst of length $B$. The parity-checks colored blue are used to recover the erased $\bv[\cdot]$ symbols in the second figure. The third figure shows the recovery of $\bu[i]$ using the parity-checks in macro-packet $i+T$.
\label{fig:decoder_case2}}
\end{center}
\end{figure*}

\subsection{Decoding Analysis}
\label{subsec:mismatch-decoding}

We show in this section that our proposed code construction can completely recover from any arbitrary burst of length $B$ within the deadline. Consider a channel that introduces such burst of length $B = bM + B'$ starting from $\bx[i,j]$ for $j \in \{1,\dots,M\}$. We first show how to recover $\bs[i]$ by the macro-packet $i+T$. Note that since our code is time invariant, it suffices to consider only the recovery of $\bs[i]$. Once $\bs[i]$ is recovered, we can compute $\bX[i,:]$ and repeat the same procedure with the smaller burst that starts at $\bx[i+1,1]$ to recover $\bs[i+1]$ and so on.

The decoding steps are as follows,

\begin{enumerate}
\item Step 1: In each macro-packet $\bX[t,:]$, for $t \in [i+b,i+T-1]$, recover all the unerased sub-symbols of $\bp_{\mrm{vec}}[t]$ by subtracting out $\bu_{\mrm{vec}}[t-T]$ from the corresponding $\bq_\mrm{vec}[t]$ as the former are not erased. 
Since, $\bu[i,1],\dots,\bu[i,j-1]$ are not erased, we can subtract these sub-symbols from the corresponding sub-symbols of $\bq_\mrm{vec}[i+T]$ to recover the respective $\bp_\mrm{vec}[i+T]$ sub-symbols.
\item Step 2: Recover all erased $\bv_{\mrm{vec}}[\cdot]$ sub-symbols by the macro-packet $i+T$ using the underlying $(k^u+k^v,k^v,T)$ Strongly-MDS code. This step will be justified later in the sequel.
\item Step 3: Compute $\bp_{\mrm{vec}}[i+T]$ as it combines $\bv_{\mrm{vec}}[\cdot]$ symbols that are either not erased or recovered in the previous step.
\item Step 4: Subtract $\bp_{\mrm{vec}}[i+T]$ from $\bq_{\mrm{vec}}[i+T]$ to recover $\bu_{\mrm{vec}}[i]$ within a delay of $T$ macro-packets. At this point both $\bu_{\mrm{vec}}[i]$ and $\bv_{\mrm{vec}}[i]$ have been recovered (and hence $\bs[i]$) with a delay of $T$ macro-packets as required.
\end{enumerate}

It only remains to show the sufficiency of the Strongly-MDS code in Step 2. To do that we use the following lemma.

\begin{lemma}
\label{lem:mdp_recovery_mismatch}
Consider any erasure burst of length $B$ starting at $\bx[i,j]$ for some $j \in \{1,\dots,M-r\}$. After Step 1 of cancelling $\bu_\mrm{vec}[t]$ sub-symbols, the total number of unrecovered sub-symbols in the sequence $\{(\bv_{\mrm{vec}}[t], \bp_{\mrm{vec}}[t])\}_{i \le t \le i+T}$ is at most $k^u(T+1)$.
\end{lemma}
\begin{proof}
See Appendix~\ref{app:mdp_recovery_mismatch}. 
\end{proof}
We next claim that the decoder can recover all the erased $\bv_\mrm{vec}[t]$ sub-symbols by the end of macro-packet $i+T$. To prove this, we recall that $(\bv_{\mrm{vec}}[t], \bp_{\mrm{vec}}[t])$ is a Strongly-MDS code with parameters $(k^v+k^u,k^v, T )$. We consider the following cases,
\begin{itemize}
\item If the burst starts at $j \in \{1,...,r+1\}$ then all the sub-symbols in $\{(\bv_{\mrm{vec}}[t], \bp_{\mrm{vec}}[t])\}_{i \le t \le i+b-1}$ are erased whereas a portion of the sub-symbols in $\{(\bv_{\mrm{vec}}[i+b], \bp_{\mrm{vec}}[i+b])\}$ are erased until the termination of the erasure burst. Furthermore $\{\bp_\mrm{vec}[i+T, l]\}_{j \le l \le r+1}$ are also considered to be erased since they are interfered by the erased $\bu_\mrm{vec}[i,l]$ symbols from macro-packet $i$. Note that all the erased sub-symbols involving $\bv_\mrm{vec}[t]$ will occur in a single erasure burst. Thus applying property L3 in Lemma~\ref{lem:mds-sub} with $j=T$ and $c=0$ and using $\hat{B} + \hat{I} \le k^u(T+1)=(n-k)(j+1)$, which follows from Lemma~\ref{lem:mdp_recovery_mismatch}, we are guaranteed that all the erased $\bv_\mrm{vec}[t]$ are recovered at the end of macro-packet $i+T$.
\item If the burst starts at $j \in \{r+2,\ldots, M-r\}$ then none of the sub-symbols $\bu_\mrm{vec}[i]$ are erased and can be subtracted out from $\bq_\mrm{vec}[i+T]$ to recover $\bp_\mrm{vec}[i+T]$. All the erased sub-symbols thus occur in a burst. Thus using property L2 in Lemma~\ref{lem:mds-sub}, and using $\hat{B} \le (n-k)(T+1)$ which follows from Lemma~\ref{lem:mdp_recovery_mismatch}, we are guaranteed that all the erased $\bv_\mrm{vec}[t]$ are recovered at the end of macro-packet $i+T$.
\item If $j \in \{M-r+1,\ldots, M-1\}$ then none of the sub-symbols in either $\bu_\mrm{vec}[i]$ or $\bv_\mrm{vec}[i]$ are erased. Thus we can proceed to block ${i+1}$ and apply the first step. 
\end{itemize}

Finally as mentioned in Step 4 above, once all the erased sub-symbols $\bv_\mrm{vec}[t]$ have been recovered by macro-packet ${i+T}$, their effect can be canceled and $\bu_\mrm{vec}[t]$, for $t \in \{i,i+1,\ldots, i+b\}$ can be sequentially recovered from macro-packet $t+T$ by computing and subtracting $\bp_\mrm{vec}[t+T]$ from $\bq_\mrm{vec}[t+T]$. Thus each $\bs[t] = (\bu_\mrm{vec}[t],\bv_\mrm{vec}[t])$ can be recovered by the end of macro-packet ${t+T}$. This completes the decoding analysis.

\begin{remark}
We discuss intuition on the fact that the capacity function does not decrease with $B'$ in the first case in~\eqref{eq:genT-cap} defined by $B' \le \frac{b}{T+b}M$. Recall that for this case the parameters that are selected are $k^u = Mb$ and $n=T+b$. Consider an erasure burst that starts at $\bx[i,1]$ and terminates at $\bx[i+b, B']$. We claim that for such an erasure burst, as long as $B' \le \frac{Mb}{T+b}$, only the $\bu[\cdot]$ symbols are erased in macro-packet $\bX[i+b,:]$. In particular the number of sub-symbols that are erased in marco-packet $\bX[i+b,:]$ is equal to $n B' = (T+b) B' \le Mb = k^u$. Since the $\bu[\cdot]$ symbols appears before any other symbols in each macro-packet only these symbols are erased. Thus during the recovery process, the number of parity-checks available for recovering $\bv[\cdot]$ symbols does not decrease as $B'$ is increased from $0$ to $\frac{Mb}{T+b}$. Thus the same code parameters can be used. The above argument assumes that the burst starts at the beginning of a macro-packet. In Appendix~\ref{app:mdp_recovery_mismatch}, in the proof of Lemma~\ref{lem:mdp_recovery_mismatch}, we show that this is indeed the worst case pattern. If the burst starts anywhere else, the number of available parity-checks could only increase. 
This explains why, remarkably, the capacity is not a strictly decreasing function of $B$.
\end{remark}

\subsection{Example}
\label{sec:code_examples}

\begingroup
\begin{table}[ht]
\begin{center}
\begin{tabular}{C{1.75cm}|C{2.0cm}|C{1.75cm}|C{2.0cm}|C{1.75cm}|C{2.0cm}|C{1.75cm}|C{2.0cm}}
\multicolumn{2}{c|}{$\mathbf{X}[0,:]$} & \multicolumn{2}{c|}{$\mathbf{X}[1,:]$} & \multicolumn{2}{c|}{$\mathbf{X}[2,:]$} & \multicolumn{2}{c}{$\mathbf{X}[3,:]$} \\ \hline
$\mathbf{x}[0,1]$ & $\mathbf{x}[0,2]$ & $\mathbf{x}[1,1]$ & $\mathbf{x}[1,2]$ & $\mathbf{x}[2,1]$ & $\mathbf{x}[2,2]$ & $\mathbf{x}[3,1]$ & $\mathbf{x}[3,2]$ \\ \hline
$u_0[0]$ 					& $v_2[0]$ 					& $u_0[1]$ 					& $v_2[1]$ 					& $u_0[2]$ 					& $v_2[2]$ 					& $u_0[3]$ 					& $v_2[3]$ \\
$u_1[0]$ 					& $v_3[0]$ 					& $u_1[1]$ 					& $v_3[1]$ 					& $u_1[2]$ 					& $v_3[2]$ 					& $u_1[3]$				 	& $v_3[3]$ \\
$u_2[0]$ 					& $u_0[-3]+p_0[0]$ 	& $u_2[1]$ 					& $u_0[-2]+p_0[1]$ 	& $u_2[2]$ 					& $u_0[-1]+p_0[2]$ 	& $u_2[3]$ 					& $u_0[0]+p_0[3]$\\
$v_0[0]$ 					& $u_1[-3]+p_1[0]$ 	& $v_0[1]$ 					& $u_1[-2]+p_1[1]$ 	& $v_0[2]$ 					& $u_1[-1]+p_1[2]$ 	& $v_0[3]$ 					& $u_1[0]+p_1[3]$ \\
$v_1[0]$ 					& $u_2[-3]+p_2[0]$ 	& $v_1[1]$ 					& $u_2[-2]+p_2[1]$ 	& $v_1[2]$ 					& $u_2[-1]+p_2[2]$ 	& $v_1[3]$ 					& $u_2[0]+p_2[3]$\\ \hline
\end{tabular}
\end{center}
\caption{Code construction for $(M=2,B=3,T=3)$ achieving a rate of $R = \frac{7}{10}$.}
\label{tab:example1}
\end{table}
\endgroup

In this section we show a code construction for parameters $M=2$, $B=3$, $T=3$. Note that $b=1$ and $B'=1 > \frac{1}{2} = \frac{b}{T+b}M$. Thus, the capacity is given by $C= \frac{M(T+b+1)-B}{M(T+b+1)} = \frac{7}{10}$ which can be achieved using the code illustrated in Table~\ref{tab:example1}.

\subsubsection*{Encoding}
\begin{enumerate}
	\item Split each source symbol into $M(T+b+1)-B = 7$ symbols i.e., $\bs[i] = (s_0[i],\cdots,s_6[i])$.
	\item Divide these into two groups $\bu_\mrm{vec}[i]$ and $\bv_\mrm{vec}[i]$ with $k^u = B = 3$ and $k^v = M(T+b+1)-2B = 4$ sub-symbols, respectively, as in~\eqref{eq:s-part}. We let
	$\bu_\mrm{vec}[i]=(u_0[i],\cdots,u_2[i])=(s_0[i],\cdots,s_2[i])$ and $\bv_\mrm{vec}[i] = (v_0[i],\cdots,v_3[i]) = (s_3[i],\ldots, s_6[i])$. 
	\item	We place $B=3$ parity symbols $\bq_\mrm{vec}[i]=(q_0[i],q_1[i],q_2[i])$ into the last channel packet of each macro-packet. These parities consist of two components, $\bq_\mrm{vec}[i]$=$\bp_\mrm{vec}[i]$+$\bu_\mrm{vec}[i-3]$. The parity symbols $\bp[i]$ are generated using a Strongly-MDS code.
\end{enumerate}

\subsubsection*{Decoding}
Since $M=2$, there are two burst patterns that need to be checked.
\begin{enumerate}
	\item Burst that erases $\bx[0,1]$, $\bx[0,2]$ and $\bx[1,1]$.
		
		\textit{Recovery of $v$ symbols:} We first subtract $\bu_\mrm{vec}[t-T]$ from $\bq_\mrm{vec}[t]$ for $t = \{1,2\}$ to recover the corresponding $\bp_\mrm{vec}[t]$. These are a total of $6$ sub-symbols and thus can be used to recover $v_0[0],\cdots,v_3[0]$ as well as $v_0[1],v_1[1]$. In other words, all erased $v$ sub-symbols are recovered by the end of the macro-packet $\bX[2,:]$.
		
		\textit{Recovery of $u$ symbols:} With all the erased $v$ symbols now recovered, we can compute the $\bp_\mrm{vec}[t]$ symbols for $t = \{3,4\}$ and subtract them from $\bq_\mrm{vec}[t]$ to recover $\bu_\mrm{vec}[0]$ and $\bu_\mrm{vec}[1]$ at their respective deadlines.
	
	\item Burst that erases $\bx[0,2], \bx[1,1], \bx[1,2]$.
		
		\textit{Recovery of $v$ symbols:} Since $\bu_\mrm{vec}[0]$ is not erased, we can subtract it from $\bq_\mrm{vec}[3]$ to recover $\bp_\mrm{vec}[3]$. This together with $\bp_\mrm{vec}[2]$ is a total of $6$ sub-symbols. Thus, they can be used to recover the erased $v$ symbols $(v_2[0],v_3[0])$ and $(v_0[1],\cdots,v_3[1])$. 
		
		\textit{Recovery of $u$ symbols:} Similar to the previous burst pattern, we compute the value of the parity-check symbols $\bp_\mrm{vec}[4]$ and subtract it from $\bq_\mrm{vec}[4]$ to recover $\mathbf{u}[1]$ by its deadline.
\end{enumerate}

\subsection{Converse}
\label{subsec:mismatch-converse}
In order to establish the converse we first consider the case when $T>b$. We show that any feasible rate satisfies
\begin{align}
R \le R^+=\min\left(\frac {M(T+b+1)-B}{M(T+b+1)}, \frac{T}{T+b}\right). \label{eq:conv-bnd}
\end{align}

\begin{figure*}
\begin{center} 
\includegraphics[width=\columnwidth]{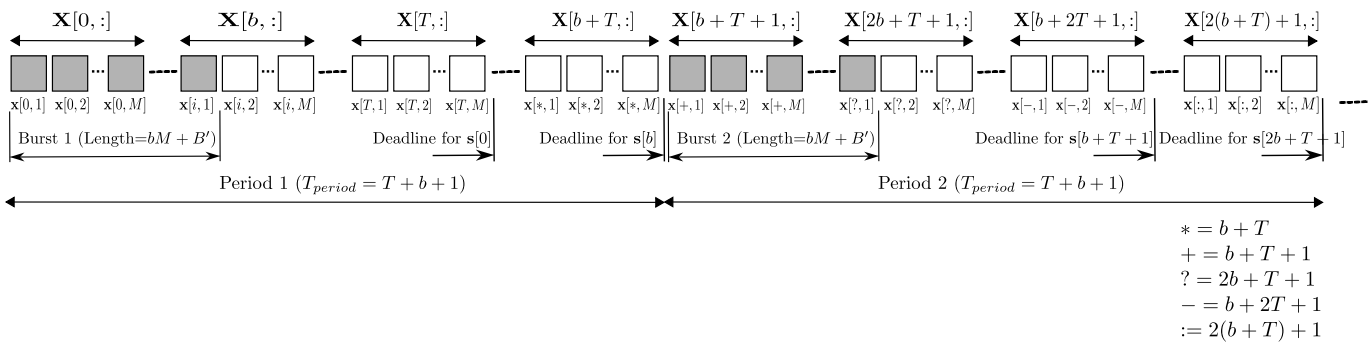}
\caption{Periodic erasure patterns\label{fig:periodic_burst}}
\end{center}
\end{figure*}


Consider a periodic erasure channel as shown in Figure \ref{fig:periodic_burst}. Each period consists of 
$\tau_P=T+b+1$ macro-packets. In each such period the first $B$ channel packets are erased and the subsequent $M(b+T+1)-B$ are not. Consider the first period with the burst starting at $\bx[0,1]$. By definition we require that $\bs[0]$ be recovered by the end of macro-packet $T$, $\bs[1]$ by macro-packet $T+1$ and likewise the last erased source symbol $\bs[b]$ by macro-packet $T+b$. Thus all the lost source packets are recovered by macro-packet $t=T+b$. Once these erased packets are recovered, we can treat these erasures as having never happened and simply repeat the argument for the next period and so on. Therefore our proposed streaming code must be a feasible code for the periodic erasure channel. Since the capacity of the erasure channel is simply the fraction of the non-erased channel symbols, it follows that
\begin{equation}
\label{eq:ub-pec}
R^+= \frac {M(T+b+1)-(bM+B')}{M(T+b+1)}.
\end{equation}
is an upper bound on the rate of any feasible streaming code.

To establish the other inequality in~\eqref{eq:conv-bnd} we consider a periodic erasure channel consisting of $\tau_P =T+b$ macro-packets and assume that in each period the first $\hat{B} =Mb \le B$ channel symbols are erased. Thus in the proposed channel first $b$ macro-packets are completely erased in each period and the remaining $T$ macro-packets are not erased. In particular in the first period, $\bs[0],\ldots, \bs[b-1]$ must be recovered at the end of macro-packets $T,\ldots, T+b-1$ respectively. At this point all the erased source symbols have been recovered and we can proceed to the recovery of the second burst starting at macro-packet $T+b$. Thus the streaming code must also be feasible on this erasure channel whose capacity is clearly $\frac{T}{T+b}$, and thus the upper bound follows.

When $T=b$ we show that
\begin{align}
C \le \min\left(\frac{M-B'}{M},\frac{1}{2} \right) \label{eq:conv-bnd-2}.
\end{align}
When $B' \le M/2$, the second condition $C \le \frac{1}{2}$ dominates. This bound immediately follows from~\eqref{eq:conv-bnd} by substituting $T=b$ in the second expression in~\eqref{eq:conv-bnd}. 
Thus we only need to show that when $B' > \frac{M}{2}$ and $T = b$ the upper bound $C \le \frac{M-B'}{M}$ is valid.

We start by considering a channel that erases the first $B = bM + B'$ channel packets $\mathbf{x}[i,1],\dots,\mathbf{x}[i+b,B']$. Since the delay constraint for $\mathbf{s}[i]$ is $i+T = i+b$, the following equation should be satisfied,
\begin{align}
\label{eq:Tmin-converse1}
& \mathsf{H}(\mathbf{s}[i] \big| \mathbf{x}[i+b,B'+1],\dots,\mathbf{x}[i+b,M]) = 0 \nonumber \\
& \Rightarrow \mathsf{H}(\bs) \le (M-B') \mathsf{H}(\bx),
\end{align}
which implies that $R = \frac{H(\rvbs)}{MH(\rvbx)} \le \frac{M-B'}{M}$ as required. This completes the proof of the upper bound. 

\subsection{Robust Extensions}
\label{subsec:robust-ext}
In Section~\ref{subsec:mismatch-construction}, we provided capacity achieving codes for $\cC(1,B,W \ge M(T+1))$. 
In order to extend the codes for channels with $N > 1$, we apply the approach used in the MiDAS construction in Section~\ref{subsec:midas}. In particular we construct an optimal burst erasure code and then append additional parity-checks for the $\bu[\cdot]$ symbols to deal with isolated losses. In particular we extend the macro-packet construction
in~\eqref{eq:X} as follows,
\footnotesize
\begin{align}
&\bX[i,:] = \left[\bx[i,1]|\dots|\bx[i,M]\right] =\nonumber \\
&\left[
\begin{array}{c|c|c|c|c|c|c|c|c}
 \begin{array}{c} \multirow{2}{*}{$\bu[i,1]$} \\ \\ \bp^u[i,1] \end{array} & \cdots & \begin{array}{c} \multirow{2}{*}{$\bu[i,r]$} \\ \\ \bp^u[i,r] \end{array} & \begin{array}{c}	\bu[i,r+1] \\ \bv[i,1] \\ \bp^u[i,r+1] \end{array} & \cdots & \begin{array}{c} \bq[i,r+1] \\ \bv[i,M-2r] \\ \bp^u[i,M-r] \end{array} & \begin{array}{c} \multirow{2}{*}{$\bq[i,r]$} \\ \\ \bp^u[i,M-r+1] \end{array} & \cdots & \begin{array}{c} \multirow{2}{*}{$\bq[i,1]$} \\ \\ \bp^u[i,M] \end{array}
\end{array} \right], \label{eq:Xrobust}
\end{align}
\normalsize
where $\bu[i,j],$ $\bv[i,j]$ and $\bq[i,j]$ are symbols obtained from the optimal code for the $\cC(N=1,B,W)$ channel.
We apply another $(k^u+M k^s,k^u,T)$ Strongly-MDS code to the $\bu_{\mrm{vec}}[\cdot]$ symbols generating $M k^s$ parity-check sub-symbols $(p^u_1[i],\dots,p^u_{M k^s}[i]) = \bp^u_{\mrm{vec}}[i] \in \mathbb{F}_q^{M k^s}$. We then concatenate the generated parities after splitting them into $M$ equal groups to each channel packet, $\bp^u_\mrm{vec}[i] = (\bp^u[i,1],\dots,\bp^u[i,M])$ as shown in~\eqref{eq:Xrobust}. The corresponding rate of such code is clearly $R = \frac{k^u+k^v}{M(n+k^s)}$, where $k^u$, $k^v$ and $n$ are based on the optimal code for the burst-only channel.

\begin{prop}
\label{prop:robust}
To recover from any $N \le \left\lfloor \frac{T}{T+b}Mb \right\rfloor$ isolated erasures when $W \ge M(T+1)$ and $T>b$, it suffices to select
\begin{align}
\label{eq:ks}
k^s = \left\lceil\frac{Nn}{M(T+1)-N} \right\rceil.
\end{align}
where $\lceil\cdot\rceil$ and $\lfloor\cdot\rfloor$ denote the ceil and floor functions respectively.
\end{prop}

\begin{proof}
We recall that there are two Strongly-MDS codes underlying our construction in~\eqref{eq:Xrobust}. A $(k^u + k^v, k^v, T)$ Strongly-MDS code is applied to $\bv_\mrm{vec}[\cdot]$ symbols to generate parity-checks $\bp_\mrm{vec}[\cdot]$ and $\bq_\mrm{vec}[t] = \bp_\mrm{vec}[t] + \bu_\mrm{vec}[{t-T}]$ are transmitted. Furthermore a $(k^u+ M k^s, k^u, T)$ Strongly-MDS code is applied to the $\bu_\mrm{vec}[\cdot]$ symbols to generate parity-checks $\bp^u_\mrm{vec}[\cdot]$. 
We recall that there are two Strongly-MDS codes underlying our construction in~\eqref{eq:Xrobust}. A $(k^u + k^v, k^v, T)$ Strongly-MDS code is applied to $\bv_\mrm{vec}[\cdot]$ symbols to generate parity-checks $\bp_\mrm{vec}[\cdot]$ and $\bq_\mrm{vec}[t] = \bp_\mrm{vec}[t] + \bu_\mrm{vec}[{t-T}]$ are transmitted. Furthermore a $(k^u+ M k^s, k^u, T)$ Strongly-MDS code is applied to the $\bu_\mrm{vec}[\cdot]$ symbols to generate parity-checks $\bp^u_\mrm{vec}[\cdot]$. 

Let us consider the window of length $T$ consisting of the macro-packets $\bX[i,:],\dots,\bX[i+T-1,:]$ and assume that there are $N$ erasures in arbitrary positions. Note that in $\bq_{\mrm{vec}}[t] = \bp_{\mrm{vec}}[t] + \bu_{\mrm{vec}}[t-T]$ for $t \in [i,i+T-1]$, the $\bu_{\mrm{vec}}[\cdot]$ are from time $i-1$ or before, and can be canceled to recover $\bp_{\mrm{vec}}[t]$. The $(k^u+k^v,k^v,T)$ Strongly-MDS code can recover $\bv_{\mrm{vec}}[i]$ if no more than $k^uT$ sub-symbols are erased among $(\bv_{\mrm{vec}}[i],\bq_{\mrm{vec}}[i],\dots,\bv_{\mrm{vec}}[i+T-1],\bq_{\mrm{vec}}[i+T-1])$. Since these sub-symbols are reshaped into columns each having no more than $n$ sub-symbols, the number of erasures that are guaranteed to be corrected is given by,

\begin{align}
N^v &= \left\lfloor \frac{k^uT}{n} \right\rfloor \nonumber \\
&\ge \min \left( \left\lfloor \frac{(bM+B')T}{T+b+1} \right\rfloor\bigg|_{B' \ge \frac{b}{T+b}M} , \left\lfloor \frac{MbT}{T+b} \right\rfloor\bigg|_{B' < \frac{b}{T+b}M} \right) \label{eq:Nv2}\\
&= \left\lfloor \frac{MbT}{T+b} \right\rfloor \label{eq:Nv3}\\
&\ge N, \label{eq:Nv}
\end{align}

where we use
\begin{align}
(k^u,n) = 
\begin{cases}
(B,T+b+1), & B' \ge \frac{b}{T+b}M \\
(Mb,T+b), & B' < \frac{b}{T+b}M \\
\end{cases}
\end{align}
to get~\eqref{eq:Nv2} and substitute for $B' \ge \frac{b}{T+b}M$ in the first term in~\eqref{eq:Nv2} to get~\eqref{eq:Nv3}.

Next we consider the number of erased symbols that can be corrected by the $(k^u+M k^s,k^u,T)$ Strongly-MDS code. Using Lemma~\ref{lem:mds-sub}, one can see that this code can recover from $Mk^s(T+1)$ erasures in the window of interest. Since each channel input can have up to $n+k^s$ sub-symbols belonging to this code, the total number of erasures that can be corrected is given by, 
\begin{align}
\label{eq:Nu}
N^u = \left\lfloor \frac{M k^s(T+1)}{n + k^s} \right\rfloor
\end{align}
which upon re-arranging gives~\eqref{eq:ks}.
\end{proof}

\begin{remark}
Unlike the case of MiDAS codes, we do not claim the optimality of the proposed robust codes. Nevertheless in the simulation results we observe that in some cases these codes outperform baseline schemes.
\end{remark}

\section{Numerical Comparisons}
\label{subsec:numerical}


\begin{table}[!htdp]
\begin{center}
\begin{tabular}{c|c|c}
Code & $N$ & $ B$ \\\hline\hline
Strongly-MDS Codes & $(1-R)(T+1)$ &$(1-R)(T+1)$ \\\hline
Maximally Short Codes & $1$ & $T\cdot\min\left(\frac{1}{R}-1, 1\right)$ \\\hline
 \!\!MiDAS Codes \!\! & $\min\left(B, T\!- \!\frac{R}{1-R}B \right) $& $B \in [1, T]$ \\\hline
E-RLC Codes~\cite{erlc-infocom} & & \\
 ${\Delta \in [R(T+1), T-1]},$ &$\frac{1-R}{R}(T-\Delta)+1$ & $\frac{1-R}{R}\Delta$ \\
 ${R \ge 1/2}$ & & \\\hline
\end{tabular}
\end{center}
\caption{Achievable $(N,B)$ for channel $\cC(N,B,W \ge T+1)$ for equal source-channel rates. Similar Tradeoffs for the first three codes can be achieved for $W < T+1$ by replacing $T$ with $W-1$.}
\label{Sco:Codes}
\end{table}%

\begin{figure}
 \begin{minipage}[t]{0.475\linewidth}
  \centering
  \includegraphics[width=\columnwidth]{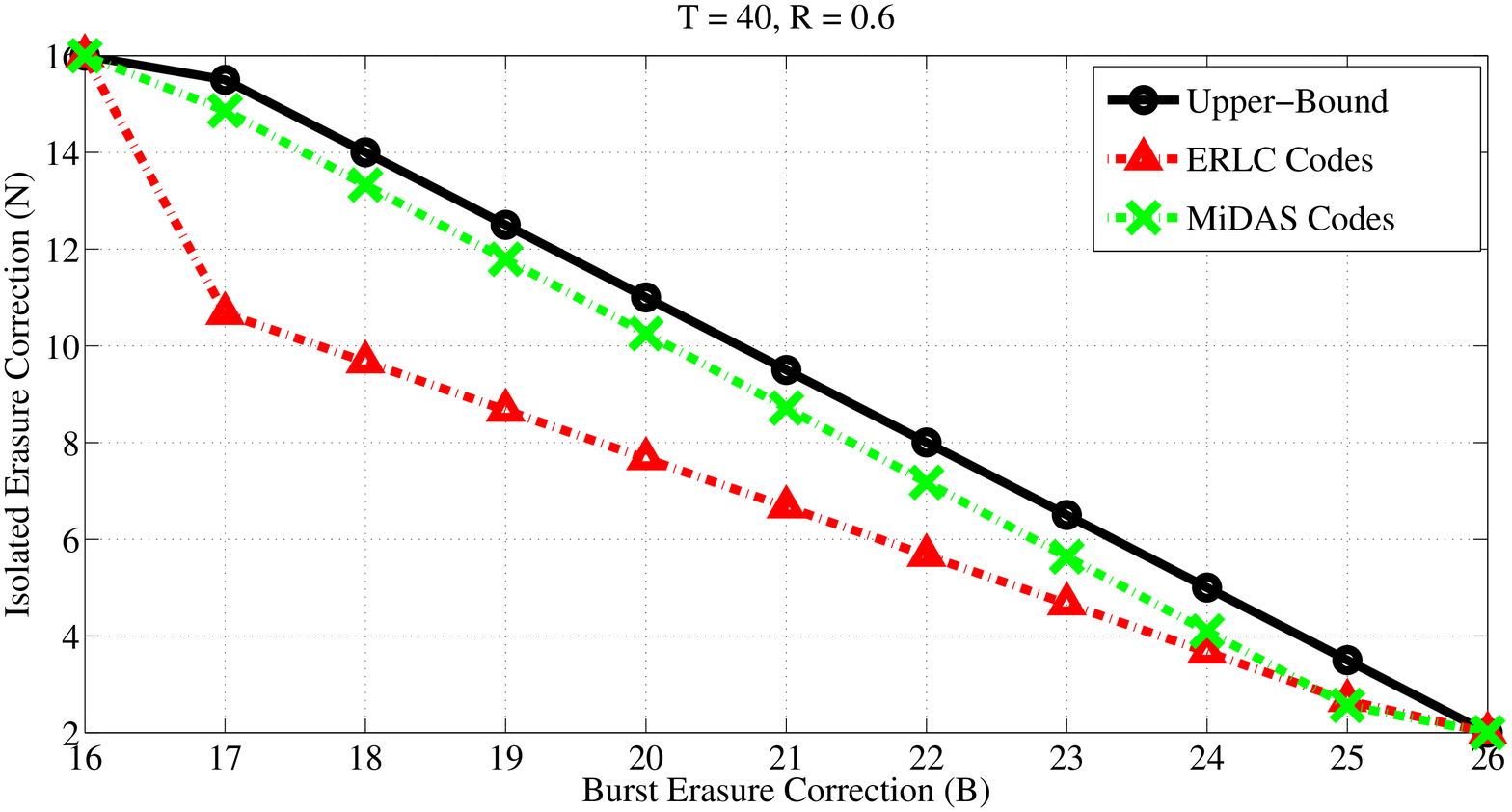}
  \caption{Achievable tradeoff between $N$ and $B$ for equal source-channel rates. The rate is fixed to $R = 0.6$ and the delay is fixed to $T=40$ and $W = T+1$. The uppermost curve (solid black lines with `o') is the upper bound in~\eqref{eq:r-ub}. The MiDAS codes are shown with broken green lines with `$\times$' and are very close to the upper bound. The E-RLC codes in~\cite{erlc-infocom} are shown with broken red lines with `$\triangle$'.}
\label{fig:cTdT_Tradeoff}
 \end{minipage}
 \hspace{0.05\linewidth}
 \begin{minipage}[t]{0.475\linewidth}
  \centering
  \includegraphics[width=\columnwidth]{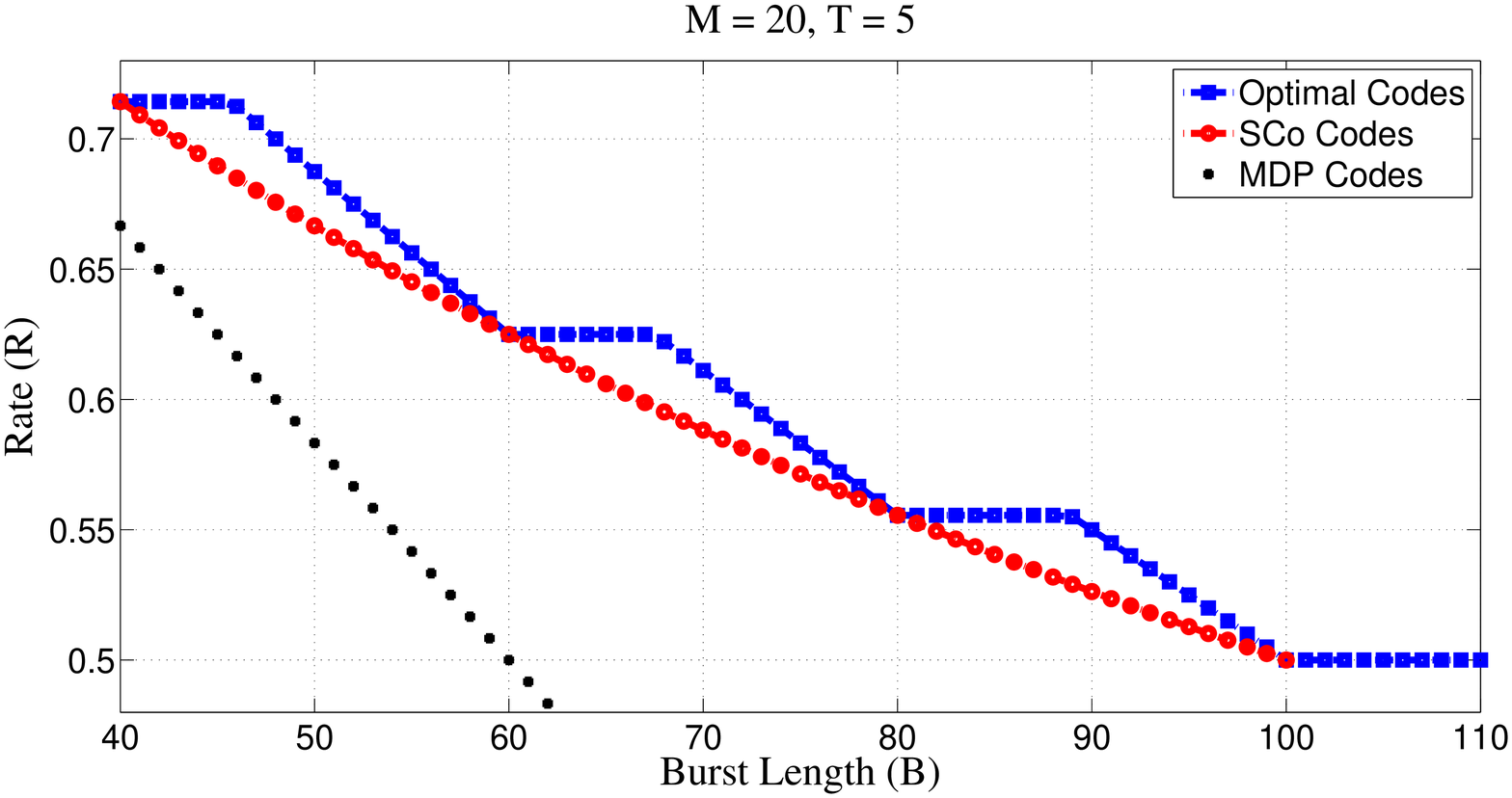}
  \caption{Achievable rates for different code constructions for the case of unequal source-channel rates for the $\cC(N=1,B, W= M(T+1))$ channel. We fix the delay to $T=5$ macro-packets and let $M=20$. The blue plot (marked with squares) corresponds to the capacity in Theorem~\ref{thm:capacity}. The red curve (marked with circles) corresponds to the rate achieved by the adapted MS code~\eqref{eq:R-SCO} whereas the black line corresponds to the rate of the Strongly-MDS code~\eqref{eq:r-MDS-mismatch}.}
\label{fig:RvsB_T5_M20}
 \end{minipage}
\end{figure}

\subsection{Equal Source-Channel Rates}
Table~\ref{Sco:Codes} summarizes the feasible values of $N$ and $B$ for different codes\footnote{We note that the floor of the values given in Table~\ref{Sco:Codes} should be considered as the values might not be integers}. For a fixed rate $R$ and delay $T$ we indicate the values of $N$ and $B$ achieved by various codes in the case of equal source-channel rates.
The first row corresponds to the Strongly-MDS in Section~\ref{subsec:strongly-mds}, while the second row corresponds to the MS codes in Section~\ref{subsec:maximally-short}. The third row corresponds to our proposed construction --- MiDAS codes --- in Theorem~\ref{thm:midas}. In contrast to the Strongly-MDS codes and MS codes, that only attain specific values of $N$ and $B$, the family of MiDAS codes can attain a range of $(N,B)$ for a given $R$ and $T$. The last row corresponds to another family of codes -- Embedded Random Linear Codes (E-RLC) -- proposed in~\cite{erlc-infocom}. While such constructions are optimal for $R=1/2,$ they are far from optimal in general and will not be discussed in this paper. 

We further numerically illustrate the achievable $(N,B)$ pairs for various codes in Fig.~\ref{fig:cTdT_Tradeoff}. We fix the rate to $R = 0.6$. As stated before, the Strongly-MDS and MS codes in sections~\ref{subsec:strongly-mds} and~\ref{subsec:maximally-short} respectively only achieve the extreme points on the tradeoff. The MiDAS codes achieve a tradeoff, very close to the upper bound for all rates. The E-RLC codes, illustrated with the red plot, are generally far from optimal except for $R=0.5$ which is not the case in this figure. 


\subsection{Unequal Source-Channel Rates}
Fig.~\ref{fig:RvsB_T5_M20} illustrates the capacity and rates achieved with baseline schemes for the case of unequal source-channel rates. In this example, we consider $M=20$ and a delay of $T=5$ macro-packets and plot the rate vs.\ correctable burst length. The capacity is shown by the blue-curve marked with squares. Note that it is constant in the intervals $B \in [40,45], [60, 67], [80, 88], [100, 110],$ which corresponds to the first case in~\eqref{eq:capacity}. The red curve marked with circles denotes the rate achieved by a suitable modification of the MS code~\eqref{eq:R-SCO}. We note that the curves intersect whenever $B$ is an integer multiple of $M$, indicating the optimality of the MS codes for these special values: $B \in \{40,60,80,100\}$. Furthermore for burst lengths $B > MT = 100$, the MS codes are no longer feasible and the associated rate is zero. The dotted black line shows the performance of the Strongly-MDS codes in~\eqref{eq:r-MDS-mismatch}. Since these codes do not perform sequential recovery, their achievable rate is significantly lower than the capacity.

\section{Simulation Results}
\label{sec:simulation}


\begin{figure}
\centering
 \begin{minipage}[t]{0.33\linewidth}
  \centering
  \includegraphics[scale=0.5]{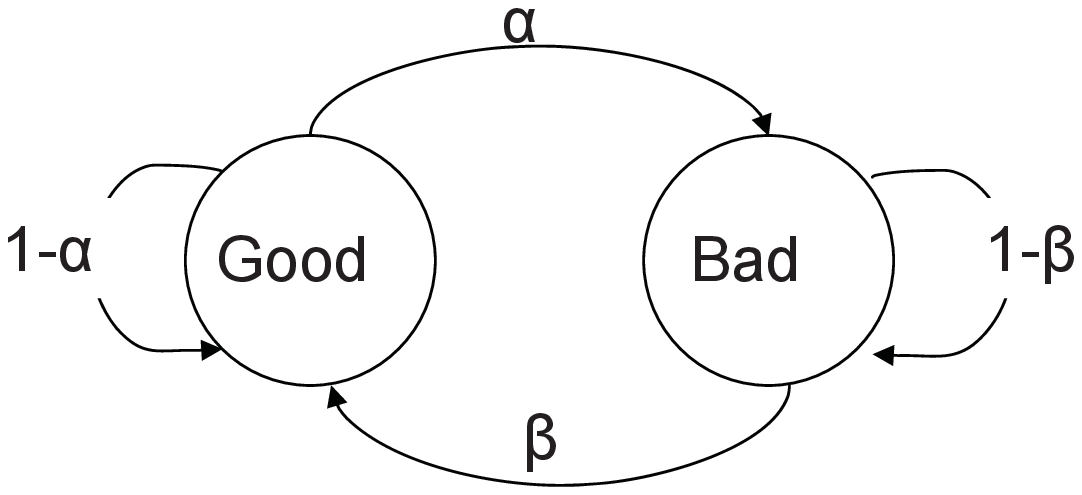}
  \caption[]{Gilbert-Elliott Channel Model}
\label{fig:GE}
 \end{minipage}
 \hspace{0.1cm}
	\centering
 \begin{minipage}[t]{0.65\linewidth}
  \centering
  \includegraphics[scale=0.45]{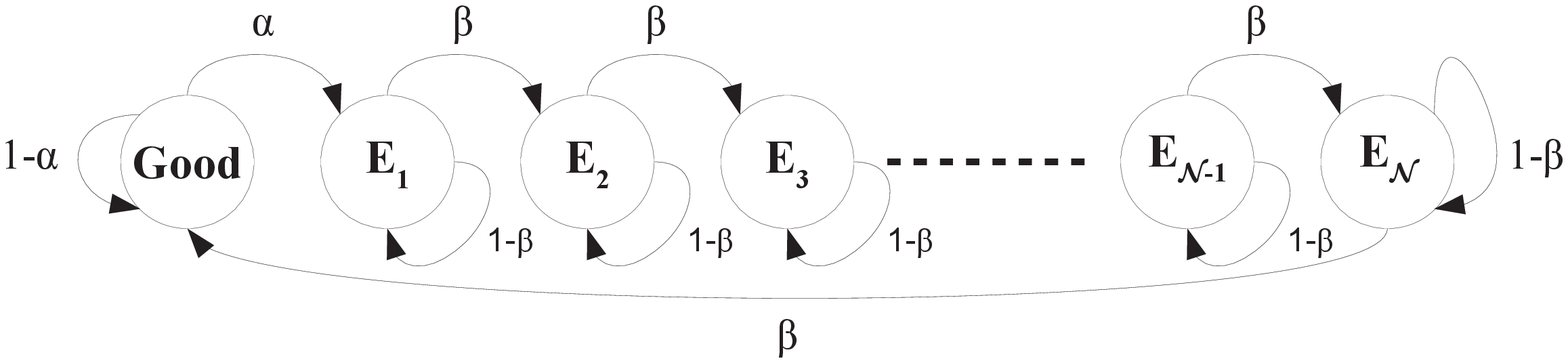}
		\captionsetup{justification=centering}
  \caption{Fritchman Channel Model}
\label{fig:fritchman}
 \end{minipage}
\end{figure}

In this section we study the validity of our proposed code constructions over statistical channel models. We consider two classes of channels that introduce both burst and isolated erasures.
A Gilbert-Elliott channel is a two-state Markov model. In the ``good state" each channel packet is lost with a probability of $\eps$ whereas in the ``bad state" each channel packet is lost with a probability of $1$. We note that the average loss rate of the Gilbert-Elliott channel is given by
\begin{align}
\Pr(\cE) = \frac{\beta}{\beta+\al}\eps + \frac{\al}{\al+\beta} \label{eq:loss-uncoded}.
\end{align}
where $\al$ and $\beta$ denote the transition probability from the good state to the bad state and vice versa. 
As long as the channel stays in the bad state the channel behaves as a burst-erasure
channel. The length of each burst is a Geometric random variable with mean of $\frac{1}{\beta}$.
When the channel is in the good state it behaves as an i.i.d.\ erasure channel with an erasure
probability of $\eps$. The gap between two successive bursts is also a geometric random variable
with a mean of $\frac{1}{\al}$. Finally note that $\eps=0$ results in a Gilbert Channel~\cite{gilbert}, which only results in burst losses.

Fig.~\ref{fig:fritchman} shows a Fritchman channel model~\cite{fritchman} with a total of ${\mathcal{N}+1}$ states. One of the states is the good state and the remaining $\mathcal{N}$ states are bad states.
We again let the transition probability from the good state to the first bad state $E_1$ to be $\al$ whereas the transition probability from each of the bad states equals $\beta$. Let $\eps$ be the probability of a packet loss in good state. We lose packets in any bad state with probability $1$. The burst length distribution in a Fritchman model is a hyper-geometric random variable instead of a geometric random variable. Fritchman and related higher order Markov models are commonly used to model fade-durations in mobile links.

\subsection{Equal Source-Channel Rates}

\begin{figure*}
  \centering
  \subfigure[Simulation results. All codes are evaluated using a decoding delay of $T=12$ symbols and a rate of $R = 12/23 \approx 0.52$.]
  {
    \includegraphics[width=0.48\linewidth, height=5cm]{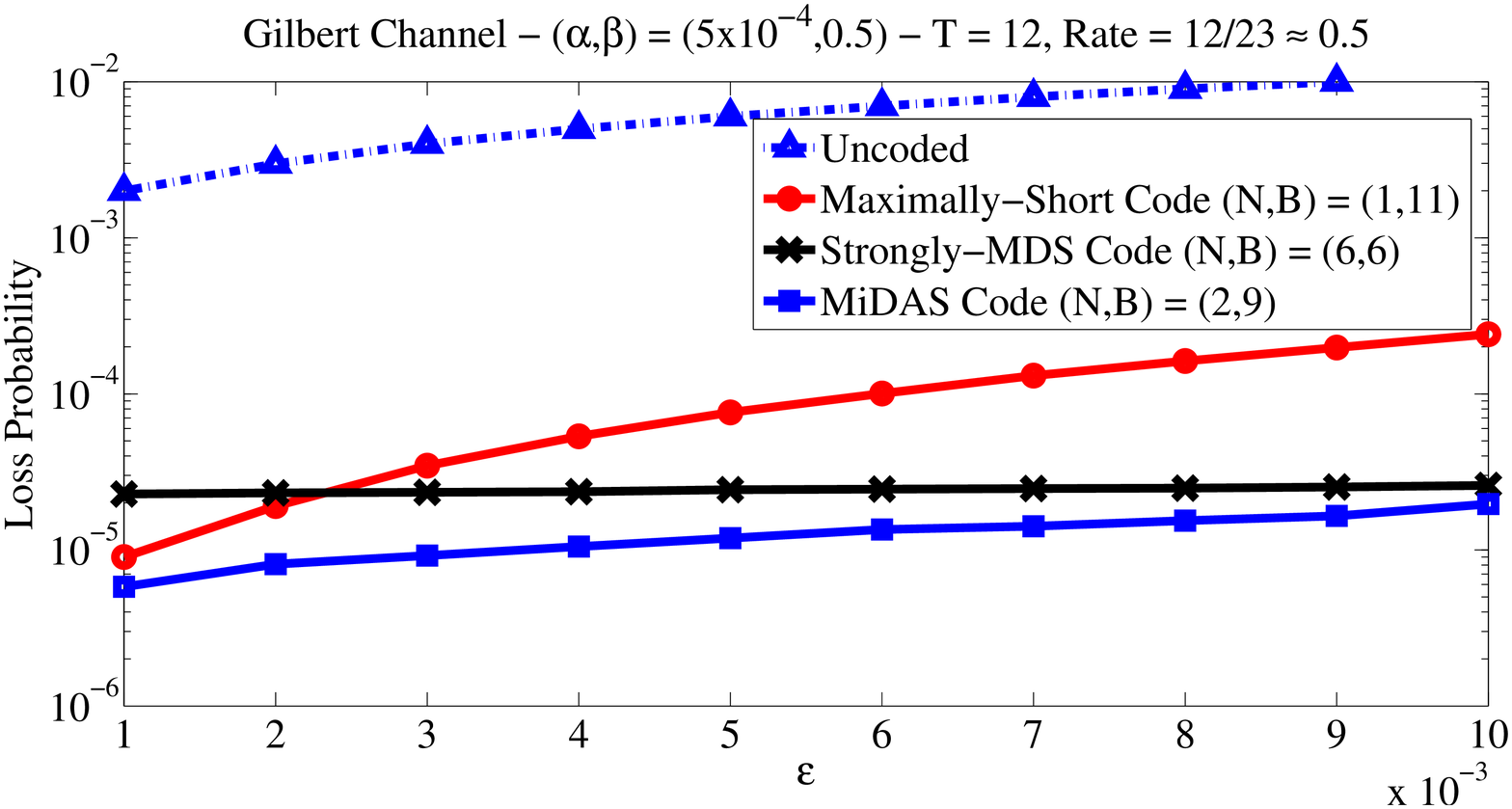}
    \label{fig:GE_T12_R5_Performance}
  }\hspace{0.1em}
  \subfigure[Burst Histogram.]
  {
    \includegraphics[width=0.48\linewidth, height=5cm]{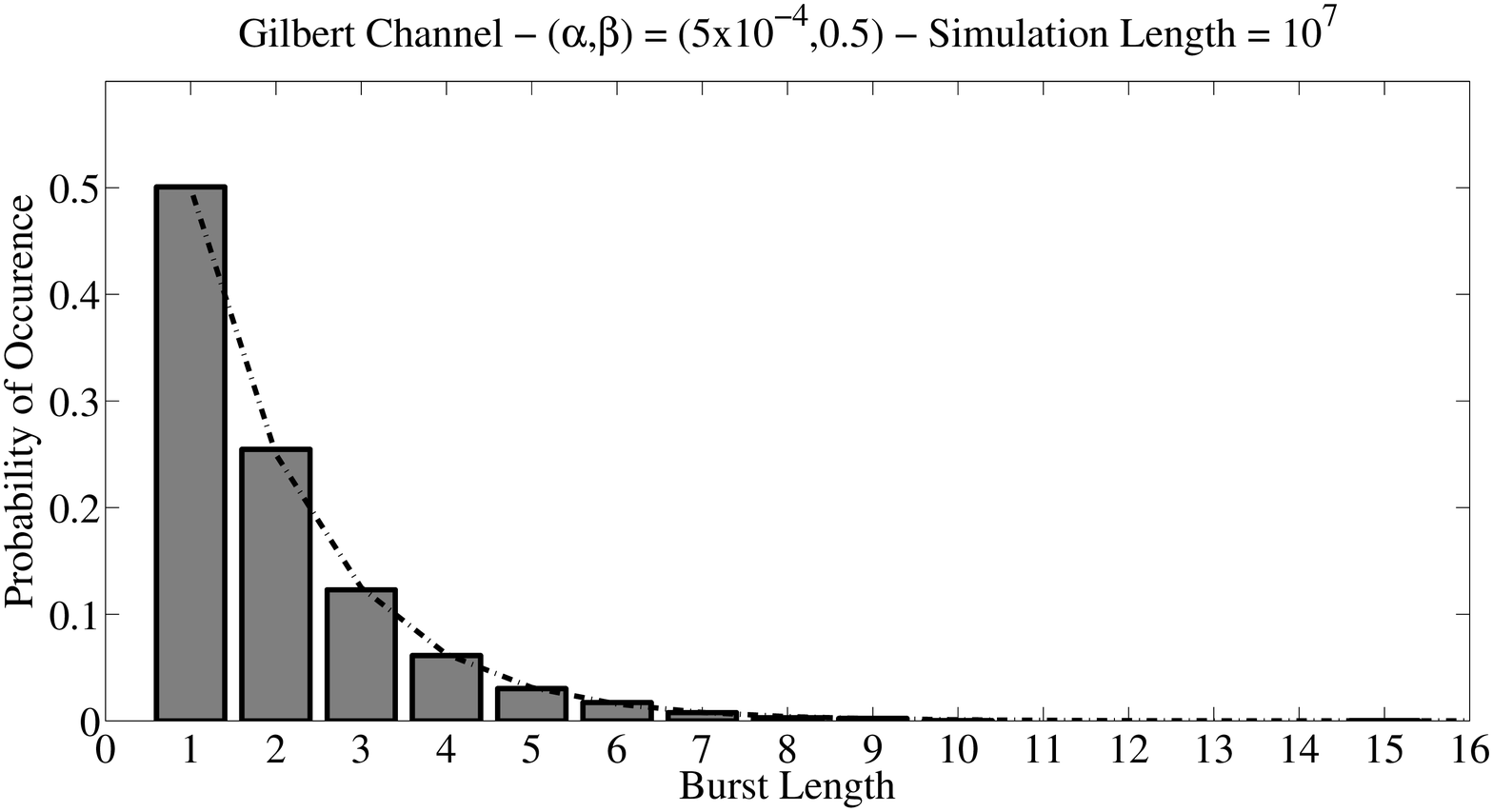}
    \label{fig:GE_T12_R5_Burst}
  }
  \caption{Simulation Experiments for Gilbert-Elliott Channel Model with $(\al,\beta) = (5 \times 10^{-4},0.5)$.}
  \label{fig:Gilbert_T12}
\end{figure*}

\begin{figure*}
  \centering
  \subfigure[Simulation results. All codes are evaluated using a decoding delay of $T=50$ symbols and a rate of $R = 50/83 \approx 0.6$.]
  {
    \includegraphics[width=0.48\linewidth, height=5cm]{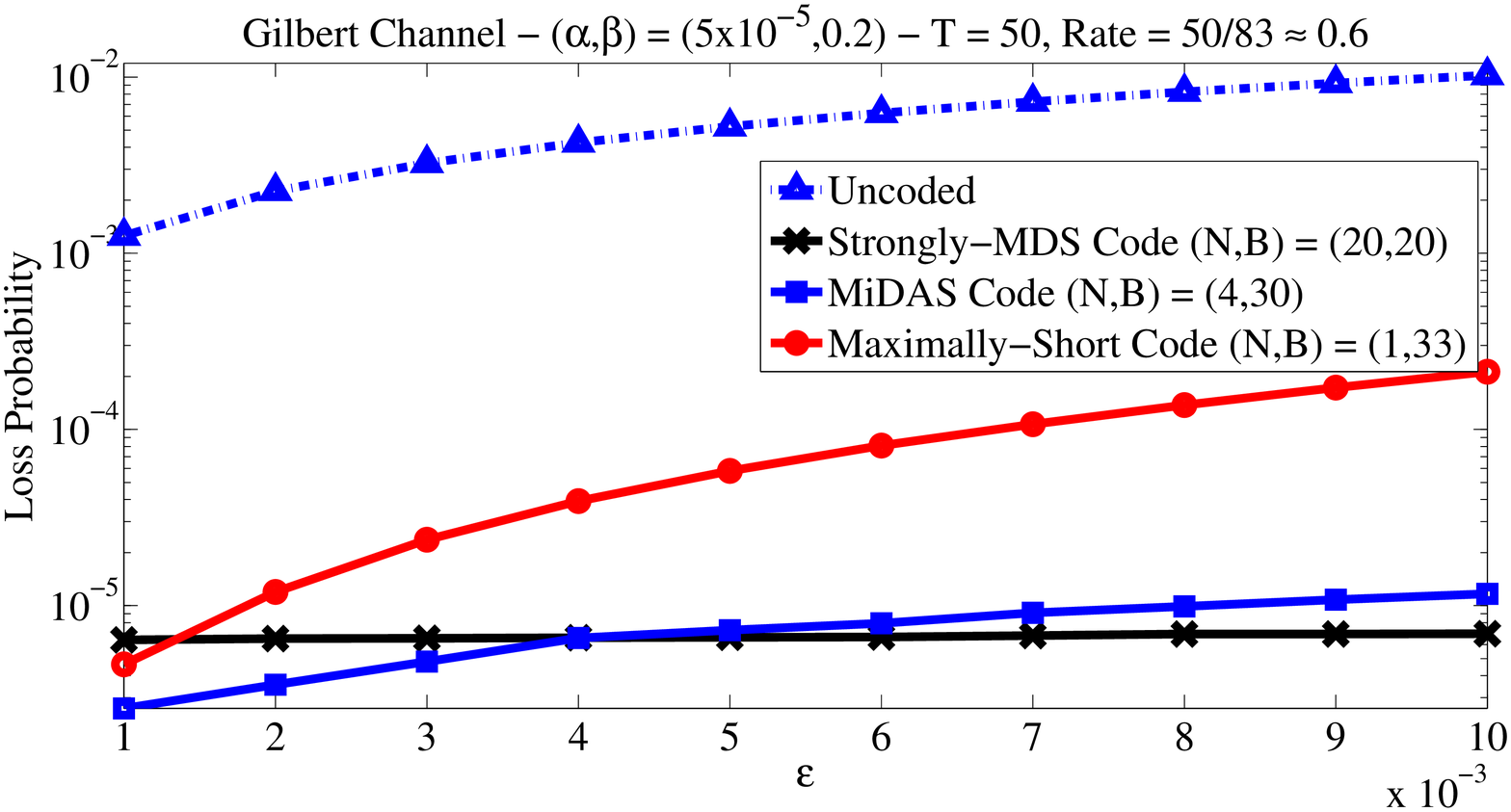}
    \label{fig:GE_T50_R6_Performance}
  }\hspace{0.1em}
  \subfigure[Burst Histogram.]
  {
    \includegraphics[width=0.48\linewidth, height=5cm]{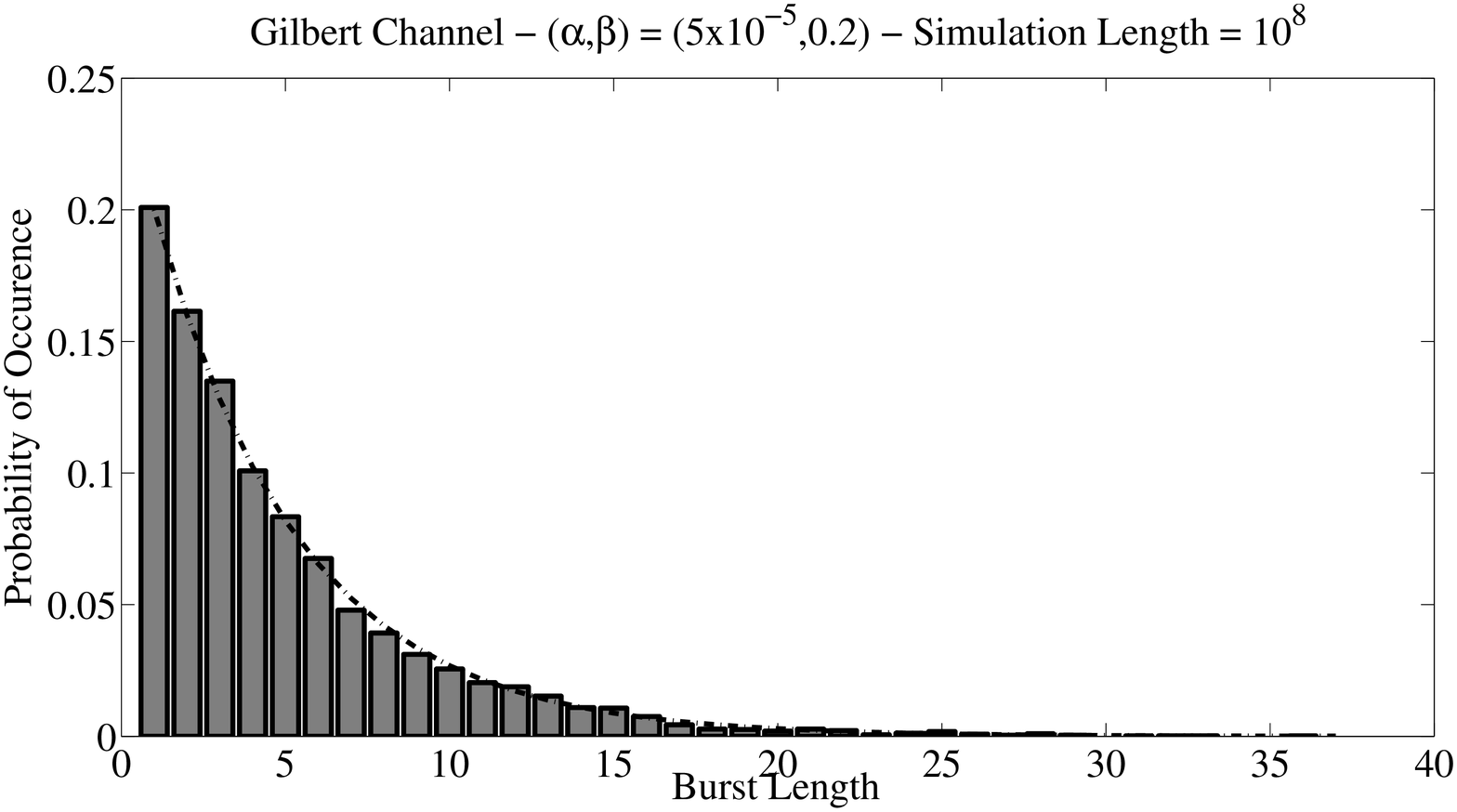}
    \label{fig:GE_T50_R6_Burst}
  }
  \caption{Simulation Experiments for Gilbert-Elliott Channel Model with $(\al,\beta) = (5 \times 10^{-5},0.2)$.}
  \label{fig:Gilbert_T50}
\end{figure*}

\begin{table}
\parbox[b]{.5\linewidth}{
\centering
\begin{tabular}{l|c|c}
			& Fig.~\ref{fig:GE_T12_R5_Performance} & Fig.~\ref{fig:GE_T50_R6_Performance} \\\hline
			Delay $T$ & 12 & 50 \\
			$(\alpha,\beta)$ & $(5 \times 10^{-4},0.5)$ & $(5 \times 10^{-5},0.2)$ \\
			Channel Length & $10^7$ & $10^8$ \\
			Rate $R$ & $12/23$ & $50/83$ \\\hline
		\end{tabular}
	\caption{Gilbert-Elliott Channel Parameters}
	\label{tab:GE-Matched}}
\hfill
\parbox[b]{.5\linewidth}{
\centering
		\begin{tabular}{l|c|c|c|c}
         & \multicolumn{2}{c|}{Fig.~\ref{fig:GE_T12_R5_Performance}} & \multicolumn{2}{c}{Fig.~\ref{fig:GE_T50_R6_Performance}} \\ \hline
        Code & N & B & N & B \\ \hline
        MiDAS Code & 2 & 9 & 4 & 30 \\
        Strongly MDS & 6 & 6 & 20 & 20\\
        MS Codes & 1 & 11 & 1 & 33 \\\hline
	\end{tabular}
	\caption{Achievable $N$ and $B$ for different streaming codes}\label{tab:GE-Matched-Codes}}
\end{table}

In Fig.~\ref{fig:GE_T12_R5_Performance} and Fig.~\ref{fig:GE_T50_R6_Performance} we study the performance of various streaming codes over the Gilbert-Elliott channel. The channel parameters and code parameters are shown in Table~\ref{tab:GE-Matched} and~\ref{tab:GE-Matched-Codes} respectively. Fig.~\ref{fig:GE_T12_R5_Burst} and~\ref{fig:GE_T50_R6_Burst} indicate the histogram of the burst lengths observed for the two channels. We remark that the channel parameters for the $T=12$ case are the same as those used in~\cite[Section 4-B, Fig.~5]{MartinianS04}. 
We remark that for this choice of $\al$, the contribution from failures due to small guard periods between bursts is not dominant. When the inter-burst gaps are smaller we believe that an extension of MiDAS codes that control the number of losses in such events may be necessary and is left for a future investigation. 

All codes in Fig.~\ref{fig:GE_T12_R5_Performance} are selected to have a rate of $R=12/23 \approx 0.52$ and the delay is $T=12$. For reference the uncoded loss-rate is also shown by the upper-most dotted blue line marked with triangles. The black horizontal line is the loss rate of the Strongly-MDS code. It achieves $B=N=6$. Thus its performance is limited by its burst-correction capability and thus is consistent with the probability of observing bursts longer than $6$ which is given by $\approx 2 \times 10^{-5}$. The red-curve which deteriorates rapidly as we increase $\eps$ is the Maximally Short code (MS). It achieves $B=11$ and $N=1$. Thus in general it cannot recover from even two losses occurring in a window of length $T+1$. 
The remaining curve marked with squares shows the MiDAS code which achieve $B=9$ and $N=2$. The loss probability also deteriorates with $\eps$ but at a much lower rate. Thus a slight decrease in $B$, while improving $N$ from $1$ to $2$ exhibits noticeable gains over both MS and Strongly-MDS codes. At the left most point i.e., when $\eps=10^{-3}$, the loss probability is dominated by burst losses, while as $\eps$ is increased, the effect of isolated losses becomes more significant.

In Fig~\ref{fig:GE_T50_R6_Performance} the rate of all codes is set to $R=50/83 \approx 0.6$. The delay is set to $T=50$. The Strongly-MDS code (black horizontal plot) achieves ${B=N=20}$ whereas the MS code (red plot) achieves $N=1$ and $B=33$. Both codes suffer from the same phenomenon discussed in the previous case. We also consider the MiDAS code (blue plot) with $N=4$ and $B=30$. We observe that its performance deteriorates as $\eps$ is increased and eventually crosses the Strongly-MDS codes. We believe that further improvements can be attained using codes that correct both burst and isolated losses~\cite{isit-midas}, but leave such extensions for a future investigation.

In Fig.~\ref{fig:Fritchman_T40} and Fig.~\ref{fig:Fritchman_R6}, we evaluate streaming codes over the Fritchman channel in Fig.~\ref{fig:fritchman}. The channel parameters and code parameters are shown in Table~\ref{tab:Fritchman} and~\ref{tab:Fritchman-Codes} respectively. We let the transition probability from the good state to the first bad state $E_1$ to be $\al$ whereas the transition probability from each of the bad states equals $\beta$. Let $\eps$ be the probability of a packet loss in good state. We lose packets in any bad state with probability $1$. Fig.~\ref{fig:Fritchman_T40_Burst} and~\ref{fig:Fritchman_R6_Burst} indicate the histogram of the burst lengths observed for the two channels. 

\begin{figure*}
  \centering
		\subfigure[Simulation over a $\mathcal{N}+1 = 9$-States Fritchman Channel with $(\al,\beta) = (10^{-5},0.5)$. All codes are evaluated using a decoding delay of $T=40$ symbols.]
		{
  \includegraphics[width=0.48\linewidth, height = 5cm]{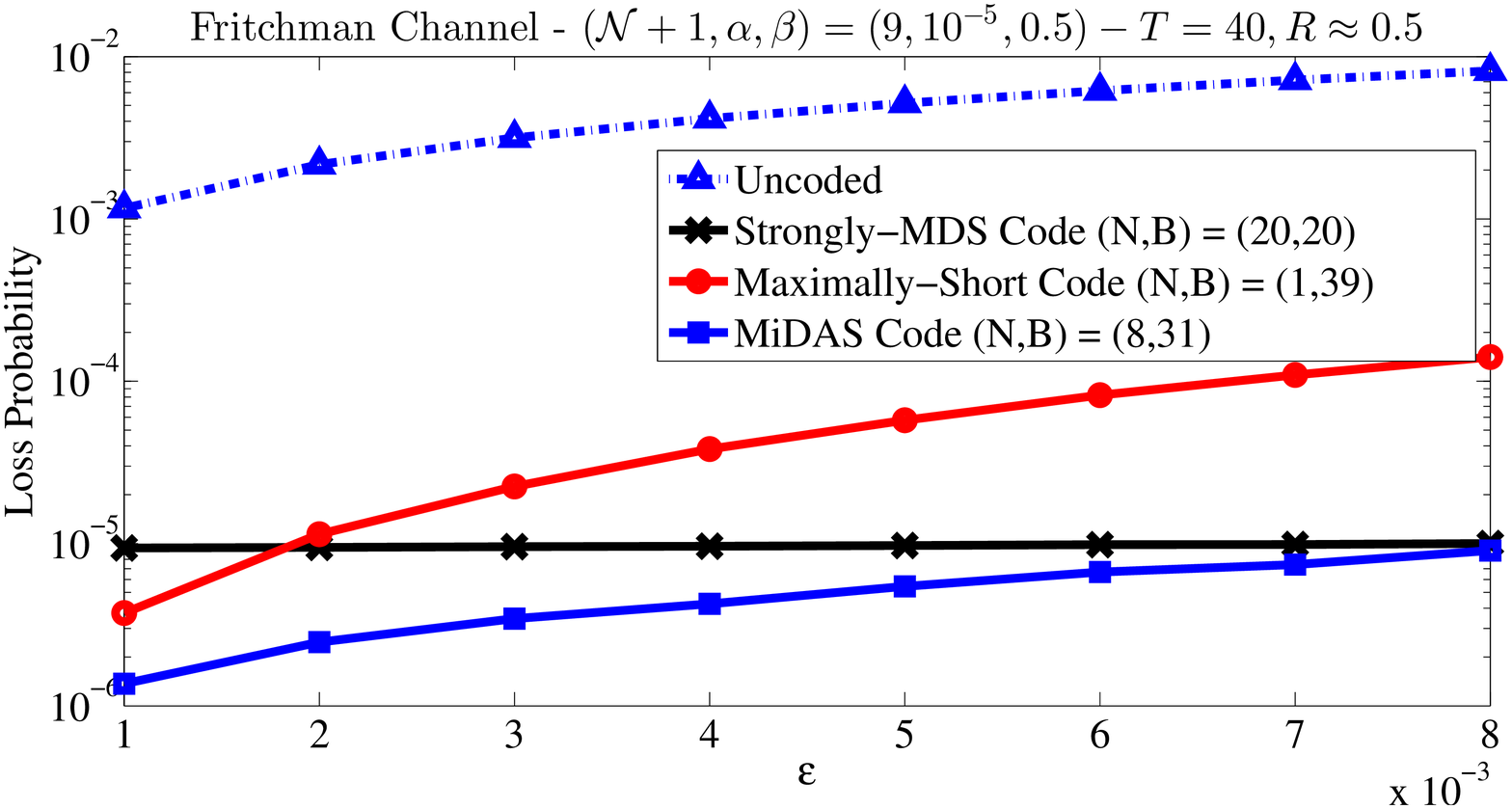}
  \label{fig:Fritchman_T40_Performance}
		}\hspace{0.1em}
		\subfigure[Burst Histogram.]
		{
		\includegraphics[width=0.48\linewidth, height = 5cm]{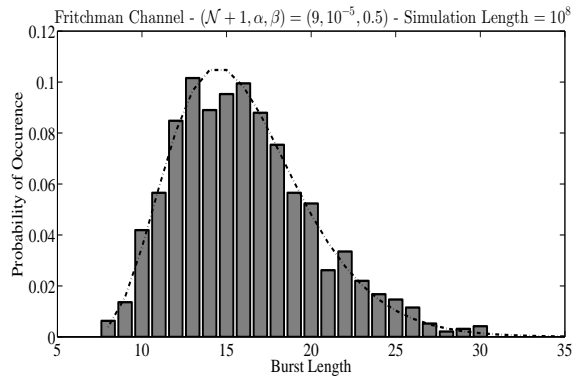}
		\label{fig:Fritchman_T40_Burst}
		}
		\caption{Simulation Experiments for Fritchman Channel Model with $(\mathcal{N},\al,\beta) = (8,10^{-5},0.5)$.}
		\label{fig:Fritchman_T40}
\end{figure*}

\begin{figure*}
  \centering
		\subfigure[Simulation over a $\mathcal{N}+1 = 12$-States Fritchman Channel with $(\al,\beta) = (2 \times 10^{-5},0.75)$. All codes are evaluated using a decoding delay of $T=40$ symbols.]
  {
		\includegraphics[width=0.48\linewidth, height = 5cm]{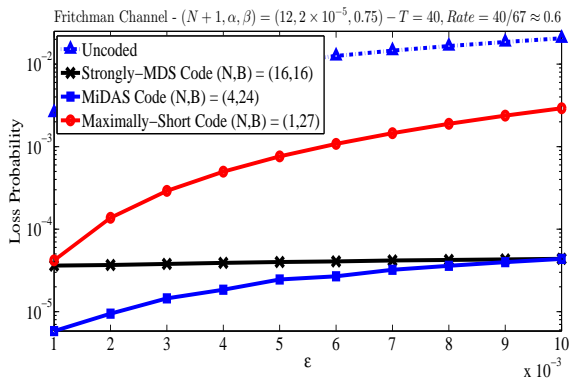}
  \label{fig:Fritchman_R6_Performance}
		}\hspace{0.1em}
  \subfigure[Burst Histogram.]
		{
		\includegraphics[width=0.48\linewidth, height = 5cm]{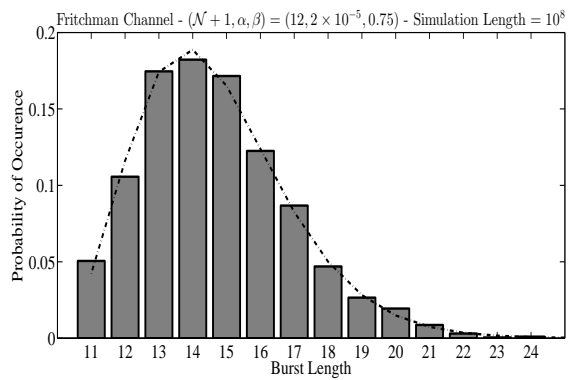}
		\label{fig:Fritchman_R6_Burst}
		}
		\caption{Simulation Experiments for Fritchman Channel Model with $(\mathcal{N},\al,\beta) = (11,2 \times 10^{-5},0.75)$.}
		\label{fig:Fritchman_R6}
\end{figure*}

\begin{table}[!htb]
\parbox[b]{.5\linewidth}{
\centering
\begin{tabular}{l|c|c}
			& Fig.~\ref{fig:Fritchman_T40} & Fig.~\ref{fig:Fritchman_R6} \\\hline
			Channel States & 9 & 12 \\
			Delay $T$ & 40 & 40 \\
			$(\alpha,\beta)$ & $(10^{-5},0.5)$ & $(2 \times 10^{-5},0.75)$ \\
			Channel Length & $10^8$ & $10^8$ \\
			Rate $R$ & $40/79 \approx 0.5$ & $40/67 \approx 0.6$ \\\hline
		\end{tabular}
	\caption{Fritchman Channel Parameters}
	\label{tab:Fritchman}
}
\hfill
\parbox[b]{.5\linewidth}{
\centering
		\begin{tabular}{l|c|c|c|c}
         & \multicolumn{2}{c|}{Fig.~\ref{fig:Fritchman_T40}} & \multicolumn{2}{c}{Fig.~\ref{fig:Fritchman_R6}} \\ \hline
        Code & N & B & N & B \\\hline
        MiDAS Codes & 8 & 31 & 4 & 24 \\
        Strongly MDS & 20 & 20 & 16 & 16\\
        MS Codes & 1 & 39 & 1 & 27 \\\hline
	\end{tabular}
	\caption{Achievable $N$ and $B$ for different streaming codes}\label{tab:Fritchman-Codes}}
\end{table}

In Fig.~\ref{fig:Fritchman_T40} and Fig.~\ref{fig:Fritchman_R6}, the uncoded loss rate is shown by the upper-most plot while the black horizontal line is the performance of Strongly-MDS code. Note that the performance of this code is essentially independent of $\eps$ in the interval of interest. As in the case of GE channels, the Strongly-MDS codes recover all the losses in the good state and fail against burst lengths longer than its burst erasure correction capability. Thus, their loss rate is consistent with the probability of observing bursts longer than $20$ and $16$ which can be calculated to be $\approx 10^{-5}$ and $\approx 3 \times 10^{-5}$, respectively. The performance of the MS codes is shown by the red-plot in both figures. We note that it is better than the Strongly-MDS codes for $\eps = 10^{-3}$, but deteriorates quickly as we increase $\eps$. The performance gains from MiDAS codes are significantly more noticeable for the Fritchman channel because the hyper-geometric burst-length distribution favors longer bursts over shorter ones. As in the case of GE Channels, we expect further performance gains to be possible by considering more sophisticated erasure patterns, such as burst plus isolated losses, but leave such an investigation for a future work.

In Fig.~\ref{fig:Gilbert_T12_B11_MDS_Performance}, we compare the performance of MiDAS and MS codes obtained by replacing the Strongly-MDS constituent code with a diagonally interleaved block MDS code (cf. Section~\ref{subsec:midas-field}). We consider the same GE channel in Fig.~\ref{fig:GE_T12_R5_Performance} and delay $T=12$. The codes involving Strongly-MDS codes are plotted using a solid line whereas the codes involving block MDS codes are shown by the dotted lines of the same color. We note that in all cases there is a noticeable increase in the loss rate when a block MDS code is used despite the fact that these codes achieve the same $(N,B)$ values over deterministic channels. This loss in performance is due to their sensitivity to non-ideal erasure patterns as discussed in Section~\ref{subsec:midas-nonideal}.

\begin{figure}
 \begin{minipage}{0.475\linewidth}
 \includegraphics[width=\columnwidth]{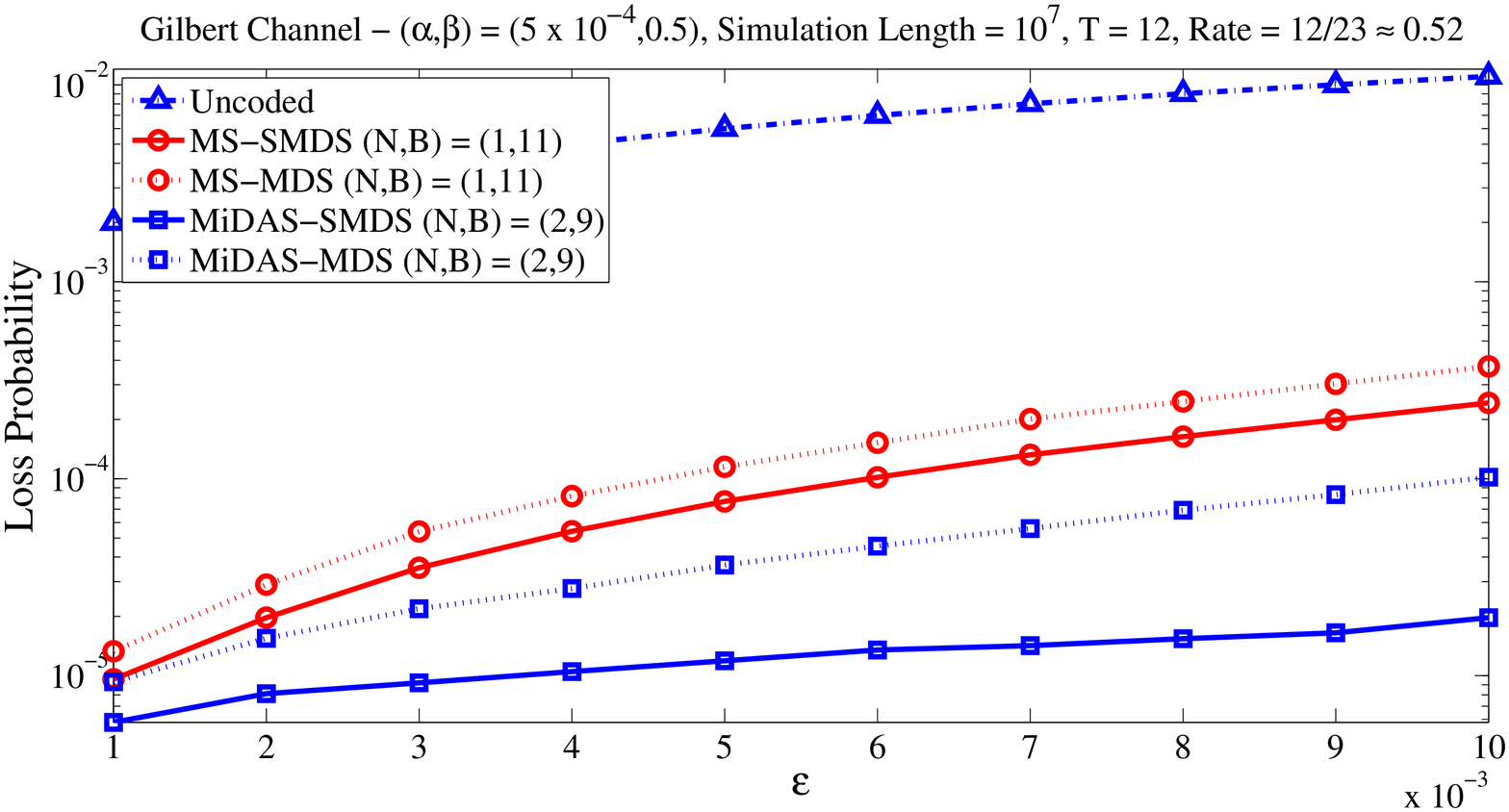}
\caption{Simulation over a Gilbert-Elliott Channel with $(\al,\beta) = (5 \times 10^{-4},0.5)$. All codes are evaluated using a decoding delay of $T=12$ symbols and a rate of $R = 12/23 \approx 0.52$.}
    \label{fig:Gilbert_T12_B11_MDS_Performance}
   \end{minipage}
   \hspace{0.05\linewidth}
   \begin{minipage}{0.475\linewidth}
 \includegraphics[width=\columnwidth]{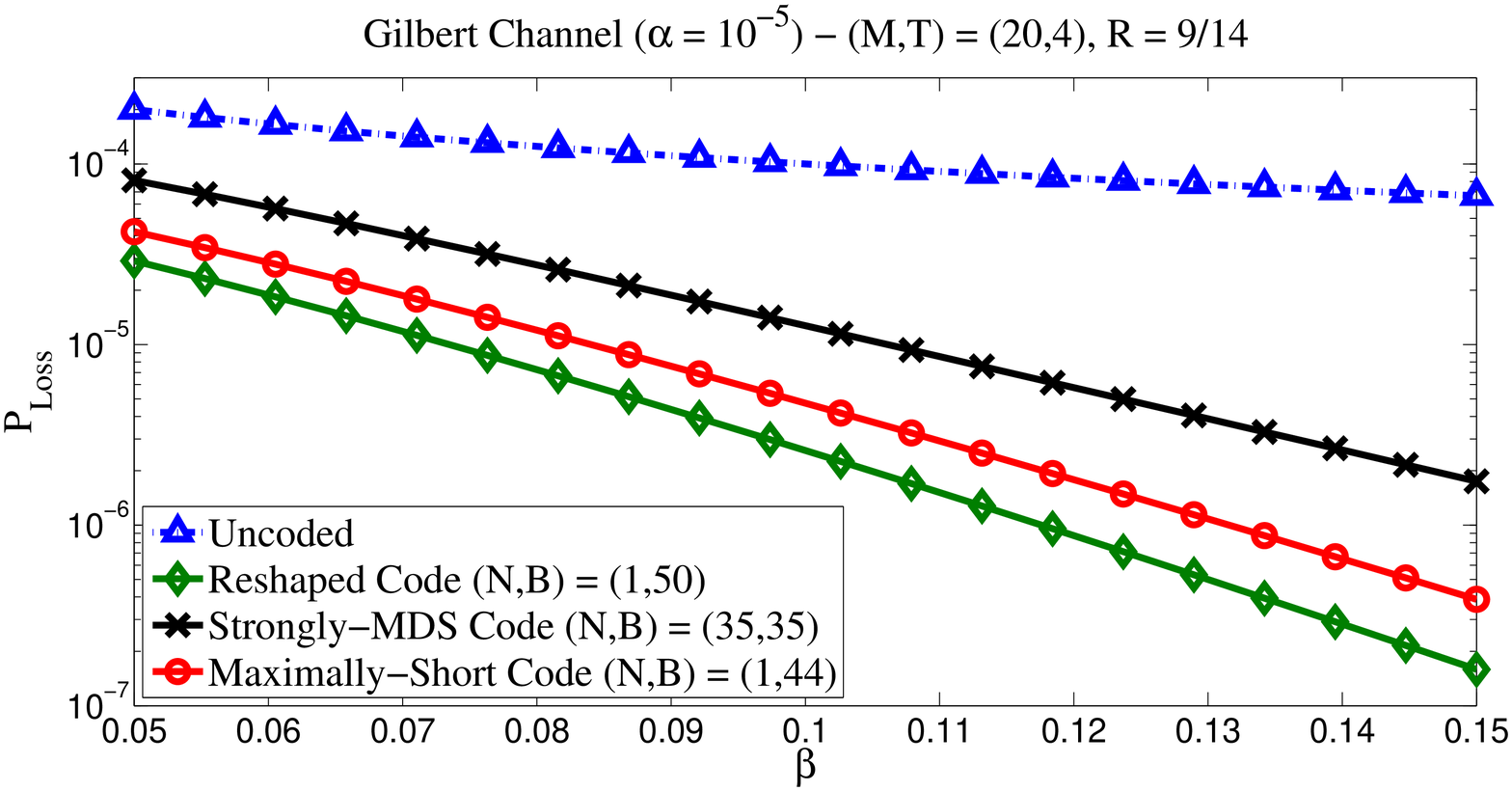}
\caption{Simulation over a Gilbert Channel with $\al = 10^{-5}$ and $\beta$ varied on the x-axis. All codes are of rate $R = \frac{9}{14}$ and evaluated using a decoding delay of $T=4$ macro-packets. Each macro-packet consists of $M=20$ channel symbols.}
\label{fig:Gilbert_Experiments_M20}
   \end{minipage}
 \end{figure}

\subsection{Unequal Source-Channel Rates}


In our simulations in Fig.~\ref{fig:Gilbert_Experiments_M20}, we consider a Gilbert channel model which is the same as a Gilbert-Elliott channel with $\eps = 0$, i.e., the loss probability is $0$ in the good state. We fix $\alpha = 10^{-5}$ and vary $\beta$ on the x-axis in the interval $[0.05,0.15]$ which in turn changes the burst length distribution. We further select $M = 20$, i.e., $20$ channel packets are generated for every source packet received at the encoder. We fix the rate $R = 9/14$ and the delay $T=4$ macro-packets. Under these conditions, the Strongly-MDS code can correct burst erasures of length up to $B = 35$, whereas a Maximally Short code achieves $B = 44$. The optimal code for $M>1$ achieves $B = 50$. This gain in terms of correctable burst-length is reflected in Fig.~\ref{fig:Gilbert_Experiments_M20} as one can see that codes designed for unequal source-channel rates, which are referred to as \emph{reshaped codes}, achieve a lower loss probability. We note that the code parameters in Fig.~\ref{fig:Gilbert_Experiments_M20} correspond to the second case in~\eqref{eq:capacity}.

\begin{table}[!htb]
\parbox[b]{.5\linewidth}{
\centering
\begin{tabular}{l|c|c}
			& Fig.~\ref{fig:Gilbert_Experiments_M20} & Fig.~\ref{fig:FC_M40_T2} \\\hline
			Channel States & $2$ & $20$ \\
			$(M,T)$ & $(20,4)$ & $(40,2)$ \\
			$(\alpha,\beta)$ & $(10^{-5},[0.05,0.15])$ & $(10^{-5},0.5)$ \\
			Channel Length & $10^9$ & $10^9$ \\
			Rate $R$ & $9/14 \approx 0.64$ & $40/63 \approx 0.63$ \\\hline
		\end{tabular}
	\caption{Unequal Source Channel Rates}
	\label{tab:Unequal}
}
\hfill
\parbox[b]{.5\linewidth}{
\centering
		\begin{tabular}{l|c|c|c|c}
         & \multicolumn{2}{c|}{Fig.~\ref{fig:Gilbert_Experiments_M20}} & \multicolumn{2}{c}{Fig.~\ref{fig:FC_M40_T2}} \\ \hline
        Code & N & B & N & B \\\hline
								Reshaped Code & 1 & 50 & 1 & 58 \\
								Robust Reshaped Code & N/A & N/A & 5 & 53 \\
        MiDAS Code & N/A & N/A & 5 & 42 \\
        Strongly MDS Code & 35 & 35 & 43 & 43\\
        MS Codes & 1 & 44 & 1 & 45 \\ \hline
	\end{tabular}
	\caption{Achievable $N$ and $B$ for different streaming codes}
	\label{tab:Unequal-Codes}}
\end{table}

In Fig.~\ref{fig:FC_M40_T2}, we consider a Fritchman channel with $(\alpha,\beta) = (10^{-5},0.5)$ and $\mathcal{N}+1 = 20$ states. The corresponding burst distribution is illustrated in Fig.~\ref{fig:FC_M40_T2_Burst}. In Fig.~\ref{fig:FC_M40_T2_Performance}, we show the performance of different streaming codes in the case of unequal source-arrival and channel-transmission rates on such channel. The rate for all codes is fixed to $R = 0.64$ and the delay constraint is $T=2$ macro-packets where each macro-packet has $M = 40$ packets. As the probability of erasure in the good state $\eps$ increases, the performance of Strongly-MDS code (black curve) does not change. The loss rate of this code is $\approx 10^{-4}$ which is dominated by the fraction of erasures introduced by bursts longer than $43$. On the other hand, both Maximally Short and reshaped codes achieve $N=1$ and thus deteriorate as quickly as $\eps^2$. For the left most point corresponding to $\eps = 0$, the probability of loss of the Maximally Short code is $\approx 10^{-4}$ which reflects the number of erasures introduced by bursts longer than $45$. Similarly, the loss probability of the reshaped code is $\approx 3 \times 10^{-6}$ which matches the fraction of losses introduced due to bursts longer than $58$. The performance of the robust versions of these codes, namely MiDAS and robust reshaped codes does not deteriorate as fast. However, the robust reshaped code outperforms the MiDAS code as the former achieves $B=53$ versus $B=42$ achieved by the later while fixing $N=5$, $R = 0.63$ and $T=2$.

\begin{figure*}
  \centering
  \subfigure[Simulation results. All codes are evaluated using a decoding delay of $T=2$ macro-packets and a rate of $R \approx 0.63$. Each macro-packet consists of $M=40$ channel packets.]
  {
    \includegraphics[width=0.48\linewidth, height=5cm]{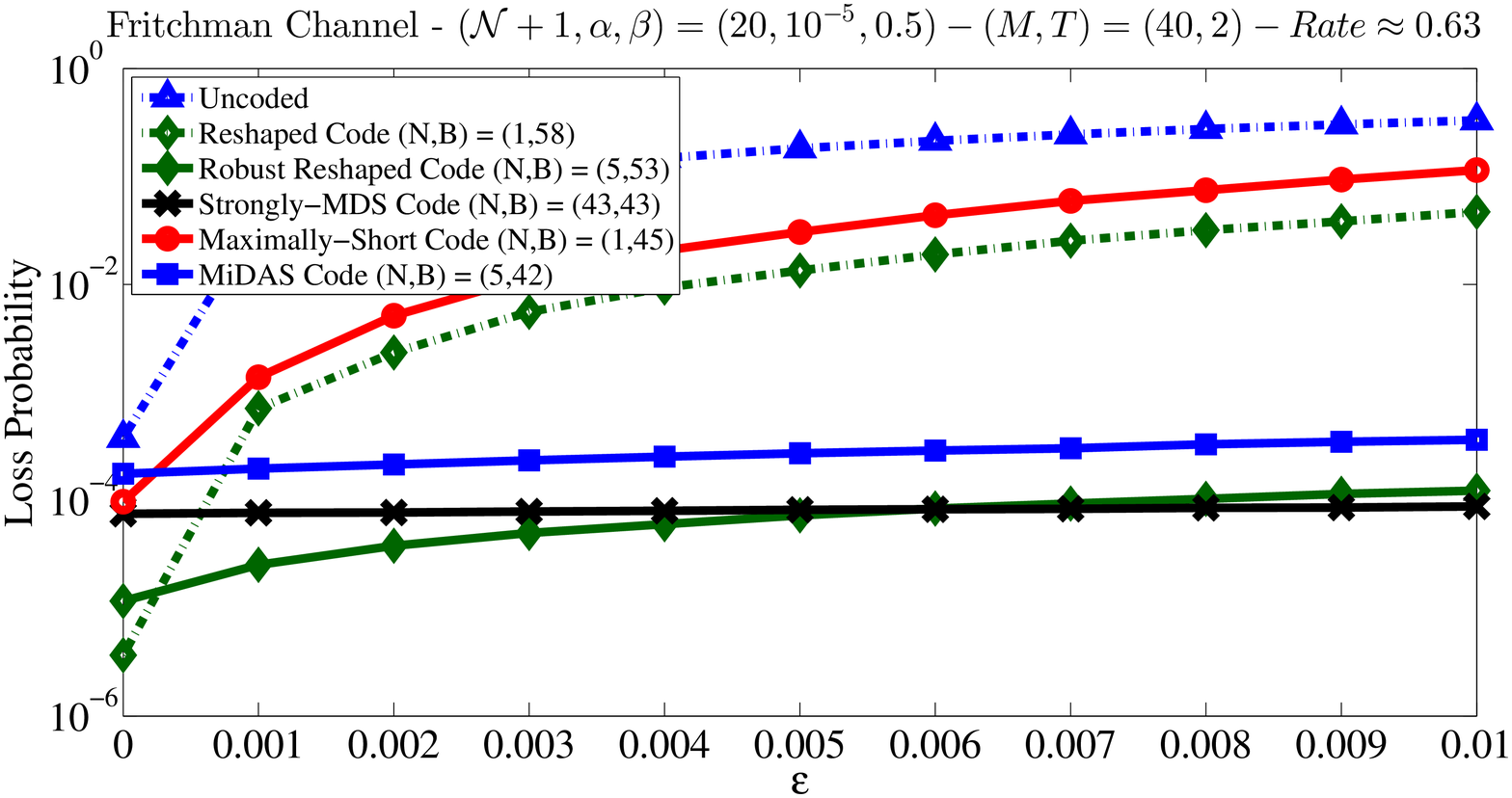}
    \label{fig:FC_M40_T2_Performance}
  }\hspace{0.1em}
  \subfigure[Burst Histogram.]
  {
    \includegraphics[width=0.48\linewidth, height=5cm]{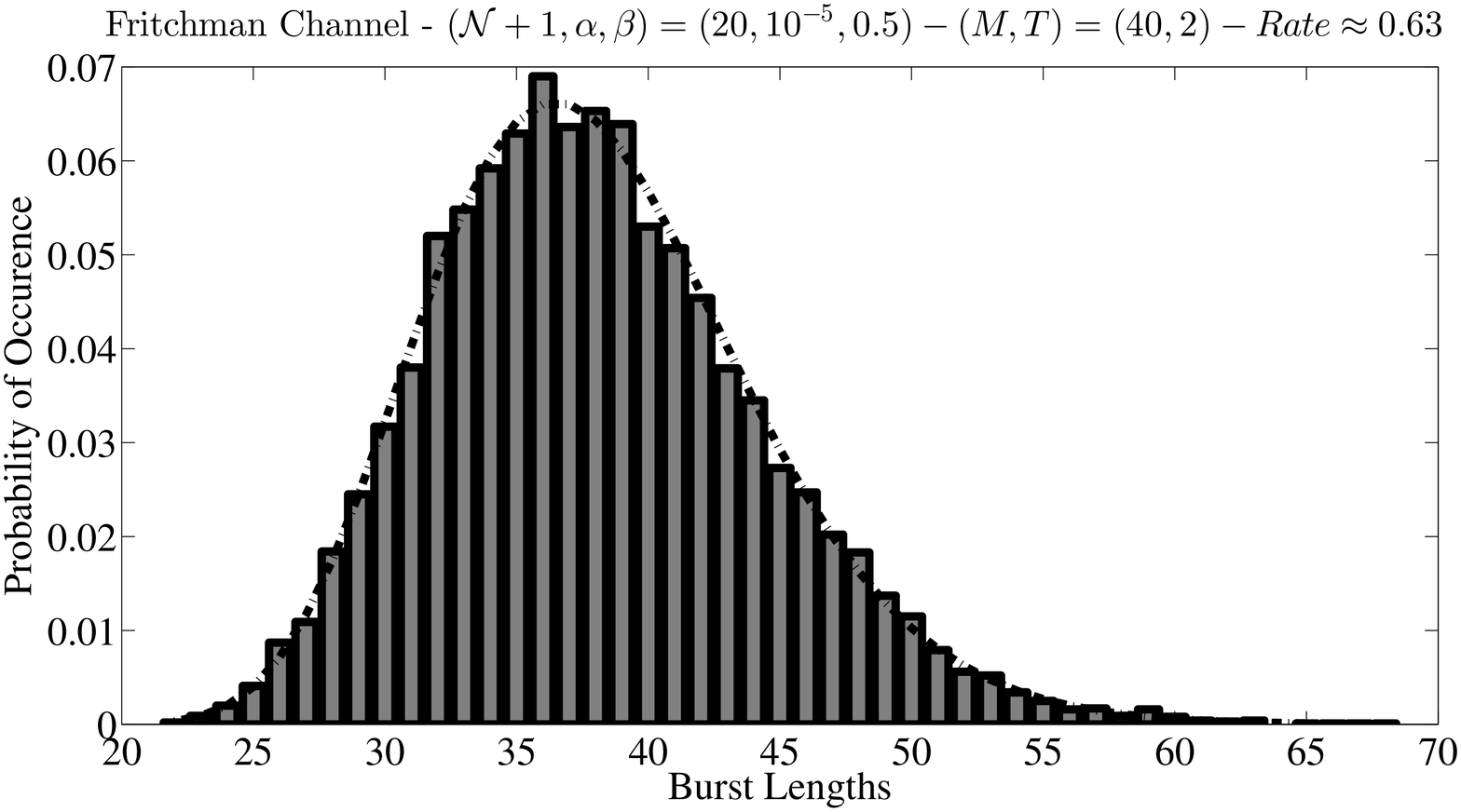}
    \label{fig:FC_M40_T2_Burst}
  }
  \caption{Simulation Experiments for Fritchman Channel Model with $\mathcal{N}+1=20$ states and $(\al,\beta) = (10^{-5},0.5)$.}
  \label{fig:FC_M40_T2}
\end{figure*}

\section{Conclusion}
\label{sec:conclusion}
We study low-delay error correction codes for streaming sources over packet erasure channels. Our constructions are based on a layered approach and use Strongly-MDS codes and repetition codes as constituent codes. For the case when the source-arrival and channel-transmission rates are equal, we establish a fundamental tradeoff between the burst-error correction and isolated-error correction capability of any streaming code and present a near-optimal construction. The relevance of the underlying column distance and column span properties for low-delay error correction is also discussed.
When the source-arrival and channel-transmission rates are not equal, we present an optimal code for the burst erasure channel with sufficiently long guard intervals. While our constructions are based on deterministic channel models, numerical simulations indicate that the proposed constructions outperform baseline schemes over statistical models.

For future work, the characterization of optimal streaming codes for unequal source-channel rates when $N>1$ remains open. The investigation of minimum required field-size for streaming codes is also highly important. Further improvements can be attained by considering streaming codes that correct both burst and isolated losses in the window of interest as observed in~\cite{isit-midas}.
Finally our constructions are tuned to specific channel parameters. In practice it is very desirable to extend such constructions that adapt to varying channel parameters with little or no feedback.

\appendices

\section{Column Distance and Column Span of Convolutional codes}
\label{app:distance-span}
In this section we show that the error correction capability of a streaming code can be expressed in terms of its column distance and column span. In our discussion we view the input symbols $\bs[i]$ as a length $\bar{k}$ vector over ${\mathbb F}_q$ and $\bx[i]$ as a length $\bar{n}$ vector over ${\mathbb F}_q$. 
We restrict our attention to time-invariant linear $(\bar{n},\bar{k},\bar{m})$ convolutional codes specified by
$$\bx[i] = \left(\sum_{j=0}^{\bar{m}} \bs^\dagger[i-j] \bG_{j} \right)^\dagger,$$
where $\bG_0,\ldots, \bG_{\bar{m}}$ are generator matrices over ${\mathbb F}_q^{\bar{k} \times \bar{n}}$. 

The first $T+1$ output symbols can be expressed as,
\begin{align}
\vspace{-1em}
\label{eq:trunc-cc}
[ \bx[0], \bx[1], \ldots, \bx[T] ] = [\bs[0], \bs[1], \ldots, \bs[T]] \cdot \bG^s_T.
\end{align}
where 
\begin{equation}\bG^s_T = \begin{bmatrix}\bG_0 & \bG_1 & \ldots & \bG_T \\ 0 & \bG_0 & & \bG_{T-1} \\ \vdots & &\ddots & \vdots \\ 0 & & \ldots & \bG_0 \end{bmatrix}\label{eq:GsT}\end{equation}
is the truncated generator matrix to the first ${T+1}$ columns. Note that $\bG_j =0$ if $j > m$. 


\begin{defn}[Column Distance]
The column distance of $\bG_T^s$ in~\eqref{eq:GsT} is defined as
\begin{equation}
d_T = \min_{\substack{\bs \equiv [\bs[0], \bs[1], \ldots, \bs[T]]\\ \bs[0] \neq 0}} \mathrm{wt}([\bx[0], \ldots, \bx[T]])\label{eq:dT-def}
\end{equation}
where $\mrm{wt}([\bx[0],\ldots, \bx[T]])$ counts the number of non-zero elements in the ${T+1}$ length vector. 
\end{defn}

Intuitively the column distance of the convolutional code finds the codeword sequence of minimum Hamming weight in the interval $[0,T]$ that diverges from the all zero state at time $t=0$. We refer the reader to~\cite[Chapter 3]{zigangirov} for some properties of $d_T$. 

\begin{fact}
A convolutional code with a column distance of $d_T$ can recover every information symbol with a delay of $T$ provided the channel introduces no more than
$N =d_T-1$ erasures in any sliding window of length ${T+1}$.
Conversely there exists at-least one erasure pattern with $d_T$ erasures in a window of length ${T+1}$ where the decoder fails to recover all source symbols.
\label{fact:dT}
\end{fact}
\begin{proof}
Consider the interval $[0,T]$ and consider two input sequences $(\bs[0],\ldots, \bs[T])$ and $(\bs'[0],\ldots, \bs'[T])$ with $\bs[0] \neq \bs'[0]$. Let the corresponding output 
be $(\bx[0],\ldots, \bx[T])$ and $(\bx'[0],\ldots, \bx'[T])$. Note that the output sequences differ in at-least $d_T$ indices since otherwise the output sequence $(\bx[0] - \bx'[0], \ldots, \bx[T]-\bx'[T])$
which corresponds to $(\bs[0]-\bs'[0], \ldots, \bs[T]-\bs'[T])$ has a Hamming weight less than $d_T$ while the input $\bs[0] -\bs'[0] \neq 0$, which is a contradiction. 
Thus if $(\bs[0],\ldots, \bs[T])$ is the input source sequence, for any sequence of $d_T-1$ or fewer erasures, there will be at-least one symbol where $(\bx'[0], \ldots, \bx'[T])$ differs from the received sequence.
Thus $\bs[0]$ is recovered uniquely at time $T$. Once $\bs[0]$ is recovered we can cancel its contribution from all the future symbols and repeat the same argument for the interval $[1,T+1]$ to recover $\bs[1]$ and proceed.

Conversely there exists at-least one output sequence whose Hamming weight equals $d_T$ and the input symbol $\bs[0] \neq 0$. 
By erasing all the non-zero $d_T$ positions for this output sequence, we cannot distinguish it from the all-zero sequence. 
\end{proof}

To the best of our knowledge the column span of a convolutional code was first introduced in~\cite{MartinianS04}
in the context of low-delay codes for burst erasure channels. 

\begin{defn}[Column Span]
The column span of $\bG^s_T$ in~\eqref{eq:GsT} is defined as
\begin{equation}
c_T = \min_{\substack{\bs \equiv [\bs[0], \bs[1], \ldots, \bs[T]]\\ \bs[0] \neq 0}} \mathrm{span}([\bx[0], \ldots, \bx[T]])
\label{eq:cT-def}
\end{equation}
where $\mrm{span}([\bx[0], \ldots, \bx[T]])$ equals the support of the underlying vector, i.e., $\mrm{span}([\bx[0], \ldots, \bx[T]]) = j-i+1,$ where $j$ is the last index where $\bx$ is non-zero and $i$ is the first such index.
\end{defn}

\begin{fact}
Consider a channel that introduces no more than a single erasure burst of maximum length $B$ in any sliding window of length ${T+1}$. A necessary and sufficient condition for a convolutional code to recover every erased symbol with a delay of $T$ is that $c_T > B$.
\label{fact:cT}
\end{fact}

The justification is virtually identical to the proof of Fact~\ref{fact:dT} and is omitted.

It follows from Facts~\ref{fact:dT} and~\ref{fact:cT} that a necessary and sufficient condition for any convolutional code
to recover each source symbol with a delay of $T$ over a channel $\cC(N,B,W=T+1)$ is that both $d_T > N$ and $c_T > B$. Thus specializing Theorem~\ref{thm:Chan-1-UB} and~\ref{thm:midas} to $W=T+1$ we have the following.
 \begin{prop}[A Fundamental Tradeoff between Column Distance and Column Span]
 For any $(\bar{n},\bar{k},\bar{m})$ convolutional code and an integer $T > 0$ we have that the column distance $d_T$
 and column span $c_T$ must satisfy
 \begin{align}
 \frac{R}{1-R} c_T + d_T \le T+1 + \frac{1}{1-R}\label{eq:cTdTub}
 \end{align}
 where $R= \frac{\bar{k}}{\bar{n}}$ denotes the rate of the code. Furthermore for any $T>0$ there exists a $(\bar{n},\bar{k},\bar{m})$ convolutional code with column distance $d_T$ and column span $c_T$, over a sufficiently large field-size such that,
 \begin{align}
 \frac{R}{1-R} c_T + d_T \ge T + \frac{1}{1-R}\label{eq:cTdTlb}
 \end{align}
$\hfill\Box$
\label{prop:cTdT-tradeoff}
 \end{prop}
 \begin{proof}
 To establish~\eqref{eq:cTdTub}, consider any convolutional code with a column distance $d_T$ and column span $c_T$. From the sufficiency parts of Facts~\ref{fact:dT} and~\ref{fact:cT} such a code is feasible over the channel $\cC(N=d_T-1,B=c_T-1,W=T+1)$ with delay $T$. Thus it must satisfy the upper bound~\eqref{eq:r-ub}. Substituting $N=d_T-1$ and $B=c_T-1$ immediately gives~\eqref{eq:cTdTub}. 

To establish~\eqref{eq:cTdTlb}, consider the code that satisfies the lower bound in~\eqref{b-achiev} in Theorem~\ref{thm:midas}. From the necessity parts of Facts~\ref{fact:dT} and~\ref{fact:cT} such a code must satisfy $c_T \ge B+1$ and $d_T \ge N+1$. Substituting in~\eqref{b-achiev} immediately leads to~\eqref{eq:cTdTlb}.
 
 \end{proof}
 
As a final remark we note Facts~\ref{fact:dT} and~\ref{fact:cT} also immediately apply to any channel with $W \ge T+1$. In particular any erasure pattern for the $\cC(N,B,W)$ channel with $W \ge T+1$ is also feasible for $\cC(N,B,W=T+1)$ and thus the sufficiency follows. Furthermore also note that whenever $W \ge T+1$, any erasure pattern in the interval $[0,T]$ used in the proof of the necessity part can also be used for the channel $\cC(N,B,W)$.

\section{Proof of Lemma~\ref{lem:mds-sub}}
\label{app:strongly-mds}

In order to establish L1, we use the following property regarding systematic Strongly-MDS codes~\cite[Corollary 2.5]{strongly-mds}. Consider the window of the first $j+1$ symbols of a $(\bar{n},\bar{k},\bar{m})$ convolutional code and let the truncated codeword associated with the input sequence $(\bs[0], \ldots, \bs[j])$ be $(\bx[0],\dots,\bx[j])$, where each $\bx[i]$ is expressed as in~\eqref{eq:strong-mds-systematic}. Then the $j$-th (sub-symbol level) column distance\footnote{This differs from~\eqref{eq:dT-def} in that we measure the Hamming weight of sub-symbols rather than the symbols $\bx[j]$.} is defined as
\begin{align}
\label{eq:djc}
d_j^c = \min_{\substack{\bs \equiv (\bs[0], \ldots, \bs[j])\\ \bs[0] \neq 0}} \mathrm{wt}^c(\bx[0],\dots,\bx[j]),
\end{align}
where recall that each channel symbol $\bx[i]$ has $\bar{n}$ sub-symbols, i.e., $\bx[i]=(x_0[i],\dots,x_{\bar{n}-1}[i])$ and $\mathrm{wt}^c(\bv)$ counts the number of non-zero \textbf{sub-symbols} in the codeword $\bv$.

It is well-known that for any $(\bar{n},\bar{k},\bar{m})$ convolutional code $d_j^c \leq (\bar{n}-\bar{k})(j+1)+1$ for all $j \ge 0$. A special class of convolutional codes -- systematic Strongly-MDS codes -- satisfy this bound with equality for $j=\{0,\dots,\bar{m}\}$~\cite[Corollary 2.5]{strongly-mds}. 

The proof of property L1 follows by using an argument similar to that in the proof of Fact~\ref{fact:dT}. We will omit it as the argument is completely analogous.


\begin{figure}[t]
	\centering
	\resizebox{\columnwidth}{!}{
	\begin{tikzpicture}[node distance=0mm]
		\node[usym, minimum height=10mm]  (x100) {$0$};
		\node[usym, minimum height=10mm, right = of x100]     (x101) {$$};
		\node[esym, minimum height=10mm, right = of x101]     (x102) {$c$};
		\node[esym, minimum height=10mm, right = of x102]     (x103) {$$};
		\node[esym, minimum height=10mm, right = of x103]     (x104) {$n-1$};
		\node[esym, minimum height=10mm, right = of x104]     (x105) {$n$};
		\node[esym, minimum height=10mm, right = of x105]     (x106) {$$};
		\node[esym, minimum height=10mm, right = of x106]     (x107) {$$};
		\node[esym, minimum height=10mm, right = of x107]     (x108) {$$};
		\node[esym, minimum height=10mm, right = of x108]     (x109) {$2n-1$};
		\node[esym, minimum height=10mm, right = of x109]     (x110) {$2n$};
		\node[esym, minimum height=10mm, right = of x110]     (x111) {$$};
		\node[esym, minimum height=10mm, right = of x111]     (x112) {$$};
		\node[esym, minimum height=10mm, right = of x112]     (x113) {$$};
		\node[esym, minimum height=10mm, right = of x113]     (x114) {$3n-1$};
		\node[esym, minimum height=10mm, right = of x114]     (x115) {$$};
		\node[esym, minimum height=10mm, right = of x115]     (x116) {$$};
		\node[esym, minimum height=10mm, right = of x116]     (x117) {$$};
		\node[esym, minimum height=10mm, right = of x117]     (x118) {$c+\hat{B}-1$};
		\node[usym, minimum height=10mm, right = of x118]     (x119) {$$};
		\node[usym, minimum height=10mm, right = of x119]     (x120) {$$};
		\node[usym, minimum height=10mm, right = of x120]     (x121) {$$};
		\node[usym, minimum height=10mm, right = of x121]     (x122) {$$};
		\node[usym, minimum height=10mm, right = of x122]     (x123) {$$};
		\node[usym, minimum height=10mm, right = of x123]     (x124) {$$};
		\node[usym, minimum height=10mm, right = of x124]     (x125) {$(j-1)n$};
		\node[usym, minimum height=10mm, right = of x125]     (x126) {$$};
		\node[usym, minimum height=10mm, right = of x126]     (x127) {$$};
		\node[usym, minimum height=10mm, right = of x127]     (x128) {$$};
		\node[usym, minimum height=10mm, right = of x128]     (x129) {$jn-1$};
		\node[usym, minimum height=10mm, right = of x129]     (x130) {$jn$};
		\node[usym, minimum height=10mm, right = of x130]     (x131) {$$};
		\node[usym, minimum height=10mm, right = of x131]     (x132) {$$};
		\node[usym, minimum height=10mm, right = of x132]     (x133) {$$};
		\node[usym, minimum height=10mm, right = of x133]     (x134) {$(j+1)n-1$};
		\node      [right = of x134]     (x1end) {$\cdots$};
		\braceup{x100}{x134}{+1mm}{{$\cW_0(j+1)$}};
		\braceup{x105}{x134}{+8mm}{{$\cW_1(j)$}};
		\braceup{x110}{x134}{+15mm}{{$\cW_2(j-1)$}};
		
		\bracedn{x100}{x104}{-1mm}{{$\bx[0]$}};
		\bracedn{x105}{x109}{-1mm}{{$\bx[1]$}};
		\bracedn{x110}{x114}{-1mm}{{$\bx[2]$}};
		\bracedn{x125}{x129}{-1mm}{{$\bx[j-1]$}};
		\bracedn{x130}{x134}{-1mm}{{$\bx[j]$}};
		\bracedn{x102}{x118}{-8mm}{{$\hat{B}$}};
	\end{tikzpicture}}
	\caption{An erasure channel with $\hat{B}$ erasures in a burst starting at time $c$ used in proving L2 in Lemma~\ref{lem:mds-sub}. Grey and white squares resemble erased and unerased sub-symbols respectively.}
	\label{fig:mdp-subsymbols}
\end{figure}
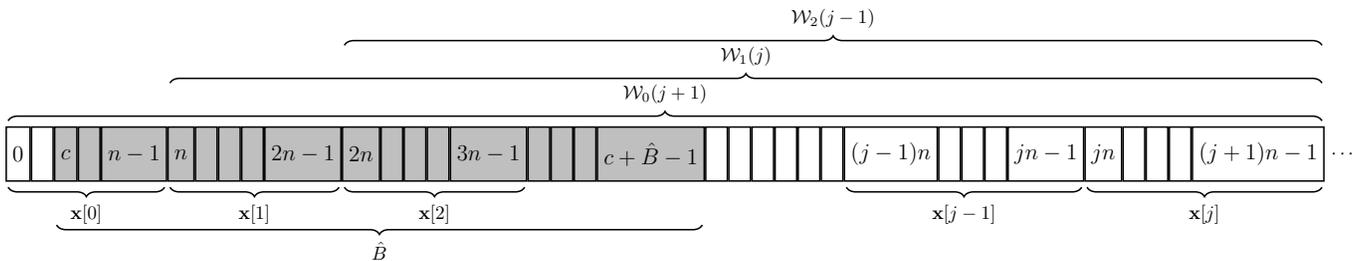

To establish L2, we use the notation $\cW_i(l)$ to denote a window of length $l \cdot \bar{n}$ starting at time $i \cdot \bar{n}$, i.e., $\cW_i(l) \defeq [i \cdot \bar{n},(i+l)\bar{n}-1]$ (see Fig.~\ref{fig:mdp-subsymbols}). We show that the entire erasure burst can be recovered through the following steps.
\begin{itemize}
\item In the window $\cW_0(j+1) = [0,(j+1)\bar{n}-1]$, the channel introduces $\hat{B} \le (\bar{n}-\bar{k})(j+1)$ erasures. Hence, we use L1 to recover $\bs[0] = (s_0[0],\dots,s_{\bar{k}-1}[0])$ at time $(j+1)n-1$ among which only the last $\bar{k}-c$ sub-symbols are erased. At this point we can also compute the $\bar{n}-\bar{k}$ sub-symbols of $\bp[0] = (p_0[0],\dots,p_{\bar{n}-\bar{k}-1}[0])$. Thus all the sub-symbols until time $t=\bar{n}-1$ have now been recovered by the decoder.
\item The next window $\cW_1(j) = [\bar{n},(j+1)\bar{n}-1]$ has $\hat{B} - (\bar{n}-c) < j(\bar{n}-\bar{k})$ erasures since $c < \bar{k}$. Hence, L1 can be used to recover $\bs[1]$ at time $(j+1)\bar{n}-1$ and $\bp[1]$ can be computed consequently. 
\item Similarly, $\cW_2(j-1) = [2\bar{n},(j+1)\bar{n}-1]$ has $\hat{B}- (\bar{n}-c) - \bar{n} < (\bar{n}-\bar{k})(j-1)$ erasures which implies the recovery of $\bs[2]$ at time $(j+1)\bar{n}-1$.
\item Repeating the previous step for $\cW_i(j-i+1) = [i \cdot \bar{n},(j+1)\bar{n}-1]$ and $i \cdot \bar{n} \le c+\hat{B}-1$, one can recover all erased symbols in the erasure burst at time $(j+1)\bar{n}-1$.
\end{itemize}

The proof of L2 is thus complete. The claim in L3 is a generalization of L2, as it permits the erasure pattern to have both burst and isolated erasures, but only guarantees the recovery of the burst erasure. To establish L3 we can proceed in a similar fashion as above and stop when the recovery of the erasure burst is complete. 

\section{Decoding Analysis of MiDAS code with MDS constituent codes}
\label{app:MiDAS-MDS}

In the decoding analysis it is sufficient to show that each source symbol $\bs[i]$ can be recovered at time ${t=i+\Te}$ if there is either an erasure burst of length $B$ or up to $N$ isolated erasures in the interval $[i,i+\Te]$.

\subsection{Burst Erasure}
First consider the case when a burst erasure spans $[i,i+B-1]$. Following this burst, we are guaranteed that for the $\cC(N,B,W)$ channel, there are no erasures in the interval $[i+B,i+ \Te+B-1]$. We argue that the decoder can first recover
$\bv[i],\ldots, \bv[i+B-1]$ simultaneously by time $t=i+\Te-1$ and then recover $\bu[i]$ at time $t=i+\Te$ by computing $\bp^v[i+\Te]$ and then $\bu[i] = \bq[i+\Te]-\bp^v[i+\Te]$. To show the recovery of $\bv[i],\ldots, \bv[i+B-1]$, note that there are no erasures in the interval spanning $[i+B, i+\Te-1]$ and the interfering $\bu[\cdot]$ symbols in $\bq[t] = \bu[t-\Te] + \bp^v[t]$ can be subtracted out to recover $\bp^v[t]$. The diagonal codewords $\left\{\bc_j^v[r]\right\}$ spanning $\bv[i],\ldots, \bv[i+B-1]$ start at $r \in \{ i-(\Te-B)+1,\dots,i+B-1 \}$. Each such codeword belongs to a $(\Te,\Te-B)$ MDS code. Hence, if no more than $B$ erasures take place in each codeword, the erased symbols can be recovered. However, we still need to take the delay into account. We first note that the $\bv[\cdot]$ symbols in the interval $[i,i+B-1]$ are erased. Also, the $\bq[\cdot]$ symbols in the interval $[i+\Te,i+\Te+B-1]$ combine $\bu[\cdot]$ which are erased and thus the corresponding $\bp[\cdot]$ symbols must also be treated as erased. We split the diagonal codewords of interest into two groups,
\begin{itemize}
\item $r \in \{i-(\Te-B)+1,\dots,i\}$: These codewords end before time $i+\Te$ where there are only $B$ erased columns in the interval $[i,i+B-1]$.
\item $r \in \{i+1,\dots,i+B-1\}$: Each of these codewords has ${\Te-B}$ sub-symbols in the interval $[i+B,i+\Te-1]$ which are not erased. Since the length of each codeword is $\Te$, then the number of erasures are at most $\Te-(\Te-B) = B$. 
\end{itemize}
We note that all the considered diagonal codewords $\left\{\bc_j^v[r]\right\}$ for $r \in \{ i-(\Te-B)+1,\dots,i+B-1 \}$ end before time $i+\Te+B$. Also, the $\bp[\cdot]$ parities in the interval $[i+\Te,i+\Te+B-1]$ cannot be used as discussed earlier. Thus it follows that the corresponding $\bv[\cdot]$ symbols are recovered by $i+\Te-1$.

\subsection{Isolated Erasures}
Next we show that when there are $N$ erasures in arbitrary locations in the interval $[i, i+\Te]$, then $\bu[i]$ is guaranteed to be recovered by time $t=i+\Te$, and $\bv[i]$ is guaranteed to be recovered by time $t=i+\Te-1$. For the recovery of $\bu[i]$ we note that the codewords $\left\{\bc_j^u[r]\right\}$ that include $\bu[i]$ start at $r \in \{ i-(\Te-N),\dots,i \}$. Since each $\bc_j^u[r]$ is a $(\Te+1, \Te-N+1)$ MDS code, and there are no more than $N$ erasures on each such sequence, it follows that all the erased symbols are guaranteed to be recovered by time $i+\Te$. The recovery of $\bu[i]$ by time $t=i+\Te$ now follows. For recovering $\bv[i]$, we consider the non-erased parity-check symbols $\bp^v[t]$ for $t\in [i, i+\Te-1]$, which can be obtained by cancelling the interfering $\bu[t-\Te]$ symbols from $\bq[t]$ as discussed in the case of burst erasure above. Notice that the diagonal codewords $\left\{\bc_j^v[r]\right\}$ spanning $\bv[i]$ start at $r \in \{ i-(\Te-B)+1,\dots,i \}$ and terminate by time $i+\Te-1$. It follows that each such sequence has no more than $N$ erasures and hence all the erased $\bv[i]$ symbols are recovered by time $t=i+\Te-1$.

\section{Proof of Lemma~\ref{lem:mdp_recovery_mismatch}}
\label{app:mdp_recovery_mismatch}

\begin{figure}
\footnotesize
\begin{align*}
j = 1 &\left[
\begin{array}{g|g|g|g|g|g|g|g|g|g|g|g|g}
  \bu[i,1] & \bu[i,2] & \cdots & \bu[i,r] & \begin{array}{c}	\bu[i,r+1] \\ \bv[i,1] \end{array} & \bv[i,2] & \bv[i,3] & \cdots & \bv[i,M-2r-1] & \begin{array}{c} \bq[i,r+1] \\ \bv[i,M-2r] \end{array} & \bq[i,r] & \cdots & \bq[i,1]\\
\end{array} \right]
\\
j = 2 &\left[
\begin{array}{c|g|g|g|g|g|g|g|g|g|g|g|g}
  \bu[i,1] & \bu[i,2] & \cdots & \bu[i,r] & \begin{array}{c}	\bu[i,r+1] \\ \bv[i,1] \end{array} & \bv[i,2] & \bv[i,3] & \cdots & \bv[i,M-2r-1] & \begin{array}{c} \bq[i,r+1] \\ \bv[i,M-2r] \end{array} & \bq[i,r] & \cdots & \bq[i,1]
\end{array} \right]
\\ \cline{1-2}
j = r+1 &\left[
\begin{array}{c|c|c|c|g|g|g|g|g|g|g|g|g}
  \bu[i,1] & \bu[i,2] & \cdots & \bu[i,r] & \begin{array}{c}	\bu[i,r+1] \\ \bv[i,1] \end{array} & \bv[i,2] & \bv[i,3] & \cdots & \bv[i,M-2r-1] & \begin{array}{c} \bq[i,r+1] \\ \bv[i,M-2r] \end{array} & \bq[i,r] & \cdots & \bq[i,1]
\end{array} \right]
\\
j = r+2 &\left[
\begin{array}{c|c|c|c|c|g|g|g|g|g|g|g|g}
  \bu[i,1] & \bu[i,2] & \cdots & \bu[i,r] & \begin{array}{c}	\bu[i,r+1] \\ \bv[i,1] \end{array} & \bv[i,2] & \bv[i,3] & \cdots & \bv[i,M-2r-1] & \begin{array}{c} \bq[i,r+1] \\ \bv[i,M-2r] \end{array} & \bq[i,r] & \cdots & \bq[i,1]
\end{array} \right]
\\
j = r+3 &\left[
\begin{array}{c|c|c|c|c|c|g|g|g|g|g|g|g}
  \bu[i,1] & \bu[i,2] & \cdots & \bu[i,r] & \begin{array}{c}	\bu[i,r+1] \\ \bv[i,1] \end{array} & \bv[i,2] & \bv[i,3] & \cdots & \bv[i,M-2r-1] & \begin{array}{c} \bq[i,r+1] \\ \bv[i,M-2r] \end{array} & \bq[i,r] & \cdots & \bq[i,1]
\end{array} \right]
\\ \cline{1-2}
j = M-r &\left[
\begin{array}{c|c|c|c|c|c|c|c|c|g|g|g|g}
  \bu[i,1] & \bu[i,2] & \cdots & \bu[i,r] & \begin{array}{c}	\bu[i,r+1] \\ \bv[i,1] \end{array} & \bv[i,2] & \bv[i,3] & \cdots & \bv[i,M-2r-1] & \begin{array}{c} \bq[i,r+1] \\ \bv[i,M-2r] \end{array} & \bq[i,r] & \cdots & \bq[i,1]
\end{array} \right]
\\
j = M-r+1 &\left[
\begin{array}{c|c|c|c|c|c|c|c|c|c|g|g|g}
  \bu[i,1] & \bu[i,2] & \cdots & \bu[i,r] & \begin{array}{c}	\bu[i,r+1] \\ \bv[i,1] \end{array} & \bv[i,2] & \bv[i,3] & \cdots & \bv[i,M-2r-1] & \begin{array}{c} \bq[i,r+1] \\ \bv[i,M-2r] \end{array} & \bq[i,r] & \cdots & \bq[i,1]
\end{array} \right]
\\ \cline{1-2}
\end{align*}
\normalsize
\caption{Different erasure patterns considered in the analysis of the decoder. The index $j$ at the left of each row, indicates the starting location of each burst in macro-block $i$. The shaded blocks shows the symbols that are erased. }
\label{eq:Bj}
\end{figure}

We need to show that the total erased sub-symbols $(\bv_{\mrm{vec}}[\cdot], \bp_{\mrm{vec}}[\cdot])$ between macro-packets $i$ to macro-packet $i+T$ i.e., in the following sequence,
\begin{multline}
\big\{\bv_\mrm{vec}[t],\bp_\mrm{vec}[t]\big\}_{i\le t \le i+T} =
\bigg(v_0[i],\dots,v_{k^v-1}[i],p_0[i],\dots, p_{k^u-1}[i],v_0[i+1],\dots,v_{k^v-1}[i+1],p_0[i+1],\dots, \\ p_{k^u-1}[i+1],\dots,v_0[i+T],\dots,v_{k^v-1}[i+T],p_0[i+T],\dots,p_{k^u-1}[i+T]\bigg),
\end{multline} after the cancellation of $\bu[\cdot]$ symbols, is $k^u(T+1)$. We start by considering that the burst begins at $j=1$ and subsequently consider other cases in Fig.~\ref{eq:Bj}. For the case when $j=1$ we consider two cases.

\begin{itemize}
\item $B' > \frac{b}{T+b}M$:
We first show that the total number of $\bv_{\mrm{vec}}[\cdot]$ and $\bp_{\mrm{vec}}[\cdot]$ sub-symbols erased due to the burst in the macro-packets $i, i+1,\ldots, i+b$ equals $k^u T = B\cdot T$. Furthermore in macro-packet $i+T$, the parity-checks $\bp_\mrm{vec}[i+T]$ combine with $\bu_\mrm{vec}[i]$ which are also erased. Hence these sub-symbols contribute to additional $k^u$ erasures leading to a total of $k^u(T+1)$ erased sub-symbols. 

Note that the erasure burst spans the entire macro-packets $\bX[i,:],\dots,\bX[i+b-1,:]$ as well as $\bx[i+b.1],\dots,\bx[i+b,B']$. The total number of sub-symbols in $\bv_{\mrm{vec}}[t]$ and $\bp_{\mrm{vec}}[t]$ in each macro-packet is $ k^v+k^u= M(T+b+1)-B$. In the $b$-th macro-packet we only have the first $B'$ columns erased. Out of these the first $k^u$ sub-symbols are from the $\bu_\mrm{vec}[\cdot]$ sub-symbols whereas the remaining
 $B' n - k^u = B'(T+b+1) - B$ sub-symbols come from $\bv_{\mrm{vec}}[i+b]$ and $\bp_{\mrm{vec}}[i+b]$ . 
 It can be easily verified that $B'(T+b+1)-B \ge 0$.
 Hence the total number of erased sub-symbols of $\bv_{\mrm{vec}}[t]$ and $\bp_{\mrm{vec}}[t]$ is 
 \begin{align}
b(k^u+k^v) + B' n - k^u &=b(M(T+b+1)-B)+B'(T+b+1) - k^u \nonumber \\
&= B(T+b+1) -bB-k^u \nonumber \\
&= B(T+1) - k^u = Tk^u
 \end{align}
 where we use the fact that $k^u=B$ in our code construction and $B = bM+B'$. 

\item $B' \leq \frac{b}{T+b}M$:
Again the macro-packets $\bX[i,:],\dots,\bX[i+b-1,:]$ are completely erased and each contributes to $b (k^v+k^u)= b(M(T+b)-Mb)$ erasures. In the $\bX[i+b,:]$ only the sub-symbols in $\bu_\mrm{vec}[i+b]$ are erased as it can be easily verified that $B'n \le k^u$. Finally as in the previous case all the $\bv_\mrm{vec}[i+T]$ sub-symbols in macro-packet $i+T$ that combine with $\bu_\mrm{vec}[i]$ must be considered erased. Thus the total number of erased sub-symbols
is $b(M(T+b)-Mb)+ k^u = bM T + k^u = k^u (T+1)$. 
\end{itemize}

To establish the claim for $j=2,3,\ldots, M-r$ it suffices to show the following lemma

\begin{lemma}
Let $N_j$ denote the total number of erased sub-symbols in $\{\bv_\mrm{vec}[t], \bp_\mrm{vec}[t]\}$ after the cancellation of non-erased $\bu_\mrm{vec}[\cdot]$ sub-symbols when the erasure burst begins at $\bx[i,j]$. Then we have that $N_j \le N_{j-1}$ for each $j=2,3,\ldots, M-r$.
\label{lem:worst-case}
\end{lemma}
Lemma~\ref{lem:worst-case} establishes that the worst case erasure sequence is the one that begins at $j=1$. 
Since we have already established that the total number of erasures in $\{\bv_\mrm{vec}[t], \bp_\mrm{vec}[t]\}$ in this case does not exceed $k^u(T+1)$, this will complete our claim. 

To establish Lemma~\ref{lem:worst-case}, we note that going from the burst pattern that starts at $\bx[i,j]$ to the pattern that start at $\bx[i,j+1]$ results in one extra erased channel packet at the end. Also, it results in revealing the first channel packet which is $\bx[i,j]$. We assume (as a worst case) that the extra erased channel packet at the end contributes to $n$ additional erased sub-symbols of either $\bv_{\mrm{vec}}[\cdot]$ or $\bp_{\mrm{vec}}[\cdot]$. We consider the effect of revealing the channel packet $\bx[i,j]$ and show that it always compensates exactly $n$ new unerased symbols
of either $\bv_{\mrm{vec}}[\cdot]$ or $\bp_{\mrm{vec}}[\cdot]$. Thus we do not increase the total number of erased symbols in such a transition.

Recall that $\bx[i,j]$ can be one of the following (cf.~Fig~\ref{eq:Bj}),
\begin{align}
\bx[i,j] = \begin{cases}
\bu[i,j] & j = \{1,\dots,r\} \\
\left[\begin{array}{c} \bu[i,r+1] \\ \bv[i,1] \end{array} \right] & j=r+1 \\
\bv[i,j-r] & j = \{r+2,\dots,M-r-1\} \\
\left[\begin{array}{c} \bq[i,r+1] \\ \bv[i,j-r] \end{array} \right] & j = M-r \\
\bq[i,M-j+1] & j = \{M-r+1,\dots,M\}. \\
\end{cases}
\end{align}

\begin{itemize}
\item $j=\{1,\cdots,r\}$:
In the case under consideration, the revealed $\bx[i,j]$ is always $\bu[i,j]$ . It can be subtracted from $\bq[i+T,j]$ to recover $\bp[i+T,j] \in \bp_{\mrm{vec}}[\cdot]$ having $n$ sub-symbols. Thus, it compensates for the $n$ extra erased sub-symbols.

\item $j = r+1$:
The $r'$ sub-symbols of $\bu[i,r+1]$ helps in recovering the $r'$ sub-symbols of $\bp[i+T,r+1] \in \bp_{\mrm{vec}}[\cdot]$. This together with the revealed $n-r'$ sub-symbols of $\bv[i,1] \in \bv_{\mrm{vec}}[\cdot]$ compensates for the $n$ extra erasures.

\item $j = \{r+2,\dots,M-r-1\}$:
In this case, the revealed channel packet is $\bx[i,j] = \bv[i,j-r] \in \bv_{\mrm{vec}}[\cdot]$ and has $n$ sub-symbols which are now available at the decoder.

\item $j = M-r$:
As shown in Fig.~\ref{eq:Bj}, the decoder can subtract $\bu[i-T,r+1]$ from $\bq[i,r+1]$ to recover the $r'$ sub-symbols $\bp[i,r+1] \in \bp_{\mrm{vec}}[\cdot]$. This together with the $n-r'$ sub-symbols of $\bv[i,j-r] \in \bv_{\mrm{vec}}[\cdot]$ add up to $n$ sub-symbols and the claim follows.
\end{itemize}
This establishes Lemma~\ref{lem:worst-case} and in turn the proof of Lemma~\ref{lem:mdp_recovery_mismatch} is complete.

\bibliographystyle{unsrt}
\bibliography{sm}

\end{document}